\documentclass[letterpaper, 11pt]{article}
\usepackage{latexsym,epsfig,secdot,url,amsbsy,amsmath,exscale,color,fancyhdr,booktabs, ulem}
\usepackage[psamsfonts]{amsfonts,euscript}
\usepackage{amsmath,amssymb,amsfonts,amsthm}
\usepackage{algorithm, algpseudocode,scalerel}
\usepackage[center,bf]{caption}
\usepackage{color}
\usepackage{bm,bbm}
\usepackage{graphicx,enumerate}
\usepackage{rotating}
\usepackage{subfigure}
\usepackage{multirow}
\usepackage{subcaption}
\usepackage{enumitem}
\usepackage{diagbox}
\usepackage{tikz}
\usetikzlibrary{shapes.geometric}

\usepackage{hyperref}

\def\double{\par\baselineskip=24pt}
\def\single{\par\baselineskip=12pt}

{\end{list}}

\def\E{{\rm E}}
\def\Var{{\rm Var}}

\def\b1{{\bf 1}}

\def\blot{\quad {$\vcenter{\vbox{\hrule height.4pt
             \hbox{\vrule width.4pt height.9ex \kern.9ex \vrule width.4pt}
             \hrule height.4pt}}$}}

\usepackage{endnotes}
\let\footnote=\endnote

\usepackage{natbib}
 \bibpunct[, ]{(}{)}{,}{a}{}{,}%
 \def\bibfont{\small}%
 \def\bibsep{\smallskipamount}%
 \def\bibhang{24pt}%

\newtheorem{theorem}{Theorem}
\newtheorem{lemma}{Lemma}

\newtheorem{assumption}{Assumption}

\newtheorem{remark}{Remark}
\newcommand{\comment}[1]{}

\newtheorem{innercustomgeneric}{\customgenericname}
\providecommand{\customgenericname}{}
\newcommand{\newcustomtheorem}[2]{%
  \newenvironment{#1}[1]
  {%
   \renewcommand\customgenericname{#2}%
   \renewcommand\theinnercustomgeneric{##1}%
   \innercustomgeneric
  }
  {\endinnercustomgeneric}
}

\newcustomtheorem{customthm}{Theorem}
\newcustomtheorem{customlemma}{Lemma}

\def\E{\mathbf{E}}
\def\Pr{\mathbf{P}}
\def\Var{\mathbf{Var}}

\begin{document}
\double
\title{Indifference-Zone Relaxation Procedures for Finding Feasible Systems}
\author{Yuwei Zhou \\
	Kelley School of Business    \\
	Indiana University  \\
	\vspace{0.3cm}
   Bloomington, IN 47405, USA  \\
	Sigr\'un Andrad\'ottir  \\
	Seong-Hee Kim\\
	H. Milton Stewart School of Industrial and Systems Engineering \\
	Georgia Institute of Technology \\
    \vspace{0.3cm}	
    Atlanta, GA 30332, USA  \\
	Chuljin Park\footnote{Corresponding author. Email: parkcj@hanyang.ac.kr. Telephone: +82-2-2220-0476}\\
	Department of Industrial Engineering \\
	Hanyang University\\
	Seoul, 04763, South Korea}
\maketitle
\vspace{-24pt}
\begin{abstract}
We consider the problem of finding feasible systems with respect to stochastic constraints when system performance is evaluated through simulation. Our objective is to solve this problem with high computational efficiency and statistical validity. Existing indifference-zone (IZ) procedures introduce a fixed tolerance level, which denotes how much deviation the decision-maker is willing to accept from the threshold in the constraint. These procedures are developed under the assumption that all systems' performance measures are exactly the tolerance level away from the threshold, leading to unnecessary simulations. In contrast, IZ-free procedures, which eliminate the tolerance level, perform well when systems' performance measures are far from the threshold. However, they may significantly underperform compared to IZ procedures when systems' performance measures are close to the threshold. To address these challenges, we propose the Indifference-Zone Relaxation (${\cal IZR}$) procedure. ${\cal IZR}$ introduces a set of relaxed tolerance levels and utilizes two subroutines for each level: one to identify systems that are clearly feasible and the other to exclude those that are clearly infeasible. We also develop the ${\cal IZR}$ procedure with estimation (${\cal IZE}$), which introduces two relaxed tolerance levels for each system and constraint: one matching the original tolerance level and the other based on an estimate of the system's performance measure. By employing different tolerance levels, these procedures facilitate early feasibility determination with statistical validity.  We prove that ${\cal IZR}$ and ${\cal IZE}$ determine system feasibility with the desired probability and show through experiments that they significantly reduce the number of observations required compared to an existing procedure.


\noindent {\em Subject classification:} Simulation, Ranking and Selection, Fully Sequential Procedure, Feasibility Check, Stochastic Constraints

\end{abstract}


\section{Introduction}

Ranking and selection (R\&S) procedures have been developed and used for finding a system with the best performance among finitely many systems when the performance is evaluated by stochastic simulation. General approaches to R\&S fall into several categories, including fully sequential indifference-zone (IZ) procedures \citep{KimNelsonHB}, optimal computing budget allocation (OCBA) procedures \citep{ChenOCBA}, and Bayesian procedures \citep{ChickBayes}.

Constrained R\&S involves selecting the best system with respect to a primary performance measure while satisfying stochastic constraints on secondary performance measures. \cite{ak} and \cite{bk:constraint} develop fully sequential IZ feasibility check procedures (FCPs) to find a set of feasible systems with respect to a constraint and multiple constraints, respectively. These FCPs are employed as subroutines to select a system with the best performance among systems satisfying constraints on one or more secondary performance measures. \cite{ak} and \cite{Healey:HAK:dormancy} propose statistically valid procedures that select the best feasible system with at least a pre-specified probability under a stochastic constraint. \cite{Healey:HAK:TOMACS} extend the procedures to handle multiple stochastic constraints and to find the best system among the feasible systems. \cite{LeeOCBACO} provide an OCBA procedure that allocates a finite simulation budget to maximize the probability of correctly selecting the best system satisfying multiple stochastic constraints. Furthermore, \cite{huterragu}, \cite{SCORE}, and \cite{Gao_TACON2015} enhance a feasibility determination procedure for multiple stochastic constraints considering OCBA and large deviation theory. \cite{BayesFD} handle stochastic constraints using Bayesian concepts. In this paper, we focus on the fully sequential IZ procedures in the presence of stochastic constraints.

Fully sequential IZ FCPs introduce one tolerance level for each constraint, and this tolerance level specifies how much the decision-maker is willing to be off from a constant threshold for checking the feasibility of the systems with respect to the constraint. This is the least absolute difference in the performance measure and the constant threshold that the decision-maker wants to detect. In practice, the true expected performance measure is unknown. When a system's performance measure is very close to the threshold, distinguishing this small difference requires a huge number of replications to determine feasibility. The tolerance level helps avoid such situations. Nevertheless, when the tolerance level is too small relative to the difference between the true expected performance measure and the threshold, computational costs for feasibility checks may be unnecessarily high. Therefore, a small value of the tolerance level may aggravate the computational burden when the number of systems is large and the range of performance measure values is also large. \cite{Lee2018} propose the adaptive feasibility check procedure that uses existing FCPs \citep{bk:constraint} with different thresholds as its subroutines and introduces a decreasing sequence of tolerance levels within the subroutines. The procedure is designed to self-adjust the tolerance levels for each system and constraint by introducing two independent FCPs with different thresholds and the same tolerance level. Although the procedure shows strong empirical performance in numerical and practical examples, it is found that the procedure may terminate without achieving the desired statistical guarantee when the expected performance measure is very close to the threshold.

Recently, fully sequential procedures that eliminate the IZ have been developed for traditional R\&S problems \citep{FHN_OR,Wang2024} and for constrained R\&S problems \citep{IZF_CRS}. Specifically, \cite{FHN_OR} proposed an IZ-free procedure for selecting the best system, using the Law of Iterated Logarithm \citep{Durrettbook} to maintain statistical validity without the IZ. Meanwhile, \cite{IZF_CRS} adapted the IZ-free procedures from \cite{FHN_OR} to identify feasible systems without a tolerance level and extended their work to address R\&S problems in the presence of stochastic constraints. Additionally, \cite{Fanetal2025} examined the impact of the first-stage sampling size and proposed improved versions of the IZ-free procedures for selecting the best system. On a different note, \cite{Wang2024} developed IZ-flexible procedures for selecting the best system, employing the Sequential Probability Ratio Test \citep{Wald1945, WW1948} to ensure statistical validity with or without considering the IZ. These procedures tend to outperform fully sequential IZ procedures, particularly when the difference between a system's performance measure and a threshold (or between the performance measures of two different systems) is significantly large. Nevertheless, when that difference is small, the IZ-free procedures require a considerably larger number of observations and tend to underperform relative to IZ procedures. It is important to note that the total number of observations needed by a procedure can be heavily affected by this latter scenario.

In this paper, we first propose a new FCP, referred to as the Indifference-Zone Relaxation (${\cal IZR}$) procedure, which (i) introduces a set of relaxed tolerance levels for each constraint and (ii) simultaneously runs two subroutines with different thresholds and the relaxed tolerance levels. In particular, one subroutine of ${\cal IZR}$ mainly identifies clearly infeasible systems, and the other subroutine distinguishes clearly feasible systems. Then, the algorithm makes the feasibility decision for each system and constraint when the two subroutines provide the same feasibility decision (i.e., either feasible or infeasible) for the same tolerance level. As a result, the feasibility decisions of different systems and constraints may be made with different tolerance levels and this leads to saving computational costs while keeping the statistical guarantee with respect to the original tolerance level.

We then analyze the impact of the number and magnitude of relaxed tolerance levels and introduce an enhanced FCP, referred to as the ${\cal IZR}$ procedure with Estimation (${ \cal IZE }$). The ${\cal IZE}$ procedure preserves the overall structure of ${\cal IZR}$ but simplifies the design by using only two tolerance levels: one equal to the original tolerance level, and the other determined adaptively for each system and constraint based on an estimate of the difference between the system mean and constraint threshold. Our numerical results reveal that the choice of relaxed tolerance levels (either specified by the user as in ${\cal IZR}$ or estimated as in ${\cal IZE}$) has a significant impact on performance. Specifically, when system means are well spread out, employing a larger relaxed tolerance level facilitates the early elimination of clearly feasible or infeasible systems, thereby improving efficiency. In contrast, when system means are densely clustered, increasing the tolerance level provides limited benefit. This highlights the importance of accurate tolerance level estimation in the effectiveness of ${\cal IZE}$.
Overall, both ${\cal IZR}$ and ${\cal IZE}$ require significantly fewer observations than the existing FCP when the difference between a system’s performance measure and the corresponding threshold is substantially greater than the original tolerance, with ${\cal IZE}$ typically outperforming ${\cal IZR}$ in such situations. 
When this difference is small or near the original tolerance level, both procedures require a similar number of observations as the existing FCP. Importantly, both ${\cal IZR}$ and ${\cal IZE}$ maintain statistical validity with respect to the original tolerance level.


The rest of the paper is organized as follows: Section \ref{sec2} provides the background for our problem.
Sections \ref{sec3} and \ref{sec4} propose our procedures and the associated proofs of statistical validity. Section \ref{sec5} presents numerical results for our procedures and compares their performance with that of an existing procedure. Concluding remarks are provided in Section \ref{sec6}. Note that \cite{Park2024} introduce an earlier version of this work that considers a single stochastic constraint, does not include mathematical proofs, and provides limited numerical results. This preliminary version also does not include the ${\cal IZE}$ procedure.

\section{Background}
\label{sec2}
We consider $k$ systems whose $s$ performance measures can be observed through stochastic simulation. Let $\Theta $ denote the set of all systems (i.e., $\Theta = \{ 1,2,\ldots, k\}$) and $Y_{i \ell j}$ for $j=1,2,\ldots,$ denote the $j$th simulation observation associated with the $\ell$th performance measure of system $i$. For any given system $i$ and performance measure $\ell$, $y_{i \ell}$ denotes the expectation of $Y_{i \ell j}$ (i.e., $y_{i \ell}=\E[Y_{i \ell j}]$) and $\sigma_{i \ell}^2$ denotes the variance of $Y_{i \ell j}$ (i.e., $\sigma_{i \ell}^2=\Var[Y_{i \ell j}]$). Let $q_\ell$ be a given threshold for the $\ell$th constraint, $\ell = 1,2,\ldots,s$. System $i$ is defined as feasible if $y_{i \ell} \leq q_\ell$ for all $\ell$, and infeasible otherwise. Our problem is to determine the set of feasible systems and we assume that the observations satisfy the following assumption throughout the paper:

\begin{assumption} \label{assump:normal}
For each $i=1,2,\ldots,k$,\\
\begin{equation*}
\left [\begin{array}{l}
  Y_{i 1 j} \\
  Y_{i 2 j} \\
  \vdots \\
  Y_{i s j}
\end{array} \right]
\overset{iid}{\sim} MN_{s}
\left(
\left [\begin{array}{l}
  y_{i 1} \\
  y_{i 2} \\
  \vdots \\
  y_{i s}
\end{array} \right]
, \Sigma_i
\right),
\end{equation*}
where $\overset{iid}{\sim}$ denotes independent and identically distributed (iid), $MN_{s}$ denotes the $s-$dimensional multivariate normal distribution, and $\Sigma_i$ is the $s \times s$ positive deﬁnite covariance matrix of the vector $(Y_{i 1 j}, Y_{i 2 j}, \ldots, Y_{i s j})$.
\end{assumption}

Assumption~\ref{assump:normal} is reasonable, particularly when the observations of the system are either (i) averages within a replication from transient or steady-state simulations or (ii) batch means \citep{LawKeltonbook}, even though they are not normally distributed. Additionally, the observations of different systems can be statistically dependent (i.e., the systems are said to be dependent). For example, one may employ common random numbers to simulate the different systems \citep{KimNelsonTOMACS}.


To define a correct decision event, \cite{ak} introduce a tolerance level, denoted by $\epsilon_\ell$, which is a user-specified positive real number for each $\ell = 1,2,\ldots,s$. Each system $i$ falls in one of the following three sets:
\begin{align*}
D& \equiv \left\{ i \in \Theta \ |\ y_{i \ell} \leq q_{\ell} - \epsilon_{\ell},\; \mbox{for  } \ell = 1,2,\ldots,s \right\};\\
A& \equiv \left\{ i \in \Theta \ |\ y_{i \ell} < q_{\ell} + \epsilon_{\ell},\; \mbox{for  } \ell = 1,2,\ldots,s \right\} \setminus D; \mbox{ and  }\\
U& \equiv \cup_{\ell=1}^{s} \left\{ i \in \Theta \ |\ q_{\ell}+ \epsilon_{\ell} \leq y_{i\ell} \right\}.
\end{align*}

The systems in the sets $D$, $A$, and $U$ are called desirable, acceptable, and unacceptable systems, respectively. A correct decision for system $i$ is denoted by ${\rm CD}_i$ and it is defined as declaring system $i$ feasible if $i \in D$ and infeasible if $i \in U$. Any decision is considered as a correct decision for $i \in A$. Then, a correct decision for the problem is defined as the event that the correct decisions for all systems are simultaneously made (i.e., ${\rm CD} = \cap_{i=1}^k {\rm CD}_i$), and a statistically valid procedure should satisfy $\Pr({\rm CD}) \ge 1-\alpha$ where $\alpha$ is a confidence level.


Procedure ${\cal F}_B$ due to \cite{bk:constraint} is a statistically valid procedure for feasibility determination. It is based on a fully sequential R\&S procedure in \cite{Fabian} and \cite{KimNelsonTOMACS}. Specifically, the procedure keeps track of a monitoring statistic that is a cumulative sum of the difference between $Y_{i \ell j}$ and $q_\ell$. An observation is sampled at each stage of the procedure while the statistic sojourns within boundaries defining the ``continuation region." When this statistic first exits the continuation region, the feasibility decision is made for constraint $\ell$. Figure~\ref{fig:ctsregion} shows a sample path of the monitoring statistics for system $i$ and constraint $\ell$ (dashed line) and the boundaries of the continuation region (bold line). As shown in the figure, if the statistic first exits through the upper boundary, then we conclude that system $i$ is infeasible for constraint $\ell$. On the other hand, if it first exits through the lower boundary, system $i$ is declared as feasible for constraint $\ell$. When simultaneously considering $k$ systems and $s$ constraints, it is possible to establish statistical validity through the Bonferroni inequality. In this case, a system is declared as feasible when all constraints are satisfied for the system, but it is declared as infeasible whenever one of the constraints is violated.

\begin{figure}[!tb]
\begin{center}
\scalebox{0.60}{\includegraphics{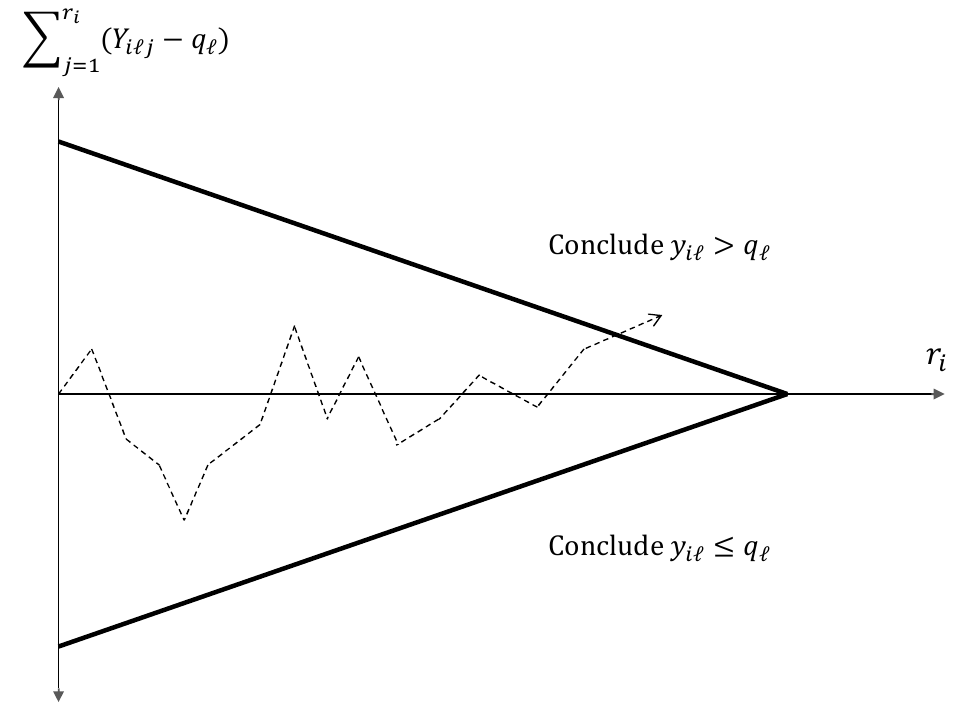}}
\caption{A triangular continuation region for system $i$ and constraint $\ell$.}\label{fig:ctsregion}
\end{center}
\end{figure}

To provide a more specific description of ${\cal F}_B$, we need a list of notations as follows:

\noindent
$n_0 \equiv$ the initial sample size for each system ($n_0 \geq 2$);\\
$r_i \equiv$ the observation counter for system $i$ ($r_i \geq n_0$);\\
$S_{i \ell}^2 \equiv$ the sample variance of $Y_{i \ell 1}, \ldots, Y_{i \ell n_0}$ for system $i$ regarding constraint $\ell$ ($i=1,2, \ldots, k$ and $\ell = 1,2,\ldots, s$);\\
$M \equiv$ the set of systems whose feasibility is not determined yet;\\
$F \equiv$ the set of systems declared as feasible; and\\
${\rm ON}_i \equiv$ the set of constraint indices whose corresponding performance measures for system $i$ have neither been deemed feasible nor infeasible.

As in \cite{Fabian}, we need the following functions to define the continuation region:
\begin{eqnarray}
R(r;v,w,z) & = & \max \left\{ 0, \frac{wz}{2cv} - \frac{v}{2c}r \right\}, \mbox{ for any } v, w, z \in \mathbb{R},\ v\ne 0,\nonumber \\   \vspace{12pt}
g_n(\eta)    & = & \sum_{\ell = 1}^c (-1)^{\ell +1} \left( 1- \frac{1}{2} {\mathbb I}(\ell = c)  \right) \times \left( 1 + \frac{2\eta (2c-\ell) \ell}{c}  \right)^{-n/2}, \label{eq:KNeta}
\end{eqnarray}
where $c, n$ are user-specified positive integers and ${\mathbb I}(\cdot)$ is an indicator function. Then, a detailed description of ${\cal F}_B$ is shown in Algorithm~\ref{alg:FB}, where $|\cdot|$ represents the cardinality of a set.

\begin{algorithm}[!htb]
	\caption{\hspace{-0.3em}: \textbf{Procedure ${\cal F}_B$}}\label{alg:FB}
    {\fontsize{11}{24}\selectfont
		\begin{algorithmic}\single
			\State [{\bf Setup}:]  Choose initial sample size $n_0 \geq 2$, confidence level $0<1-\alpha<1$, $c\in \mathbb{N}^+$, and $\Theta = \{1,2, \ldots, k\}$. Set tolerance levels $\epsilon_\ell>0$ and thresholds $q_{\ell}$ for constraint $\ell=1,2,\ldots, s$. Set $M = \{1,2,\ldots, k\}$ and $F = \emptyset$. Calculate $h_B^2= 2c\eta_B (n_0-1)$, where $\eta_B>0$ satisfies
\begin{equation} \label{eq:getafb}
g_{n_0-1}(\eta_B) = \beta_B = \left\{
\begin{array}{ll}
  \left[1-(1-\alpha)^{1/k}\right]/s, & \mbox{ if systems are independent};\\
  \alpha/(ks), & \mbox{ otherwise}.
\end{array}\right.
\end{equation}
	\State [{\bf Initialization}:]
			\For{$i \in \Theta$}
            \begin{itemize}
				\item[] Obtain $n_0$ observations $Y_{i \ell 1}, Y_{i \ell 2}, \ldots, Y_{i \ell n_0}$, compute $S_{i \ell}^2$, and set $Z_{i\ell}=0$ for $\ell=1,2,\ldots, s$.
				\item[] Set $r_i=n_0$ and ${\rm ON}_i = \{1,2,\ldots, s\}$.
			\end{itemize}
            \EndFor
			\State [{\bf Feasibility Check}:]
            Set $M^{old} = M$.
                        \For{$i \in M^{old}$}
			\For{$\ell \in {\rm ON}_i$}
			\State If $\sum_{j=1}^{r_i} (Y_{i \ell j}-q_{\ell} ) \leq -R(r_i; \epsilon_{\ell}, h_B^2, S_{i \ell}^2)$, set $Z_{i \ell} = 1$ and ${\rm ON}_i = {\rm ON}_i \setminus \{\ell \}$.
			\State If $\sum_{j=1}^{r_i} (Y_{i \ell j}-q_{\ell} ) \geq R(r_i; \epsilon_{\ell}, h_B^2, S_{i \ell}^2)$, eliminate $i$ from $M$ and exit the current \textbf{for} loop.
			\EndFor\\
$\quad\quad$ If $Z_{i\ell}=1$ for all $\ell$, move $i$ from $M$ to $F$.
\EndFor
			\State[{\bf Stopping Condition}:]
			\begin{itemize}
				\item[] If $M = \emptyset$, return $F$. Otherwise, {}set $r_i =r_i +1$ for all $i \in M$, take one additional observation $Y_{i\ell r_i}$ for all $i \in M$ and all $\ell \in {\rm ON}_i$, and then go to [{\bf Feasibility Check}].
			\end{itemize}
		\end{algorithmic}
	}
\end{algorithm}

{\bf Remark 1.} To implement Procedure ${\cal F}_B$, \cite{bk:constraint} recommend the choice of $c=1$. In this case, $g_n(\eta) = \frac{1}{2} \left(1+ 2\eta \right)^{-n/2}$ and the solution to equation~(\ref{eq:getafb}) is $\eta_B = \frac{1}{2}\left[(2\beta_B)^{-2/(n_0-1)}-1\right]$.


\section{Indifference-Zone Relaxation Procedure}
\label{sec3}
In this section, we discuss a new procedure, referred to as the Indifference-Zone Relaxation (${\cal IZR}$) procedure, in Section~\ref{sec3_1}, prove the statistical validity of our proposed procedure in Section~\ref{sec3_SG}, and analytically explain and empirically demonstrate the impact of key parameter values of ${\cal IZR}$ in Sections~\ref{sec3_PA} and \ref{sec3_ES}, respectively.

\subsection{Procedure ${\cal IZR}$}
\label{sec3_1}

In this section, we describe our new ${\cal IZR}$ procedure. For each constraint $\ell$, the ${\cal IZR}$ procedure introduces $T_\ell$ relaxed tolerance levels, denoted by $\epsilon_\ell^{(\tau)}$ for $\tau = 1, 2, \ldots, T_\ell$, satisfying $\epsilon_\ell^{(1)}> \cdots > \epsilon_\ell^{(T_\ell)} = \epsilon_\ell$. For each relaxed tolerance level $\epsilon \in \{\epsilon_\ell^{(1)}, \epsilon_\ell^{(2)}, \ldots, \epsilon_\ell^{(T_\ell)}\}$, our procedure utilizes two subroutines: ${\cal F_U}$ and ${\cal F_D}$. Subroutine ${\cal F_U}$ is designed to quickly detect clearly ${\cal U}$ndesirable systems using threshold $q_\ell + \epsilon_\ell - \epsilon$ with tolerance level $\epsilon$ while subroutine ${\cal F_D}$ is designed to facilitate the detection of ${\cal D}$esirable systems using threshold $q_\ell - \epsilon_\ell + \epsilon$ with tolerance level $\epsilon$. Observe that $q_\ell + \epsilon_\ell - \epsilon < q_\ell < q_\ell - \epsilon_\ell + \epsilon$ for $\epsilon \in \{\epsilon_\ell^{(1)}, \epsilon_\ell^{(2)}, \ldots, \epsilon_\ell^{(T_\ell-1)}\}$ and $q_\ell + \epsilon_\ell - \epsilon = q_\ell = q_\ell - \epsilon_\ell + \epsilon$ for $\epsilon = \epsilon_\ell^{(T_\ell)}$. Subroutines ${\cal F_U}$ and ${\cal F_D}$ simultaneously check the feasibility of each system $i$ for all tolerance levels using the monitoring statistics, $(\sum_{j=1}^{r_i} Y_{i \ell j}) - r_i(q_\ell+\epsilon_\ell- \epsilon)$ and $(\sum_{j=1}^{r_i} Y_{i \ell j}) - r_i(q_\ell-\epsilon_\ell + \epsilon)$, respectively. If both ${\cal F_U}$ and ${\cal F_D}$ result in the same feasibility decision (declaring system $i$ as either feasible or infeasible regarding constraint $\ell$) with some $\epsilon \in \{\epsilon_\ell^{(1)}, \epsilon_\ell^{(2)}, \ldots, \epsilon_\ell^{(T_\ell)}\}$, then the ${\cal IZR}$ procedure returns that feasibility decision for system $i$ regarding constraint $\ell$. Subroutines ${\cal F_U}$ and ${\cal F_D}$ maintain their own sets of relaxed tolerance levels, denoted by ${\cal E}_{{\cal U}i\ell}$ and ${\cal E}_{{\cal D}i\ell}$, respectively. For each system $i$, the sets ${\cal E}_{{\cal U}i\ell}$ and ${\cal E}_{{\cal D}i\ell}$ are initialized as $\{\epsilon_\ell^{(1)}, \epsilon_\ell^{(2)}, \ldots, \epsilon_\ell^{(T_\ell)}\}$ and are updated as feasibility decisions are made. These two sets may be different because (i) $\epsilon_\ell^{(\tau)}$ needs to be removed from ${\cal E}_{{\cal U}i\ell}$ or ${\cal E}_{{\cal D}i\ell}$ once the monitoring statistics of ${\cal F_U}$ or ${\cal F_D}$ first exit the continuation region with $\epsilon = \epsilon_\ell^{(\tau)}$ which is defined by $\pm R(r_i;\epsilon, h^2, S_{i\ell}^2)$, and (ii) their exit times may be different if $\tau < T_\ell$ (i.e., $\epsilon = \epsilon_\ell^{(\tau)} > \epsilon_\ell$).

\begin{figure}[!t]
        \centering
        \begin{minipage}[t]{\textwidth}
            \centering
            \subfigure[$\epsilon = \epsilon^{(1)}_\ell = 2\epsilon_\ell$]{
            \scalebox{0.55}{\includegraphics[width=10.0in]{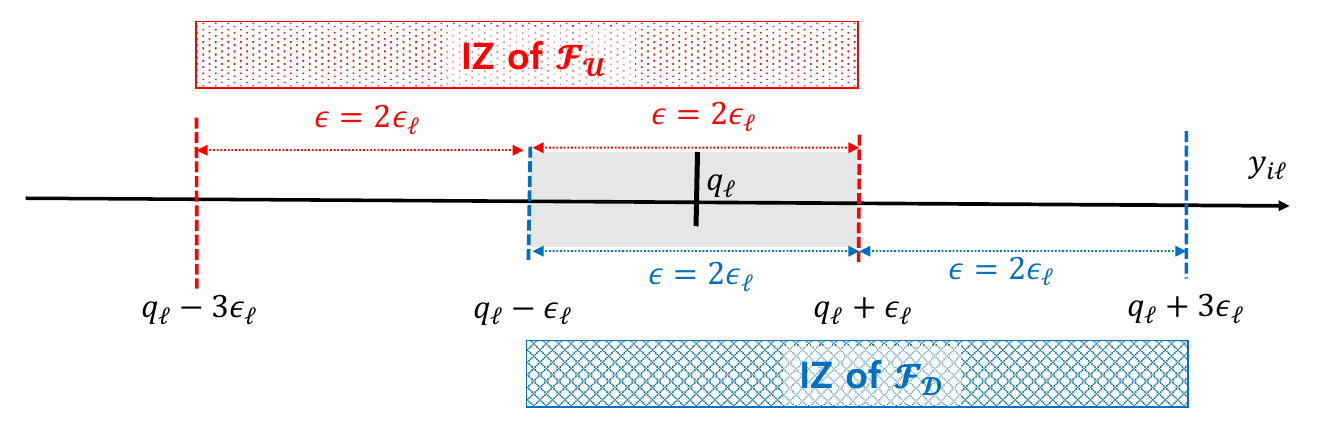}}}
        \end{minipage}
        \begin{minipage}[t]{\textwidth}
            \centering
            \subfigure[$\epsilon = \epsilon^{(2)}_\ell = \epsilon_\ell$]{
            \scalebox{0.55}{\includegraphics[width=10.0in]{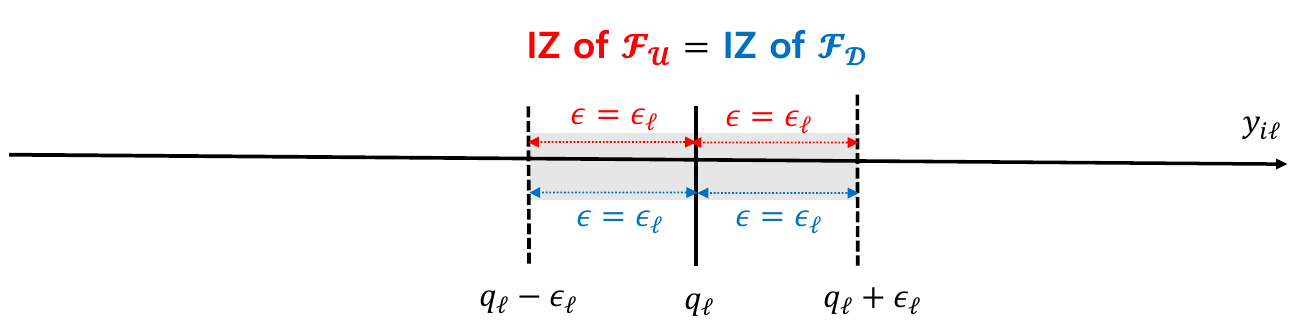}}}
        \end{minipage}
        \vspace{-6pt}
        \caption{An example of thresholds and acceptable regions for ${\cal F_U}$ and ${\cal F_D}$ with $T_\ell=2$.}
	\label{fig:example}
\end{figure}

We explain how the ${\cal IZR}$ procedure with subroutines ${\cal F_U}$ and ${\cal F_D}$ works using a simple example with one constraint. Specifically, we consider $T_\ell=2$, $\epsilon_\ell^{(1)} = 2 \epsilon_\ell$, and $\epsilon_\ell^{(2)}=\epsilon_\ell$. Figure~\ref{fig:example}(a) presents the acceptable regions (i.e., the IZs) of ${\cal F_U}$ and ${\cal F_D}$ for $\epsilon^{(1)}_\ell = 2\epsilon_\ell$ with the red-dotted and blue-checkered intervals, respectively. Additionally, the original acceptable region with threshold $q_\ell$ and tolerance level $\epsilon_\ell$ is shown as a gray-shaded interval. From the figure, one can see that the threshold of ${\cal F_U}$ is $q_\ell+\epsilon_\ell-\epsilon_\ell^{(1)} = q_\ell-\epsilon_\ell$, and the upper bound $\left(q_\ell-\epsilon_\ell\right)+\epsilon_\ell^{(1)} = q_\ell + \epsilon_\ell$ of its red-dotted acceptable region of ${\cal F_U}$ matches that of the original gray-shaded acceptable region. Similarly, the threshold of ${\cal F_D}$ is $q_\ell-\epsilon_\ell+\epsilon_\ell^{(1)} = q_\ell+\epsilon_\ell$, and the lower bound $\left(q_\ell+\epsilon_\ell\right)-\epsilon_\ell^{(1)} = q_\ell - \epsilon_\ell$ of its blue-checkered acceptable region of ${\cal F_D}$ matches that of the original gray-shaded acceptable region. Figure~\ref{fig:example}(b) presents the acceptable regions of ${\cal F_U}$ and ${\cal F_D}$ for $\epsilon^{(2)}_\ell = \epsilon_\ell$. As $\epsilon^{(2)}_\ell$ is the same as the original tolerance level, ${\cal F_U}$ and ${\cal F_D}$ become identical and have the same acceptable region as ${\cal F}_B$. 

Then we can consider the following possible cases for constraint $\ell$ of system $i$:
\begin{description}
\item [Case 1:] If $y_{i\ell} < \left(q_\ell +\epsilon_\ell-\epsilon_\ell^{(1)}\right)-\epsilon_\ell^{(1)} = q_\ell - 3 \epsilon_\ell$, then system $i$ is desirable for both ${\cal F_U}$ and ${\cal F_D}$ under the relaxed tolerance level $\epsilon^{(1)}_\ell = 2\epsilon_\ell$ and is likely to be declared feasible by both subroutines, which is a correct decision for system $i$.

\item [Case 2:] If $q_\ell-3 \epsilon_\ell \le y_{i\ell} < \left(q_\ell -\epsilon_\ell+\epsilon_\ell^{(1)}\right)-\epsilon_\ell^{(1)} = q_\ell-\epsilon_\ell$, then system $i$ is acceptable for ${\cal F_U}$ but desirable for ${\cal F_D}$ under the relaxed tolerance level $\epsilon^{(1)}_\ell$. Therefore, ${\cal F_D}$ is likely to declare system $i$ feasible while ${\cal F_U}$ can declare system $i$ either feasible or infeasible. If they make the same decision (i.e., feasible decision), then the procedure stops and declares system $i$ feasible, which is a correct decision. Otherwise, ${\cal IZR}$ proceeds with the subroutines ${\cal F_U}$ and ${\cal F_D}$ using the smaller tolerance level $\epsilon^{(2)}_\ell$.

\item [Case 3:] If $q_\ell-\epsilon_\ell \le y_{i\ell}  <\left(q_\ell +\epsilon_\ell-\epsilon_\ell^{(1)}\right)+\epsilon_\ell^{(1)} = q_\ell + \epsilon_\ell$, then system $i$ is acceptable for both subroutines under the relaxed tolerance level $\epsilon^{(1)}_\ell$. Thus, if both subroutines make the same decision, any decision is a correct decision for system $i$. Otherwise, it proceeds with the two subroutines using the smaller tolerance level $\epsilon^{(2)}_\ell$.

\item [Case 4:] If $q_\ell + \epsilon_\ell \le y_{i\ell} < \left(q_\ell -\epsilon_\ell+\epsilon_\ell^{(1)}\right)+\epsilon_\ell^{(1)}= q_\ell + 3\epsilon_\ell$, then system $i$ is unacceptable for ${\cal F_U}$ but acceptable for ${\cal F_D}$ under the relaxed tolerance level $\epsilon^{(1)}_\ell$. By similar arguments as in Case 2, if the two subroutines make the same infeasibility decision, then the procedure stops and declares system $i$ infeasible, which is a correct decision. Otherwise, it proceeds with the subroutines using the smaller tolerance level $\epsilon^{(2)}_\ell$.

\item [Case 5:] If $y_{i\ell} \ge  q_\ell + 3 \epsilon_\ell$, then system $i$ is unacceptable for both subroutines under the relaxed tolerance level $\epsilon^{(1)}_\ell$ and is likely to be declared infeasible by both subroutines, which is a correct decision for system $i$.
\end{description}

A detailed description of our proposed procedure is given in Algorithm~\ref{alg:IZR} whose statistical properties are precisely provided in Section~\ref{sec3_SG}.

\begin{algorithm}[!htb]
	\caption{\hspace{-0.3em}: \textbf{Procedure ${\cal IZR}$}}\label{alg:IZR}
    {\fontsize{11}{24}\selectfont
		\begin{algorithmic}\single 	
			\State [{\bf Setup}:] Choose initial sample size $n_0 \geq 2$, confidence level $0<1-\alpha<1$, $c\in \mathbb{N}^+$, and $\Theta = \{1,2, \ldots, k\}$. Set tolerance levels $\epsilon_\ell>0$, thresholds $q_{\ell}$, and relaxed tolerance sets ${\cal E}_{{\cal U}i\ell} = {\cal E}_{{\cal D}i\ell} = \{\epsilon_\ell^{(1)}, \ldots, \epsilon_\ell^{(T_\ell)}\}$ for each $i \in \Theta$ where $T_\ell \in \mathbb{N}^+$ and $\epsilon_\ell^{(1)}> \cdots > \epsilon_\ell^{(T_\ell)} = \epsilon_\ell$ for constraint $\ell=1,2,\ldots, s$. Set $M = \{1,2,\ldots, k\}$ and $F = \emptyset$. Calculate $h^2= 2c\eta (n_0-1)$, where $\eta>0$ satisfies
\begin{equation} \label{eq:getafb1}
g_{n_0-1}(\eta)=\beta = \left\{
\begin{array}{ll}
  \left[1-(1-\alpha)^{1/k}\right]/(\sum_{\ell=1}^s T_\ell), & \mbox{ if systems are independent}; \\
  \alpha/(k\sum_{\ell=1}^s T_\ell), & \mbox{ otherwise}.
\end{array}\right.
\end{equation}
			\State [{\bf Initialization}:] 
            \For{$i \in\Theta$}
		      \begin{itemize}\small
				\item[] Obtain $n_0$ observations $Y_{i \ell 1}, Y_{i \ell 2}, \ldots, Y_{i \ell n_0}$ and compute $S_{i \ell}^2$ for $\ell=1,2,\ldots, s$.
				
				\item[] Set $r_i=n_0$, ${\rm ON}_i = \{1,2,\ldots, s\}$, and $Z_{i \ell}=Z_{{\cal U}i \ell}^{(\epsilon)} = Z_{{\cal D}i \ell}^{(\epsilon)}=0$ for all $\epsilon \in {\cal E}_{{\cal U}i\ell} (= {\cal E}_{{\cal D}i\ell})$ and $\ell \in {\rm ON}_i$.
			\end{itemize} \EndFor
			\State [{\bf Feasibility Check}:]
            Set $M^{old} = M$.
            \For{$i \in M^{old}$}
			\For{$\ell \in {\rm ON}_i$}
			    \State [{\bf Subroutine ${\cal F_U}$}:]
                \For{$\epsilon \in {\cal E}_{{\cal U}i\ell}$},
                \begin{itemize}\small
                \item[] $\quad\quad\quad$If $(\sum_{j=1}^{r_{i}} Y_{i \ell j}) - r_{i}(q_\ell+\epsilon_\ell-\epsilon) \leq - R(r_{i};\epsilon, h^2, S_{i \ell}^2)$, set $Z_{{\cal U}i \ell}^{(\epsilon)}=1$ and ${\cal E}_{{\cal U}i\ell} = {\cal E}_{{\cal U}i\ell}\setminus \{ \epsilon\}$;
                \item[] $\quad\quad\quad$Else if $(\sum_{j=1}^{r_{i}} Y_{i \ell j}) - r_{i}(q_\ell+\epsilon_\ell-\epsilon) \geq R(r_{i};\epsilon, h^2, S_{i \ell}^2)$, set $Z_{{\cal U}i \ell}^{(\epsilon)}=-1$ and ${\cal E}_{{\cal U}i\ell} = {\cal E}_{{\cal U}i\ell}\setminus \{ \epsilon\}$.
                \item[]$\quad\quad\quad$If  $Z_{{\cal U}i \ell}^{(\epsilon)}\times Z_{{\cal D}i \ell}^{(\epsilon)}=1$, set $Z_{i\ell}=Z_{{\cal U}i \ell}^{(\epsilon)}(=Z_{{\cal D}i \ell}^{(\epsilon)})$ and ${\rm ON}_i = {\rm ON}_i\setminus \{ \ell\}$.
                \end{itemize}
                \EndFor
                \State [{\bf Subroutine ${\cal F_D}$}:]
                \For{$\epsilon \in {\cal E}_{{\cal D}i\ell}$},
                \begin{itemize}\small
                \item[] $\quad\quad\quad$If $(\sum_{j=1}^{r_{i}} Y_{i \ell j}) - r_{i}(q_\ell-\epsilon_\ell+\epsilon) \leq - R(r_{i};\epsilon, h^2, S_{i \ell}^2)$, set $Z_{{\cal D}i \ell}^{(\epsilon)}=1$ and ${\cal E}_{{\cal D}i\ell} = {\cal E}_{{\cal D}i\ell}\setminus \{ \epsilon\}$;
                \item[] $\quad\quad\quad$Else if $(\sum_{j=1}^{r_{i}} Y_{i \ell j}) - r_{i}(q_\ell-\epsilon_\ell+\epsilon) \geq R(r_{i};\epsilon, h^2, S_{i \ell}^2)$, set $Z_{{\cal D}i \ell}^{(\epsilon)}=-1$ and ${\cal E}_{{\cal D}i\ell} = {\cal E}_{{\cal D}i\ell}\setminus \{ \epsilon\}$.
                \item[]$\quad\quad\quad$If  $Z_{{\cal U}i \ell}^{(\epsilon)}\times Z_{{\cal D}i \ell}^{(\epsilon)}=1$, set $Z_{i\ell}=Z_{{\cal U}i \ell}^{(\epsilon)}(=Z_{{\cal D}i \ell}^{(\epsilon)})$ and ${\rm ON}_i = {\rm ON}_i \setminus \{ \ell\}$.
			    \end{itemize}
                \EndFor\\
                $\quad\quad\quad$If $Z_{i\ell}=-1$, eliminate $i$ from $M$ and exit the current \textbf{for} loop.
            \EndFor \\
            $\quad\quad$If $Z_{i\ell}=1$ for all $\ell$, move $i$ from $M$ to $F$.
            \EndFor
            \State[{\bf Stopping Condition}:]
			\begin{itemize}\small
				\item[] If $M=\emptyset$, return $F$. Otherwise, set $r_i =r_i +1$ for all $i \in M$, take one additional observation $Y_{i\ell r_i}$ for all $i \in M$ and all $\ell \in {\rm ON}_i$, and then go to [{\bf Feasibility Check}].
			\end{itemize}
		\end{algorithmic}
	}
\end{algorithm}

\noindent{\bf Remark 2.} Similar to ${\cal F}_B$, the choice of $c=1$ is recommended for ${\cal IZR}$, resulting in $\eta = \frac{1}{2}\left[(2\beta)^{-2/(n_0-1)}-1\right]$, where $\beta$ is defined in {\bf Setup}. 



\subsection{Statistical Guarantee of ${\cal IZR}$}
\label{sec3_SG}
In this section, we first discuss the statistical guarantee of the ${\cal IZR}$ procedure for the case of a single system with a single constraint and then extend the discussion to the case of multiple systems with multiple constraints.

Consider constraint $\ell$ of system $i$ in isolation. For a fixed $\epsilon \in \{\epsilon_\ell^{(1)}, \epsilon_\ell^{(2)}, \ldots, \epsilon_\ell^{(T_\ell)}\}$, let ${\cal U}^{(\epsilon)}_{i\ell}(1)$ and ${\cal U}^{(\epsilon)}_{i\ell}(-1)$ be the events that sample path $(\sum_{j=1}^{r} Y_{i \ell j}) - r(q_\ell+\epsilon_\ell-\epsilon)$ first exits the lower and upper bounds of the continuation region defined by $\left[-R(r;\epsilon, h^2, S_{i \ell}^2), R(r;\epsilon, h^2, S_{i \ell}^2)\right]$, respectively, as $r$ increases by 1 from $n_0$. Similarly, for a fixed $\epsilon \in \{\epsilon_\ell^{(1)}, \epsilon_\ell^{(2)}, \ldots, \epsilon_\ell^{(T_\ell)}\}$, let ${\cal D}^{(\epsilon)}_{i\ell}(1)$ and ${\cal D}^{(\epsilon)}_{i\ell}(-1)$ be the events that sample path $(\sum_{j=1}^{r} Y_{i \ell j}) - r(q_\ell-\epsilon_\ell+\epsilon)$ first exits the lower and upper bounds of the continuation region defined by $\left[-R(r;\epsilon, h^2, S_{i \ell}^2), R(r;\epsilon, h^2, S_{i \ell}^2)\right]$, respectively, as $r$ increases by 1 from $n_0$. Then, we may derive the following lemma.

\begin{lemma} \label{lemma:IZR1}
Under Assumption~\ref{assump:normal}, when $y_{i\ell} \geq q_\ell+\epsilon_{\ell}$, ${\cal IZR}$ guarantees
\begin{equation} \label{eq:LowerProb}
\Pr\left\{{\cal U}^{(\epsilon)}_{i\ell}(1)\right\} \leq \beta,\; \mbox{ for any } \epsilon \in \{\epsilon_\ell^{(1)}, \ldots, \epsilon_\ell^{(T_\ell)}\},
\end{equation} 
and
when $y_{i\ell} \leq q_\ell-\epsilon_{\ell}$, ${\cal IZR}$ guarantees
\begin{equation} \label{eq:UpperProb}
\Pr\left\{{\cal D}^{(\epsilon)}_{i\ell}(-1)\right\} \leq \beta,\; \mbox{ for any } \epsilon \in \{\epsilon_\ell^{(1)}, \ldots, \epsilon_\ell^{(T_\ell)}\}.
\end{equation} \vspace{-36pt}
\end{lemma}
\begin{proof}
When $y_{i\ell} = q_\ell+\epsilon_{\ell}$, the drift of the discrete process, $(\sum_{j=1}^{r} Y_{i \ell j}) - r(q_\ell+\epsilon_\ell-\epsilon)$, equals $\epsilon$ which is the tolerance level, and thus the result in (\ref{eq:LowerProb}) follows from the proof of Theorem 1 in \cite{KimNelsonTOMACS}. Similarly, when  $y_{i\ell} = q_\ell-\epsilon_{\ell}$, the drift of the discrete process, $(\sum_{j=1}^{r} Y_{i \ell j}) - r(q_\ell-\epsilon_\ell+\epsilon)$, equals $-\epsilon$, and thus the result in (\ref{eq:UpperProb}) also follows from the proof of Theorem 1 in \cite{KimNelsonTOMACS}.
\end{proof}

Let $\epsilon^*_{i\ell}$ be the random value of $\epsilon$ for which the feasibility decision is finally made by ${\cal IZR}$ (i.e., $Z_{{\cal U}i\ell}^{(\epsilon^*_{i\ell})}\times Z_{{\cal D}i\ell}^{(\epsilon^*_{i\ell})}=1$). Let ${\rm CD}_{i\ell}$ be the correct decision event that $Z_{{\cal U}i\ell}^{(\epsilon^*_{i\ell})}=Z_{{\cal D}i\ell}^{(\epsilon^*_{i\ell})}=1$ (or $Z_{{\cal U}i\ell}^{(\epsilon^*_{i\ell})}=Z_{{\cal D}i\ell}^{(\epsilon^*_{i\ell})}=-1$), i.e., system $i$ is declared as feasible (or infeasible) regarding constraint $\ell$ by ${\cal IZR}$, when $y_{i\ell} \leq q_\ell - \epsilon_\ell$ (or $y_{i\ell} \geq q_\ell + \epsilon_\ell$). Similarly, let ${\rm ICD}_{i\ell}(={\rm CD}_{i\ell}^c)$ be the incorrect decision event that $Z_{{\cal U}i\ell}^{(\epsilon^*_{i\ell})}=Z_{{\cal D}i\ell}^{(\epsilon^*_{i\ell})}=1$ (or $Z_{{\cal U}i\ell}^{(\epsilon^*_{i\ell})}=Z_{{\cal D}i\ell}^{(\epsilon^*_{i\ell})}=-1$), i.e., system $i$ is declared as feasible (or infeasible) regarding constraint $\ell$ by ${\cal IZR}$, when $y_{i\ell} \geq q_\ell + \epsilon_\ell$ (or $y_{i\ell} \leq q_\ell - \epsilon_\ell$). It is noted that ${\rm ICD}_{i\ell}$ cannot happen (i.e., $\Pr\left\{{\rm ICD}_{i\ell}\right\}=1-\Pr\left\{{\rm CD}_{i\ell}\right\}=0$) when $q_\ell - \epsilon_\ell < y_{i\ell} < q_\ell + \epsilon_\ell$ due to the definition of acceptable systems. Then, we derive the following lemma.

\begin{lemma} \label{lemma:IZR2}
Under Assumption~\ref{assump:normal}, when either $y_{i\ell} \geq q_\ell+\epsilon_{\ell}$ or $y_{i\ell} \leq q_\ell-\epsilon_{\ell}$, ${\cal IZR}$ guarantees
\begin{eqnarray*}
\Pr \left\{{\rm ICD}_{i\ell}\right\}&\leq& T_\ell \beta.
\end{eqnarray*}
\end{lemma}
\begin{proof}

We start by pointing out that since $(\sum_{j=1}^{r} Y_{i \ell j}) - r(q_\ell+\epsilon_\ell-\epsilon)\geq(\sum_{j=1}^{r} Y_{i \ell j}) - r(q_\ell-\epsilon_\ell+\epsilon)$ {for any positive integer $r$, we have ${\cal U}_{i\ell}^{(\epsilon)} (1) \subseteq {\cal D}_{i\ell}^{(\epsilon)} (1)$, ${\cal D}_{i\ell}^{(\epsilon)} (-1) \subseteq {\cal U}_{i\ell}^{(\epsilon)} (-1)$, and ${\cal U}_{i\ell}^{(\epsilon)} (1) \cap {\cal D}_{i\ell}^{(\epsilon)} (-1) = \emptyset$ for all $\epsilon \in \{\epsilon_\ell^{(1)}, \ldots, \epsilon_\ell^{(T_\ell)}\}$.} Therefore, when $y_{i\ell}\geq q_\ell+\epsilon_\ell$,
\begin{eqnarray*}
\Pr \left\{{\rm ICD}_{i\ell}\right\}
& = & \sum_{\tau=1}^{T_\ell} \Pr \left\{{\rm ICD}_{i\ell}, \epsilon^*_{i\ell} = \epsilon_\ell^{(\tau)} \right\} \\
&\leq& \sum_{\tau=1}^{T_\ell}\Pr \left\{{\cal U}_{i\ell}^{(\epsilon_\ell^{(\tau)})} (1), {\cal D}_{i\ell}^{(\epsilon_\ell^{(\tau)})} (1)\right\}\\
&=&\sum_{\tau=1}^{T_\ell}\Pr \left\{{\cal U}_{i\ell}^{(\epsilon_\ell^{(\tau)})}(1) \right\}\\
&\leq& T_\ell \beta,
\end{eqnarray*}
where the last inequality follows from Lemma~\ref{lemma:IZR1}. Similarly, we consider the case of $y_{i\ell} \leq q_\ell-\epsilon_{\ell}$. Then,
\begin{eqnarray*}
\Pr \left\{{\rm ICD}_{i\ell}\right\}
& = & \sum_{\tau=1}^{T_\ell} \Pr \left\{{\rm ICD}_{i\ell},\epsilon^*_{i\ell} = \epsilon_\ell^{(\tau)} \right\} \\
&\leq& \sum_{\tau=1}^{T_\ell}\Pr \left\{{\cal U}_{i\ell}^{(\epsilon_\ell^{(\tau)})} (-1), {\cal D}_{i\ell}^{(\epsilon_\ell^{(\tau)})} (-1)\right\}\\
&=&\sum_{\tau=1}^{T_\ell}\Pr \left\{{\cal D}_{i\ell}^{(\epsilon_\ell^{(\tau)})}(-1) \right\}\\
&\leq& T_\ell \beta, \label{eq:ICD22}
\end{eqnarray*}
where the last inequality follows from Lemma~\ref{lemma:IZR1}.
\end{proof}

Then, we derive the following theorem when determining the feasibility of $k$ systems with respect to $s$ constraints. For each system $i\in A \cup U$, the events ${\rm CD}_{i\ell}$ and ${\rm ICD}_{i\ell}$ should be interpreted as if constraint $\ell$ of system $i$ were considered in isolation (until a feasibility decision is made).

\begin{theorem}\label{thm:IZR}
Under Assumption~\ref{assump:normal}, ${\cal IZR}$ guarantees
\begin{equation*}
\Pr \left\{ \bigcap_{i=1}^k \bigcap_{\ell =1}^s {\rm CD}_{i \ell} \right\} \geq \ 1-\alpha.
\end{equation*}
\end{theorem}
\begin{proof}
For the case of independent systems, since $\beta = \left[ 1- (1-\alpha)^{1/k}\right]/(\sum_{\ell=1}^s T_\ell)$ in ${\cal IZR}$, the Bonferroni inequality yields
\begin{eqnarray}
{\Pr \left\{ \bigcap_{i=1}^k \bigcap_{\ell =1}^s {\rm CD}_{i \ell} \right\}}& = &  \prod_{i=1}^k \Pr \left\{\bigcap_{\ell =1}^s {\rm CD}_{i \ell}\right\} \nonumber \\
& \geq &  \prod_{i=1}^k  \left\{ 1-  \sum_{\ell = 1}^s \Pr \{{\rm ICD}_{i\ell}\} \right\} \nonumber \\
& \geq &  \prod_{i=1}^k \left\{ 1- \sum_{\ell = 1}^s T_\ell \beta \right\} \label{eq:CD11}\\
& = & \prod_{i =1}^{k} \left\{ 1- \left[1- (1-\alpha)^{1/k}\right] \right\} \nonumber\\
&= &  1-\alpha,\nonumber
\end{eqnarray}
where the inequality (\ref{eq:CD11}) holds due to the definition of acceptable systems and Lemma~\ref{lemma:IZR2}. For the case of dependent systems, since $\beta = \alpha/(k\sum_{\ell=1}^s T_\ell)$ in ${\cal IZR}$, the Bonferroni inequality yields
\begin{eqnarray}
{\Pr \left\{ \bigcap_{i=1}^k \bigcap_{\ell =1}^s \mbox{CD}_{i \ell} \right\}}& \geq &  1- \sum_{i=1}^k  \sum_{\ell =1}^s \Pr \{\mbox{ICD}_{i \ell}\} \nonumber \\
& \geq &  1-\sum_{i=1}^k   \sum_{\ell = 1}^s T_\ell \beta \label{eq:CD12}\\
& = &  1- \sum_{i=1}^k  \alpha/k \nonumber \\
&= &  1-\alpha, \nonumber
\end{eqnarray}
where the inequality (\ref{eq:CD12}) holds due to the definition of acceptable systems and Lemma~\ref{lemma:IZR2}.
\end{proof}


\subsection{Impact of $T_\ell$ and $\epsilon_\ell^{(\tau)}$ Values}
\label{sec3_PA}
In this section, we examine the impacts of different values of $T_\ell$ and $\epsilon_\ell^{(\tau)}$ for $\tau = 1, 2, \ldots, T_\ell$ through an example when $c=1$. Note that the continuation region of ${\cal F}_B$ applies to the monitoring statistic $(\sum_{j=1}^{r_i} Y_{i \ell j}) - r_i q_\ell$ and can be expressed as
\begin{eqnarray}
{\cal F}_B(\epsilon)\hspace{-0.6em}&:&\left[-\frac{h_B^2 S_{i \ell}^2}{2\epsilon_\ell} + \frac{\epsilon_\ell}{2}r_i,\; \frac{h_B^2 S_{i \ell}^2}{2\epsilon_\ell} - \frac{\epsilon_\ell}{2}r_i\right],\mbox{ for } n_0 \leq r_i \leq \frac{h_B^2S_{i \ell}^2}{\epsilon_\ell^2}.\label{eq:cts3}
\end{eqnarray}
Let ${\cal F_U}(\epsilon)$ and ${\cal F_D}(\epsilon)$ be the subroutines ${\cal F_U}$ and ${\cal F_D}$ with a relaxed tolerance level $\epsilon$, respectively. Resetting the monitoring statistic as $(\sum_{j=1}^{r_i} Y_{i \ell j}) - r_i q_\ell$ for the two subroutines ${\cal F_U}$ and ${\cal F_D}$ of ${\cal IZR}$ to facilitate comparison with ${\cal F}_B$ results in the shifted continuation regions of ${\cal F_U}(\epsilon)$ and ${\cal F_D}(\epsilon)$ as
\begin{eqnarray}
{\cal F_U}(\epsilon)\hspace{-0.6em} &:&\left[-\frac{h^2 S_{i \ell}^2}{2\epsilon} + \frac{\epsilon}{2}r_i+ \epsilon_\ell r_i- \epsilon r_i,\; \frac{h^2 S_{i \ell}^2}{2\epsilon} - \frac{\epsilon}{2}r_i+ \epsilon_\ell r_i- \epsilon r_i \right],\mbox{ for } n_0 \leq r_i \leq \frac{h^2S_{i \ell}^2}{\epsilon^2}, \label{eq:cts1}\\
{\cal F_D}(\epsilon)\hspace{-0.6em}&:&\left[-\frac{h^2 S_{i \ell}^2}{2\epsilon} + \frac{\epsilon}{2}r_i- \epsilon_\ell r_i+ \epsilon r_i,\; \frac{h^2 S_{i \ell}^2}{2\epsilon} - \frac{\epsilon}{2}r_i- \epsilon_\ell r_i+ \epsilon r_i \right],\mbox{ for } n_0 \leq r_i \leq \frac{h^2S_{i \ell}^2}{\epsilon^2}.\label{eq:cts2}
\end{eqnarray}
It should be noted that if $\epsilon = \epsilon^{(T_\ell)}_\ell = \epsilon_\ell$, the shifted terms, $\epsilon_\ell r_i- \epsilon r_i$ and $-\epsilon_\ell r_i+ \epsilon r_i$ in (\ref{eq:cts1}) and (\ref{eq:cts2}), respectively, disappear, and the continuation regions in (\ref{eq:cts1}) and (\ref{eq:cts2}) become identical.

Synchronizing the monitoring statistic $(\sum_{j=1}^{r_i} Y_{i \ell j}) - r_i q_\ell$ for both ${\cal F}_B$ and ${\cal IZR}$ allows one to directly compare the computational efficiency of these two procedures by simultaneously checking the continuation region in (\ref{eq:cts3}) and the shifted continuation regions in (\ref{eq:cts1}) and (\ref{eq:cts2}). We consider an example with one constraint by setting $n_0=20$, $\alpha=0.05$, $\sigma_i^2 = 1$, and $\epsilon_\ell = 0.02$. For the values of $T_\ell$ and $\epsilon_\ell^{(\tau)}$ for $\tau = 1, 2, \ldots, T_\ell$, we consider the following two settings:
\begin{itemize}
\item $T_\ell=2$ with relaxed tolerance levels $\epsilon_\ell^{(1)} =2 \epsilon_\ell = 0.04$ and $\epsilon_\ell^{(2)} = \epsilon_\ell  =0.02$; and
\item $T_\ell=3$ with relaxed tolerance levels $\epsilon_\ell^{(1)} =2^2 \epsilon_\ell  = 0.08,\; \epsilon_\ell^{(2)}=2 \epsilon_\ell =0.04$, and $\epsilon_\ell^{(3)}= \epsilon_\ell =0.02$.
\end{itemize}

With these settings, the shifted continuation regions of ${\cal F_U}(\epsilon)$ and ${\cal F_D}(\epsilon)$ in (\ref{eq:cts1}) and (\ref{eq:cts2}), and the continuation region of ${\cal F}_B$ in (\ref{eq:cts3}) are shown in Figure~\ref{fig:ctsregions} when $k \in \{1, 100\}$. Additionally, the shifted continuation region of ${\cal F_U}(\epsilon^{(\tau-1)}_\ell)$ is always included in that of ${\cal F_U}(\epsilon^{(\tau)}_\ell)$, and the shifted continuation region of ${\cal F_D}(\epsilon^{(\tau-1)}_\ell)$ is always included in that of ${\cal F_D}(\epsilon^{(\tau)}_\ell)$ because $\epsilon^{(\tau-1)}_\ell > \epsilon^{(\tau)}_\ell$ for any $\tau = 2,3,\ldots,T_\ell$ (see Appendix A). Also, the lower bound of ${\cal F_U}(\epsilon)$ is less than the lower bound of ${\cal F_D}(\epsilon)$ because $\epsilon \geq \epsilon_\ell$ and the same is true for the upper bounds.

\begin{figure}[!tb]
        \centering
        \begin{minipage}[b]{0.48\textwidth}
            \centering
            \subfigure[$T_\ell=2$ and $k = 1$]{
            \scalebox{0.27}{\includegraphics[width=10.5in]{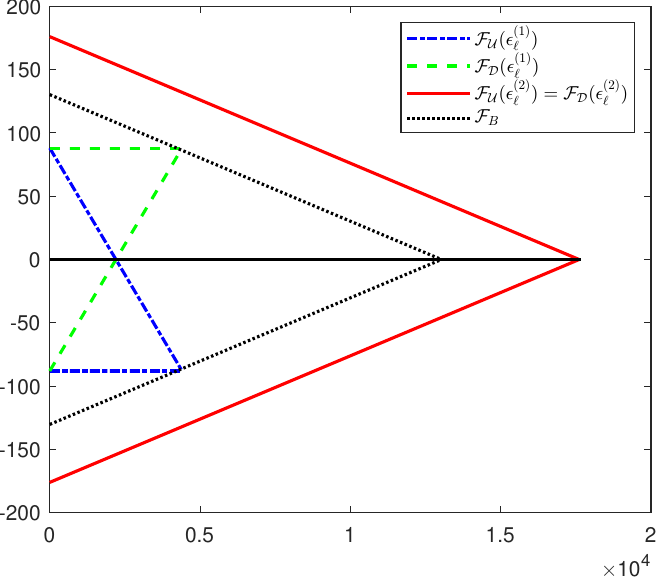}}}
        \end{minipage}
        \begin{minipage}[b]{0.48\textwidth}
            \centering
            \subfigure[$T_\ell=2$ and $k = 100$]{
            \scalebox{0.28}{\includegraphics[width=10.5in]{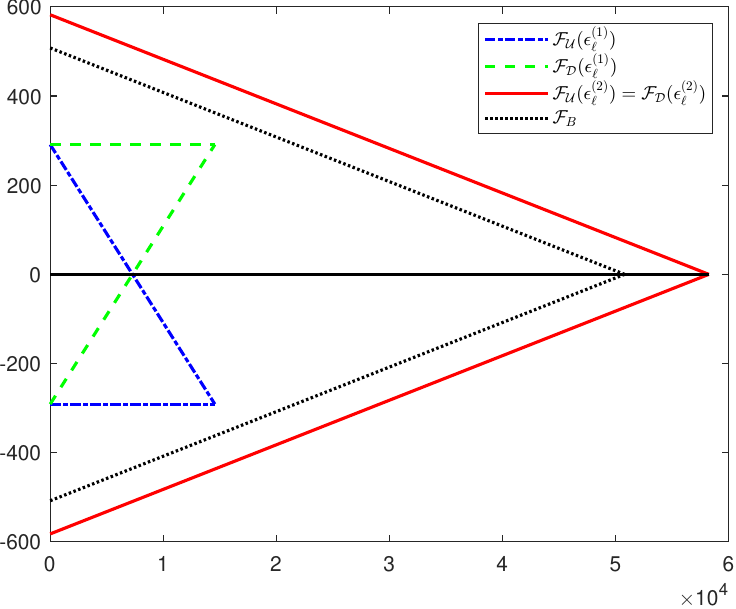}}}
        \end{minipage}
        \begin{minipage}[b]{0.48\textwidth}
            \centering
            \subfigure[$T_\ell=3$ and $k = 1$]{
            \scalebox{0.28}{\includegraphics[width=10.5in]{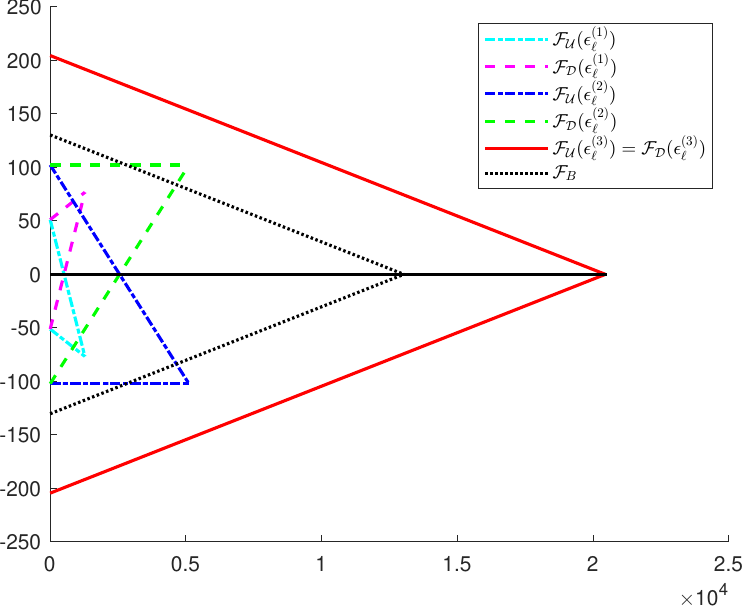}}}
        \end{minipage}
        \begin{minipage}[b]{0.48\textwidth}
            \centering
            \subfigure[$T_\ell=3$ and $k = 100$]{
            \scalebox{0.28}{\includegraphics[width=10.5in]{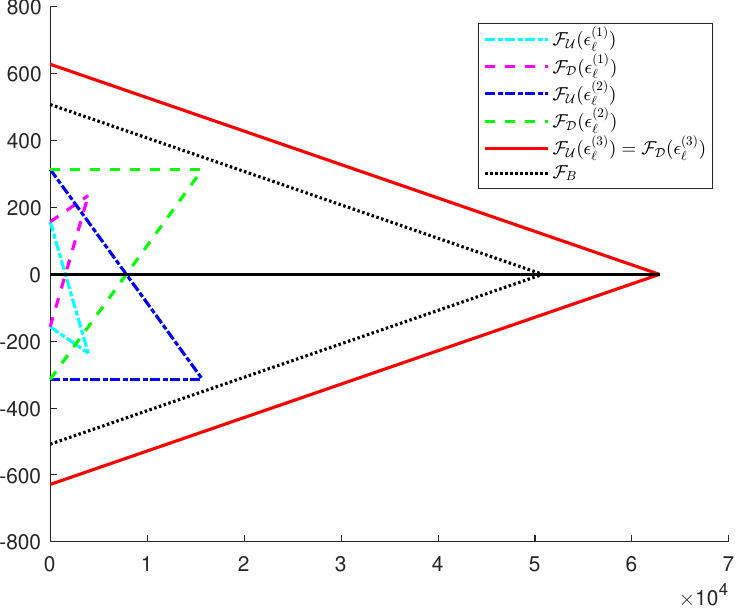}}}
        \end{minipage}
        \vspace{-6pt}
        \caption{Example continuation regions of ${\cal F_U}(\epsilon), {\cal F_D}(\epsilon)$, and ${\cal F}_B$.} \label{fig:ctsregions}
\end{figure}

For a certain $\tau <T_\ell$ (i.e., $\epsilon = \epsilon^{(\tau)}_\ell > \epsilon_\ell$), one of the following three scenarios may happen: the monitoring statistic first exits the continuation regions of ${\cal F_U}(\epsilon^{(\tau)}_\ell)$ and ${\cal F_D}(\epsilon^{(\tau)}_\ell)$ via
\begin{description}
    \item[](i) the upper boundary of ${\cal F_D}(\epsilon^{(\tau)}_\ell)$ (after exiting the upper boundary of ${\cal F_U}(\epsilon^{(\tau)}_\ell)$ as well),
    \item[](ii) the lower boundary of ${\cal F_U}(\epsilon^{(\tau)}_\ell)$ (after exiting the lower boundary of ${\cal F_D}(\epsilon^{(\tau)}_\ell)$ as well), and 
    \item[](iii) the upper boundary of ${\cal F_U}(\epsilon^{(\tau)}_\ell)$ and the lower boundary of ${\cal F_D}(\epsilon^{(\tau)}_\ell)$. 
 \end{description}
 If either scenario (i) or scenario (ii) occurs, then the system is declared as infeasible in scenario (i) or feasible in scenario (ii), respectively, and the ${\cal IZR}$ procedure terminates without requiring additional observations. Otherwise (i.e., if scenario (iii) happens), then ${\cal IZR}$ needs to wait for the feasibility decisions of ${\cal F_U}(\epsilon^{(\tau+1)}_\ell)$ and ${\cal F_D}(\epsilon^{(\tau+1)}_\ell)$. Thus there is a tradeoff in choosing the relaxed tolerance levels. In particular, a larger relaxed tolerance level $\epsilon$ constructs smaller shifted continuation regions of ${\cal F_U}(\epsilon)$ and ${\cal F_D}(\epsilon)$, and hence ${\cal IZR}$ may terminate with fewer observations while the likelihood of scenarios (i) or (ii) is lower. 

When $\tau = T_\ell$ (i.e., $\epsilon = \epsilon^{(T_\ell)}_\ell = \epsilon_\ell$), then the continuation regions of ${\cal F_U}(\epsilon^{(T_\ell)}_\ell)$ and ${\cal F_D}(\epsilon^{(T_\ell)}_\ell)$ are always identical, and thus ${\cal IZR}$ must terminate with the same decision. The common continuation region of ${\cal F_U}(\epsilon^{(T_\ell)}_\ell)$, ${\cal F_D}(\epsilon^{(T_\ell)}_\ell)$ (marked in the red solid triangle in Figure~\ref{fig:ctsregions}) has a different size compared to that of ${\cal F}_B$ (marked in a black dotted triangle in Figure~\ref{fig:ctsregions}) due to the values of $\eta$ and $\eta_B$. It should be noted that $\beta_B \geq \beta$ (and $\eta \geq \eta_B$) due to equations (\ref{eq:getafb}) and (\ref{eq:getafb1}), $\sum_{\ell=1}^s T_\ell \geq s$, and $g_n(\cdot)$ being a decreasing function, and thus the common continuation region of ${\cal F_U}(\epsilon^{(T_\ell)}_\ell)$, ${\cal F_D}(\epsilon^{(T_\ell)}_\ell)$ is larger than the continuation region of ${\cal F}_B$, even though ${\cal F_U}(\epsilon^{(T_\ell)}_\ell)$, ${\cal F_D}(\epsilon^{(T_\ell)}_\ell)$, and ${\cal F}_B$ share the same acceptable region (i.e., the same IZ) as shown in Figure~\ref{fig:example} (b). As a result, if ${\cal IZR}$ terminates with the smallest relaxed tolerance level which is the same as the tolerance level of ${\cal F}_B$, (i.e., $\epsilon = \epsilon_\ell^{(T_\ell)} = \epsilon_\ell$ or $\tau = T_\ell$), then one would expect ${\cal F}_B$ to require fewer simulation observations than ${\cal IZR}$.  Moreover, if one increases the value of $T_\ell$ (and thus $\sum_{\ell=1}^s T_\ell$ in (\ref{eq:getafb1}) becomes larger), then $\beta$ becomes smaller (i.e., $\eta$ becomes larger). Therefore, large $T_\ell$ may reduce the efficiency of ${\cal IZR}$ and thus $T_\ell>3$ is rarely recommended. Section~\ref{sec3_ES} provides a brief discussion based on a numerical example to empirically compare $T_\ell=2$ with $T_\ell=3$.

Nevertheless, when the number of systems (i.e., $k$) or the initial sample size for each system (i.e., $n_0$) increases, $\eta_B$ becomes similar to $\eta$. Specifically, the relative efficiency of ${\cal IZR}$ when $\tau = T_\ell$ compared with ${\cal F}_B$ can be assessed by $\frac{\eta}{\eta_B}$ which is greater than 1. If $c=1$, one may have
\begin{eqnarray} \label{eq:etaratio} 
\frac{\eta}{\eta_B} & = & \frac{\left(2 \beta \right)^{-\frac{2}{n_0-1}}-1}{\left(2 \beta_B \right)^{-\frac{2}{n_0-1}}-1}= \left\{
\begin{array}{ll}
  \frac{\left({\sum_{\ell = 1}^s T_\ell} \right)^{\frac{2}{n_0-1}}-\left(2 \left[1-(1-\alpha)^{1/k}\right] \right)^{\frac{2}{n_0-1}}}{\left({s} \right)^{\frac{2}{n_0-1}}-\left(2 \left[1-(1-\alpha)^{1/k}\right] \right)^{\frac{2}{n_0-1}}}, & \mbox{ if systems are independent}; \\
  \frac{\left({\sum_{\ell = 1}^s T_\ell} \right)^{\frac{2}{n_0-1}}-\left(2 \alpha/k \right)^{\frac{2}{n_0-1}}}{\left({s} \right)^{\frac{2}{n_0-1}}-\left(2 \alpha/k \right)^{\frac{2}{n_0-1}}}, & \mbox{ otherwise}.
\end{array}\right.
\end{eqnarray}


\begin{table}[!tb]\small 
\vspace{6pt}
		\centering
			\begin{tabular}{  c | cccc|cccc}
				\toprule
				    \multicolumn{9}{c}{$\sum_{\ell = 1}^s T_\ell = 2s$} \\  \midrule
		  	& \multicolumn{4}{c|}{$s=1$} & \multicolumn{4}{c}{$s=5$}   \\ 
              $k$  & $n_0=10$ & $n_0=20$ & $n_0=30$ & $n_0=50$ & $n_0=10$ & $n_0=20$ & $n_0=30$ & $n_0=50$  \\ \midrule
			1&	1.4158& 	1.3517& 	1.3335& 	1.3199& 	1.2867& 	1.2242& 	1.2071& 	1.1944 \\
	        10&	1.2607& 	1.1978& 	1.1807& 	1.1682& 	1.2228& 	1.1581& 	1.1409& 	1.1285 \\ 
                100&	1.2126& 	1.1469& 	1.1296& 	1.1172& 	1.1963& 	1.1281& 	1.1105& 	1.0980\\ 
                1000&	1.1914& 	1.1221& 	1.1043& 	1.0918& 	1.1832& 	1.1115& 	1.0933& 	1.0805  \\ \bottomrule
                \toprule
				    \multicolumn{9}{c}{$\sum_{\ell = 1}^s T_\ell = 3s$} \\  \midrule
		  	& \multicolumn{4}{c|}{$s=1$} & \multicolumn{4}{c}{$s=5$}   \\ 
              $k$  & $n_0=10$ & $n_0=20$ & $n_0=30$ & $n_0=50$ & $n_0=10$ & $n_0=20$ & $n_0=30$ & $n_0=50$  \\ \midrule
			1&	1.6904& 	1.5696& 	1.5361& 	1.5113& 	1.4761& 	1.3632& 	1.3329& 	1.3108 \\
	        10&	1.4329& 	1.3204& 	1.2905& 	1.2689& 	1.3700& 	1.2560& 	1.2265& 	1.2054 \\ 
                100&	1.3530& 	1.2379& 	1.2083& 	1.1873& 	1.3259& 	1.2075& 	1.1776& 	1.1566\\ 
                1000&	1.3178& 	1.1978& 	1.1677& 	1.1467& 	1.3041& 	1.1806& 	1.1499& 	1.1287   \\ \bottomrule
			\end{tabular}
		\caption{Values of $\eta/\eta_B$ for different $k$, $s$, and $n_0$ values with independent systems and $c=1$.}
		\label{tab:etaratio}
	\end{table}

Table~\ref{tab:etaratio} shows some values of $\frac{\eta}{\eta_B}$ for different $k$, $s$, and $n_0$, when systems are independent and $c=1$. In equation (\ref{eq:etaratio}), since (i) $\left({\sum_{\ell = 1}^s T_\ell} \right)^{\frac{2}{n_0-1}} \geq \left(s\right)^{\frac{2}{n_0-1}}>0$ and (ii) $\left(2 \left[1-(1-\alpha)^{1/k}\right] \right)^{\frac{2}{n_0-1}}$ and $\left(2 \alpha/k \right)^{\frac{2}{n_0-1}}$ are non-negative decreasing sequences converging to $0$ as $k$ increases under fixed $\alpha$ and $n_0$ values, one may have
\begin{equation*} \label{eq:asymptotick} 
\frac{\eta}{\eta_B} \geq \lim_{k \rightarrow \infty} \frac{\eta}{\eta_B}  = \left(\frac{1}{s} \sum_{\ell = 1}^s T_\ell\right)^{\frac{2}{n_0-1}}.
\end{equation*}

Similarly, L'Hôpital's rule yields that under fixed $\alpha$ and $k$ values, 
\begin{equation*} \label{eq:asymptotick2} 
\frac{\eta}{\eta_B} \geq \lim_{n_0 \rightarrow \infty} \frac{\eta}{\eta_B}  = \left\{
\begin{array}{ll}
  \frac{\ln{\left({\sum_{\ell = 1}^s T_\ell} \right)}-\ln{\left(2 \left[1-(1-\alpha)^{1/k}\right] \right)}}{\ln{\left({s} \right)}-\ln{\left(2 \left[1-(1-\alpha)^{1/k}\right] \right)}}, & \mbox{ if systems are independent}; \\
  \frac{\ln{\left({\sum_{\ell = 1}^s T_\ell} \right)}-\ln{\left(2 \alpha/k \right)}}{\ln{\left({s} \right)}-\ln{\left(2 \alpha/k \right)}}, & \mbox{ otherwise}.
\end{array}\right.
\end{equation*}

From the above results, if $\sum_{\ell = 1}^s T_\ell \simeq s$ (even though $\sum_{\ell = 1}^s T_\ell \geq s$), $k$ is large, or $n_0$ is large, then $\frac{\eta}{\eta_B}\simeq1$. In other words, the sizes of the continuation regions of ${\cal F_U}(\epsilon^{(T_\ell)}_\ell)$, ${\cal F_D}(\epsilon^{(T_\ell)}_\ell)$, and ${\cal F}_B$ become similar as can be seen by comparing Figures~\ref{fig:ctsregions} (a) and (b) (or Figures~\ref{fig:ctsregions} (c) and (d)) where $k=1$ and $k=100$, respectively. Consequently, when $k$ or $n_0$ is large, ${\cal IZR}$ and ${\cal F}_B$ require similar numbers of observations in the worst case when $\tau=T_\ell$ and ${\cal IZR}$ terminates with the smallest (original) tolerance level $\epsilon_\ell$.


\subsection{Empirical Study for $T_\ell$ and $\epsilon_\ell^{(\tau)}$ Values}
\label{sec3_ES}
In this section, we provide a brief numerical analysis to support the general recommendation of setting $T_\ell = 2$. We consider a case with one system and one constraint (i.e., $k = s = 1$), focusing on two possible values for $T_1\in \{2, 3\}$. We set $n_0=20$, $\alpha=0.05$, and $\sigma_i^2 = 1$, and also set the original tolerance level to $\epsilon_1=\epsilon_1^{(T_1)}=0.02$. The threshold constant is set to $q_1=0$ and two different levels of $y_{11}$ are considered as $y_{11}\in \{0.02, 0.5\}$. Note that $y_{11}=0.02$ corresponds to ``the slippage configuration," which is considered the most difficult case for checking feasibility. In contrast, determining feasibility for $y_{11}=0.5$ is relatively easy.

To determine the values of $\{\epsilon_1^{(1)},\ldots, \epsilon_1^{(T_1)}\}$, we introduce a positive integer $\xi >1$, where the relaxed tolerance levels are set as $\{\epsilon_1^{(1)}=\xi\epsilon_1, \epsilon_1^{(2)}=\epsilon_1\}$ for $T_1=2$ and $\{\epsilon_1^{(1)}=\xi^2 \epsilon_1, \epsilon_1^{(2)}=\xi \epsilon_1, \epsilon_1^{(3)}=\epsilon_1\}$ for $T_1=3$.
Figure \ref{fig:T2Recommendation} shows the number of observations required before concluding a feasibility decision, where we consider different values of $\xi$ as $\xi \in \{2, 3, \ldots, 16\}$. 

\begin{figure}[tb!]
        \centering
        \begin{minipage}[b]{0.48\textwidth}
            \centering
            \subfigure[$y_{11}=0.02$]{
            \scalebox{0.3}{\includegraphics[width=10.5in]{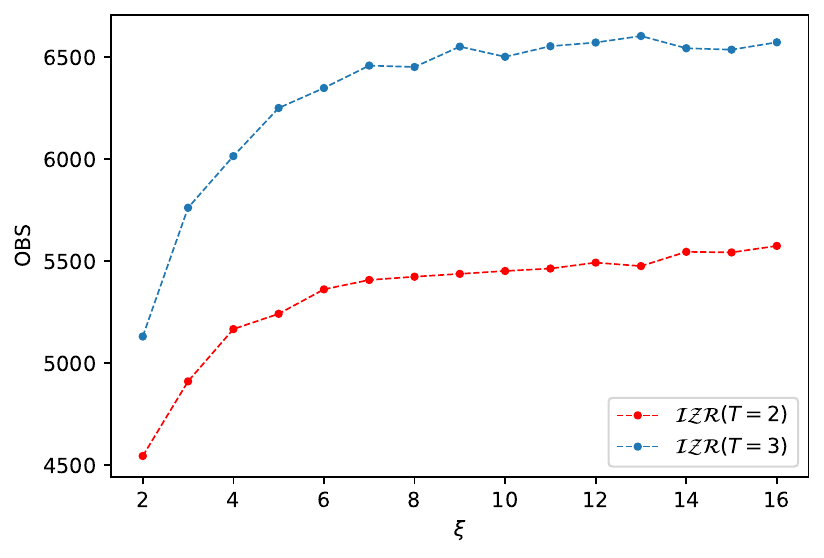}}}
        \end{minipage}
        \begin{minipage}[b]{0.48\textwidth}
            \centering
            \subfigure[$y_{11}=0.5$]{
            \scalebox{0.3}{\includegraphics[width=10.5in]{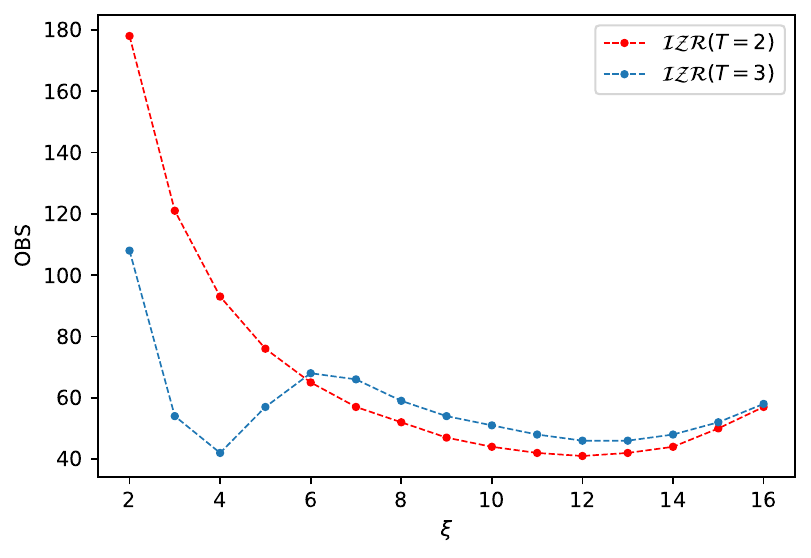}}}
        \end{minipage}
        \vspace{-6pt}
        \caption{Required number of observations to conclude feasibility decision for a single system with a single constraint and mean $y_{11}\in \{0.02, 0.5\}$. We set threshold constant $q_1=0$ in this experiment.} \label{fig:T2Recommendation}
\end{figure}

For the most difficult case (i.e., $y_{11}=0.02$), $T_1=2$ consistently outperforms $T_1=3$ for all values of $\xi$ we considered. This is because ${\cal IZR}$ tends to terminate at the smallest relaxed tolerance level (i.e., $\epsilon_1=\epsilon_1^{(T_1)}=0.02$) when $y_{11}=0.02$. As discussed in Section \ref{sec3_PA}, selecting a larger $T_1$ reduces $\beta$, which can lower the efficiency of the procedure. For the easier case (i.e., $y_{11}=0.5$), $T_1=3$ performs better than $T_1=2$ when $\xi$ is small (e.g., $\xi \leq 5)$, but both $T_\ell$ values yield similar performance for larger values of $\xi$ (i.e., $\xi \geq 6$). Moreover, the best results for $T_\ell = 2$ (obtained for $\xi=12$) and for $T_\ell =3$ (obtained for $\xi=4$) are similar. This suggests that setting $T_1=2$ with an appropriate value of $\xi$ can perform similarly as setting $T_1=3$ even when $y_{11}$ is large.  Motivated by the above results, we propose a version of the ${\cal IZR}$ procedure in Section \ref{sec4} that sets $T_\ell=2$ for $\ell=1,\ldots,s$ and selects the parameter $\xi$ effectively.

\section{Indifference-Zone Relaxation Procedure with Estimation}
\label{sec4}

As described in Sections \ref{sec3_PA} and \ref{sec3_ES}, the values of $T_\ell$ and $\epsilon_\ell^{(\tau)}$ for $\tau = 1, 2, \ldots, T_\ell$ can affect the computational efficiency of the ${\cal IZR}$ procedure, the best choices of these parameters depend on the problem, and letting $T_\ell=2$ with a good choice of $\epsilon^{(1)}_\ell$ typically works well. As problem features are often difficult to predict in advance, in this section, we provide and analyze a version of our procedure, referred to as the Indifference-Zone Relaxation with Estimation (${\cal IZE}$) procedure, that sets $T_\ell=2$ and employs preliminary simulation data to estimate a relaxed tolerance level $\epsilon^{(1)}_{i \ell}$ for system $i$ and constraint $\ell$. We provide the ${\cal IZE}$ procedure in Section~\ref{sec:4.1} and prove the statistical validity of ${\cal IZE}$ in Section~\ref{sec4.2}.



\subsection{Procedure ${\cal IZE}$}
\label{sec:4.1}

In this section, we first focus on the case with one system and one constraint to determine appropriate values of the relaxed tolerance level $\epsilon^{(1)}_{i \ell}$. We consider two scenarios for each $i$ and $\ell$: (i) using $T_{\ell}=1$ and tolerance level $\epsilon^{(1)}_{i\ell} = \epsilon_\ell$ or (ii) using $T_{\ell}=2$ and tolerance levels $\epsilon^{(1)}_{i\ell} = \xi_{i\ell}\epsilon_\ell$ and $\epsilon^{(2)}_{i\ell} = \epsilon_\ell$, where $\xi_{i\ell}$ is a constant greater than 1.

To compare the above two scenarios, one may estimate the computational efficiency of each scenario. We start by considering the case when $|y_{i\ell}-q_\ell|$ and $\sigma_{i\ell}^2$ are known. Under $k=s=1$, let $\eta_{T}$ be the solution of $g_{n_0-1}(\eta_{T})=\frac{\alpha}{T}$ and $h_T^2 = 2c(n_0-1)\eta_{T} $ when $T=1$ or $2$, see (\ref{eq:getafb1}). As shown in \cite{Healey:HAK:TOMACS}, the approximate number of observations required by our procedure with $T_{\ell}=1$ for checking feasibility can be described as
\begin{equation} \label{eq:appnum1}
\frac{h^2_1 \sigma_{i\ell}^2}{2c \epsilon_\ell \left[|y_{i\ell}-q_\ell|+\frac{\epsilon_\ell}{2c} \right]}. 
\end{equation}
Similarly, as shown in \cite{Healey:HAK:TOMACS} and \cite{Lee2018}, the approximate number of observations required by our procedure with $T_{\ell}=2$ for checking feasibility can be described as
\begin{align}
\frac{h^2_2 \sigma_{i\ell}^2}{2c \xi_{i\ell} \epsilon_\ell \left[|y_{i\ell}-q_\ell|-(\xi_{i\ell}-1)\epsilon_\ell + \frac{\xi_{i\ell} \epsilon_\ell}{2c} \right]},&\;  \mbox{ if ${\cal IZR}$ terminates with $\epsilon^{(1)}_{i\ell} = \xi_{i\ell}\epsilon_\ell$}, \label{eq:appnum2_1}\\
\frac{h^2_2 \sigma_{i\ell}^2}{2c \epsilon_\ell \left[|y_{i\ell}-q_\ell|+\frac{\epsilon_\ell}{2c} \right]},&\;  \mbox{ if ${\cal IZR}$ terminates with $\epsilon^{(2)}_{i\ell} = \epsilon_\ell$}. \label{eq:appnum2_2}
\end{align} 

It should be noted that the value of (\ref{eq:appnum1}) is always smaller than the value of (\ref{eq:appnum2_2}) due to $h^2_1 <h^2_2$, even though the continuation regions for $T=1$ and $T=2$ under tolerance level $\epsilon_\ell$ become similar if $k$ or $n_0$ is large as described at the end of Section~\ref{sec3_PA}.

The first motivation of the ${\cal IZE}$ procedure is based on the fact that our ${\cal IZE}$ procedure with $T_{\ell}=2$ should use $\xi_{i\ell}$ satisfying the following two conditions: (i) the procedure is likely to terminate with $\epsilon^{(1)}_{i\ell}$ and (ii) the value of (\ref{eq:appnum2_1}) needs to be minimized. For the first condition, one may recall Cases 1 and 5 (i.e., the system is clearly desirable or unacceptable) in Section~\ref{sec3_1}. Specifically, if $|y_{i\ell}-q_\ell| \geq (2\xi_{i\ell}-1) \epsilon_\ell$ (equivalently, either $y_{i\ell} \leq q_\ell + (1-2\xi_{i\ell}) \epsilon_\ell$ or $y_{i\ell} \geq q_\ell + (2\xi_{i\ell}-1) \epsilon_\ell$), our procedure with $T_{\ell}=2$ determines the feasibility of the system using $\epsilon^{(1)}_{i\ell}$ with high probability. Hence, $\xi_{i\ell}$ should satisfy $\xi_{i\ell} \leq \frac{|y_{i\ell}-q_\ell|+\epsilon_\ell}{2\epsilon_\ell}$. For the second condition, we minimize the value of (\ref{eq:appnum2_1}) within the range of $1 \leq \xi_{i\ell} \leq \frac{|y_{i\ell}-q_\ell|+\epsilon_\ell}{2\epsilon_\ell}$ and the optimal value $\xi_{i\ell}^*$ of $\xi_{i\ell}$ can be derived as $\xi_{i\ell}^* = \frac{|y_{i\ell}-q_\ell|+\epsilon_\ell}{2\epsilon_\ell}$, resulting in $|y_{i\ell}-q_\ell| = (2\xi_{i\ell}^* -1)\epsilon_\ell$.

The second motivation of the ${\cal IZE}$ procedure is based on the fact the value of (\ref{eq:appnum1}) should be larger than the value of (\ref{eq:appnum2_1}) when the ${\cal IZE}$ procedure introduces $T_{\ell}=2$ with $\xi_{i\ell}^*$. It is also undesirable to use $\xi^*_{i \ell}$ values close to 1, as such values are unlikely to provide significant computational savings. Therefore, we set a minimum value of $\xi^*_{i \ell}$ as 2 as in Section~\ref{sec3_1}. If (i) $|y_{i\ell}-q_\ell| = (2\xi_{i\ell}^*-1) \epsilon_\ell$ and (ii) $\xi_{i\ell}^* \geq 2$, then with $c=1$,  
\begin{equation*} 
\frac{\xi_{i\ell}^* \left[|y_{i\ell}-q_\ell|-(\xi_{i\ell}^*-1)\epsilon_\ell + \frac{\xi_{i\ell}^* \epsilon_\ell}{2} \right]}{\left[|y_{i\ell}-q_\ell|+\frac{\epsilon_\ell}{2} \right]}\geq\frac{12}{7},
\end{equation*}
and hence (\ref{eq:appnum1}) $>$ (\ref{eq:appnum2_1}) is guaranteed by 
\begin{equation}\label{eq:hratio} 
\frac{\eta_2}{\eta_1}  =  \frac{h_2^2}{h_1^2} <\frac{12}{7}. 
\end{equation}
In fact, $\frac{\eta}{\eta_B}$ in (\ref{eq:etaratio}) is equivalent to $\frac{\eta_2}{\eta_1}$ in (\ref{eq:hratio}) when $k=s=1$. It should be noted that (\ref{eq:hratio}) holds for $n_0 \geq 6$ and $0<\alpha\leq0.1$ (i.e., (\ref{eq:etaratio}) leads to $\frac{\eta_2}{\eta_1} = \frac{(2)^{\frac{2}{n_0-1}}- (2\alpha)^{\frac{2}{n_0-1}}}{1- (2\alpha)^{\frac{2}{n_0-1}}} \leq  \frac{(2)^{\frac{2}{5}}- (0.2)^{\frac{2}{5}}}{1- (0.2)^{\frac{2}{5}}} \simeq 1.673 < \frac{12}{7} \simeq 1.714$), and this guarantees (\ref{eq:appnum1}) $>$ (\ref{eq:appnum2_1}) in turn when $c=1$. As a result, we reformulate $\xi_{i\ell}^* = \max\left\{2,\frac{|y_{i\ell}-q_\ell|+\epsilon_\ell}{2\epsilon_\ell}\right\}$ based on the above two motivations.


To implement the above concepts within our procedure, $|y_{i\ell}-q_\ell|$ needs to be estimated. It should be noted that simulation observations for estimating $|y_{i\ell}-q_\ell|$ cannot be shared by the monitoring statistics to ensure that the monitoring statistics are independent of the parameter values and hence the statistical validity of the procedure. Therefore $n'_0 \geq2$ data points, $Y'_{i \ell 1}, Y'_{i \ell 2}, \ldots, Y'_{i \ell n'_0}$, are introduced to estimate both $|y_{i\ell}-q_\ell|$ and $\sigma_{i\ell}^2$, but those observations are not included in the monitoring statistics. On the other hand, $n''_0 \geq 0$, $n''_0 \neq 1$, data points, $Y_{i \ell 1}, Y_{i \ell 2}, \ldots, Y_{i \ell n''_0}$ are used only to estimate $\sigma_{i\ell}^2$ (but not to estimate $|y_{i\ell}-q_\ell|$), and are included in monitoring statistics. Thus, we collect a total of $n'_0+n''_0$ observations to estimate the needed parameters.

Let $\widehat{\Delta}_{i\ell}$ denote an estimator of $|y_{i\ell}-q_\ell|$. It should be noted that it is not necessary to handle $|y_{i\ell}-q_\ell| < \epsilon_\ell$ (i.e., system $i$ is acceptable regarding constraint $\ell$) because the ${\cal IZR}$ procedure with $T_{\ell}=2$ is unlikely to terminate with $\epsilon^{(1)}_{i\ell}$, and setting $T_\ell=1$ may perform better than setting $T_\ell=2$ when $|y_{i\ell}-q_\ell| < \epsilon_\ell$. Let $\bar{Y}'_{i\ell}$ denote $\frac{1}{n'_0} \sum_{j=1}^{n'_0} Y'_{i\ell j}$ and, therefore, we set the minimum value of the estimator of $|y_{i\ell}-q_\ell|$ as $\epsilon_\ell$ and we define our estimator  $\widehat{\Delta}_{i\ell}$ of $|y_{i\ell}-q_\ell|$ by
	\begin{eqnarray} \label{eq:delta}
		\widehat{\Delta}_{i\ell} &= \max \left\{ \epsilon_\ell, \left| \bar{Y}'_{i\ell}-q_\ell \right| \right\}.
	\end{eqnarray}
Since $n'_0$ cannot be large (e.g., for fairness sake, we will let $n'_0+n''_0=n_0$ in our numerical results), directly estimating $|y_{i\ell}-q_\ell|$ by $\left|\frac{1}{n'_0} \sum_{j=1}^{n'_0} \left(Y'_{i\ell j}-q_\ell\right) \right|$ can be too aggressive because of estimation error. Thus, we introduce a constant $0.5 < \nu \leq 1$ to alleviate the impact of estimation error of $\widehat{\Delta}_{i\ell}$. Therefore, we finally suggest 
\begin{equation}\label{eq:xihat} 
\epsilon_{i\ell}^{(1)} = \nu \hat{\xi}^*_{i\ell} \epsilon_\ell\; \mbox{ where }\; \hat{\xi}_{i\ell}^*= \max \left\{2, \frac{\widehat{\Delta}_{i\ell}+\epsilon_\ell}{2\epsilon_\ell}\right\}. 
\end{equation}
Note that $0.5 < \nu$ guarantees $\epsilon^{(1)}_{i\ell} > \epsilon^{(2)}_{i\ell}= \epsilon_\ell$ and  $\nu \leq 1$ implies a conservative selection of $\epsilon^{(1)}_{i\ell}$. 

To obtain a more precise sample variance while maintaining statistical validity, we suggest a new estimator of $\sigma_{i\ell}^2$ in our situation. Let $S_{i\ell}^{'2}$ be the sample variance of the dataset, $\{Y'_{i \ell j}; j=1,2, \ldots, n'_0\}$ and $S_{i\ell}^{''2}$ be the sample variance of $\{Y_{i \ell j}; j=1,2, \ldots, n''_0\}$. Then, a new sample variance $\widetilde{S}_{i\ell}^2$ is defined by
\begin{equation}\label{eq:samplevariance} 
\widetilde{S}_{i\ell}^2 = 
\left\{
\begin{array}{ll}
\frac{n'_0-1}{n'_0 + n''_0 -2} S_{i\ell}^{'2} + \frac{n''_0-1}{n'_0 + n''_0 -2} S_{i\ell}^{''2},& \mbox{if } n''_0\geq2,\\ 
\vspace{-6pt}\\
S_{i\ell}^{'2},& \mbox{if } n''_0=0.
\end{array}\right.
\end{equation}
A detailed description of ${\cal IZE}$ is shown in Algorithm~\ref{alg:IZR+++1}. It should be noted that $\eta_E = \frac{1}{2}\left[(2\beta_E)^{-2/n}-1\right]$ when $g_{n}(\eta_E)=\beta_E$ and $c=1$, where $n=n'_0+n''_0-2$ when $n''_0\geq 2$ and $n=n'_0-1$ when $n''_0=0$. Moreover, $\frac{\eta_E}{\eta_B}$ is equivalent to $\frac{\eta_2}{\eta_1}$ in (\ref{eq:hratio}) when $k=s=1$. 

\begin{algorithm}[!htb]
	\caption{\hspace{-0.3em}: \textbf{Procedure ${\cal IZE}$}}\label{alg:IZR+++1}
    {\fontsize{11}{24}\selectfont
		\begin{algorithmic}\single 	
						\State [{\bf Setup}:]
\begin{itemize}\small
				\item[] Choose integers $n'_0\geq 2$ and $n''_0\geq0$, $n''_0\neq 1$, confidence level $0<1-\alpha<1$, $c\in \mathbb{N}^+$, and $\Theta = \{1,2, \ldots, k\}$.
                \item[] Set $0.5 < \nu \leq 1$, tolerance levels $\epsilon_\ell>0$, and thresholds $q_{\ell}$  for $\ell=1,2,\ldots,s$.
                \item[] Set $M = \{1,2,\ldots, k\}$ and $F = \emptyset$.
			\end{itemize}
            \State [{\bf Initialization}:]
			\For{$i \in\Theta$}
                \begin{itemize}\small
				\item[] Obtain observations, $Y'_{i \ell 1}, Y'_{i \ell 2}, \ldots, Y'_{i \ell n'_0}, Y_{i \ell 1}, Y_{i \ell 2}, \ldots, Y_{i \ell n''_0} \overset{iid}{\sim} Y_{i\ell}$ and compute $\widetilde{S}_{i \ell}^2$ using equation (\ref{eq:samplevariance}) for $\ell=1,2,\ldots, s$.
                \item[] For each constraint $\ell=1,2,\ldots, s$, set $T_{\ell}=2$ and relaxed tolerance sets, ${\cal E}_{{\cal U}i \ell} = {\cal E}_{{\cal D}i \ell} = \{\epsilon_{i\ell}^{(1)}, \epsilon_{i\ell}^{(2)} \}$, where $\epsilon_{i\ell}^{(1)}$ is defined in equations (\ref{eq:delta})--(\ref{eq:xihat}) and $\epsilon_{i\ell}^{(2)}=\epsilon_\ell$.
                \item[] Set $r_i = n''_0$ and set $Z_{i \ell}=Z_{{\cal U}i \ell}^{(\epsilon)} = Z_{{\cal D}i \ell}^{(\epsilon)}=0$ for all $\epsilon \in {\cal E}_{{\cal U}i\ell} (= {\cal E}_{{\cal D}i\ell})$ and $\ell \in {\rm ON}_i$.
                \end{itemize} \EndFor
                
                \noindent If $n''_0\geq 2$, calculate $h_E^2= 2c\eta_E (n'_0+n''_0-2)$, where $\eta_E>0$ satisfies
                \begin{equation*}
                    g_{n'_0+n''_0-2}(\eta_E)=\beta_E =  \left\{
\begin{array}{ll}
  \left[1-(1-\alpha)^{1/k}\right]/(2s), & \mbox{ if systems are independent}; \\
  \alpha/(2ks), & \mbox{ otherwise}.
\end{array}\right.
                \end{equation*}
                Otherwise, if $n''_0 = 0$, calculate $h_E^2= 2c\eta_E (n'_0-1)$, where $\eta_E>0$ satisfies $g_{n'_0-1}(\eta_E)=\beta_E$.
				
            \noindent If $n''_0=0$, go to [{\bf Stopping Condition}], otherwise go to [{\bf Feasibility Check}].
			\State [{\bf Feasibility Check}:] Same as Algorithm~\ref{alg:IZR} except replacing $S_{i \ell}^2$ and $h^2$ by $\widetilde{S}_{i \ell}^2$ and $h_E^2$.
            \State[{\bf Stopping Condition}:]
			Same as Algorithm~\ref{alg:IZR}.
		\end{algorithmic}
	}
\end{algorithm}

\subsection{Statistical Guarantee of ${\cal IZE}$}
\label{sec4.2}

Similar to Lemma~\ref{lemma:IZR1}, we provide the following lemma regarding ${\cal IZE}$.

\begin{lemma} \label{lemma:IZE1}
Under Assumption~\ref{assump:normal}, when $y_{i\ell} \geq q_\ell+\epsilon_{\ell}$, ${\cal IZE}$ guarantees
\begin{equation} \label{eq:LowerProb2} 
\Pr\left\{{\cal U}^{(\epsilon)}_{i\ell}(1)\right\} \leq \beta_E, \; \mbox{ for any } \epsilon \in \{\epsilon_\ell^{(1)}, \epsilon_\ell^{(2)}\},
\end{equation} 
and
when $y_{i\ell} \leq q_\ell-\epsilon_{\ell}$, ${\cal IZE}$ guarantees
\begin{equation} \label{eq:UpperProb2} 
\Pr\left\{{\cal D}^{(\epsilon)}_{i\ell}(-1)\right\}  \leq \beta_E, \; \mbox{ for any } \epsilon \in \{\epsilon_\ell^{(1)}, \epsilon_\ell^{(2)}\}.
\end{equation} \vspace{-36pt}
\end{lemma}
\begin{proof}
For a general output process ${\bf G}_{i \ell} = \{G_{i\ell j} , j = 1, 2, . . .\}$, let $T_{{\bf G}_{i \ell}}$ represent the stage at which $\sum_{j=1}^r G_{i\ell j}$ exits the triangular region defined by $R(r; \epsilon, h_E^2, \widetilde{S}_{i \ell}^2)$ for the first time after $n''_0$, i.e.,
\begin{equation*}
    T_{{\bf G}_{i \ell}} = \min\left\{r:r\geq n''_0\; \mbox{ and } -R(r; \epsilon, h_E^2, \widetilde{S}_{i \ell}^2) < \sum_{j=1}^r G_{i\ell j} < R(r;\epsilon, h_E^2, \widetilde{S}_{i \ell}^2)\;\mbox{is violated} \right\}.
\end{equation*}
First consider the case of $y_{i\ell} \geq q_\ell+\epsilon_{\ell}$ and define $\Upsilon_{i\ell j} =  Y_{i \ell j}-(q_\ell+\epsilon_\ell-\epsilon)$ for all $j$. Then,
\begin{eqnarray*} 
\Pr\left\{{\cal U}^{(\epsilon)}_{i\ell}(1)\right\}
& = & \Pr\left\{ \left(\sum_{j=1}^{T_{{\bf \Upsilon}_{i \ell}}} \Upsilon_{i \ell j}\right) \leq \min \left\{0, -\frac{h_E^2 \widetilde{S}_{i\ell}^2}{2c\epsilon} + \frac{\epsilon}{2c} T_{{\bf \Upsilon}_{i \ell}} \right\}  \right\}\\
& = & \E \left[\Pr\left\{\sum_{j=1}^{T_{{\bf \Upsilon}_{i \ell}}} \left(Y_{i \ell j} - q_\ell - \epsilon_\ell+\epsilon\right) \leq \min \left\{0, -\frac{h_E^2 \widetilde{S}_{i\ell}^2}{2c\epsilon } + \frac{\epsilon}{2c} T_{{\bf \Upsilon}_{i \ell}} \right\}  \;\middle|\; \widetilde{S}_{i\ell}, \epsilon \right\} \right].
\end{eqnarray*}
Now, we define $X_{i\ell j} = \left(Y_{i \ell j} - q_\ell - \epsilon_\ell+\epsilon\right) - \left(y_{i \ell} - q_\ell - \epsilon_\ell\right)$ so that $\E[X_{i\ell j}|\epsilon]=\epsilon$. Notice that $y_{i\ell} - q_\ell-\epsilon_{\ell}\geq0$ and therefore $X_{i\ell j} \leq Y_{i \ell j} - q_\ell - \epsilon_\ell+\epsilon$. This implies that $\sum_{j=1}^r X_{i\ell j}$ is more likely exit a given continuation region through the lower boundary than $\sum_{j=1}^{r} \left(Y_{i \ell j} - q_\ell - \epsilon_\ell+\epsilon\right)$. Thus,
\begin{eqnarray*} 
\Pr\left\{{\cal U}^{(\epsilon)}_{i\ell}(1)\right\}&\leq& \E \left[\Pr\left\{\sum_{j=1}^{T_{{\bf X}_{i \ell}}} X_{i \ell j} \leq \min \left\{0, -\frac{h_E^2 \widetilde{S}_{i\ell}^2}{2c\epsilon} + \frac{\epsilon}{2c} T_{{\bf X}_{i \ell}} \right\}  \;\middle|\; \widetilde{S}_{i\ell}, \epsilon \right\} \right]\\
&=&\E \left[\Pr\left\{\sum_{j=1}^{T_{{\bf X}_{i \ell}}} \frac{X_{i \ell j}}{\sigma_{i\ell}} \leq \min \left\{0, -\frac{h_E^2 \widetilde{S}_{i\ell}^2}{2c\epsilon \sigma_{i\ell}} + \frac{\epsilon}{2c \sigma_{i\ell}} T_{{\bf X}_{i \ell}} \right\}  \;\middle|\; \widetilde{S}_{i\ell}, \epsilon \right\} \right].
\end{eqnarray*}
If we consider $n''_0\geq2$, $(n'_0+n''_0-2) \frac{\widetilde{S}_{i\ell}^2}{\sigma_{i\ell}^2}=\frac{\left\{(n'_0-1) S_{i\ell}^{'2}+(n''_0-1) S_{i\ell}^{''2}\right\}}{\sigma_{i \ell}^2}$ has a chi-squared distribution with $n'_0+n''_0-2$ degrees of freedom. One important finding is that $\frac{1}{n''_0} \sum_{j=1}^{n''_0} Y_{i\ell j}$ is independent of $S_{i\ell}^{''2}$ under Assumption~\ref{assump:normal}, and it is also independent of $\bar{Y}'_{i\ell}$ and $S_{i\ell}^{'2}$. In this case, $\sum_{j=1}^{n''_0}Y_{i\ell j}$ is independent of $\widetilde{S}_{i\ell}^2$ and $\epsilon$ which is a function of $\bar{Y}'_{i\ell}$. If $n''_0=0$, then $(n'_0-1) \frac{\widetilde{S}_{i\ell}^2}{\sigma_{i\ell}^2}$ has a chi-squared distribution with $n'_0-1$ degrees of freedom. The observations we take after $n''_0$ do not depend on $\widetilde{S}_{i\ell}^2$ and $\epsilon$ as we assume that the $Y_{i\ell j}$ are iid; that is, the infinite sample path after $n''_0$ does not depend on $\widetilde{S}_{i\ell}^2$ and $\epsilon$. Also, notice that the conditional distribution of $\sum_{j=1}^r \frac{X_{i\ell j}}{\sigma_{i\ell}}$ given $\epsilon$ is identical to that of a Brownian motion process with drift $\frac{\epsilon}{\sigma_{i \ell}}$ and variance parameter $1$, defined by ${\cal W}(t, \frac{\epsilon}{\sigma_{i \ell}})$, at times $r=t=n''_0, n''_0+1, \ldots$.
Let 
\begin{equation*}
  a = \frac{h_E^2 \widetilde{S}_{i\ell}^2}{2c\epsilon \sigma_{i\ell}} =\left\{
\begin{array}{ll}
 \frac{\eta_E (n'_0+n''_0-2) \widetilde{S}_{i\ell}^2}{\epsilon \sigma_{i\ell}} ,& \mbox{if } n''_0\geq2,\\ 
\vspace{-6pt}\\
\frac{\eta_E (n'_0-1) \widetilde{S}_{i\ell}^2}{\epsilon \sigma_{i\ell}} ,& \mbox{if } n''_0=0,\\ 
\end{array}\right.  
\end{equation*}
and $\gamma = \frac{\epsilon}{2c\sigma_{i \ell}}$.  By an argument similar to that in Lemmas 3 \citep{Fabian} and 4 \citep{Jennison1980} and the proof of Theorem 1 in \cite{ak},
\begin{eqnarray*} 
\lefteqn{\E \left[\Pr\left\{\sum_{j=1}^{T_{{\bf X}_{i \ell}}} \frac{X_{i \ell j}}{\sigma_{i\ell}} \leq \min \left\{0, -\frac{h_E^2 \widetilde{S}_{i\ell}^2}{2c\epsilon \sigma_{i\ell}} + \frac{\epsilon}{2c \sigma_{i\ell}} T_{{\bf X}_{i \ell}} \right\}  \;\middle|\; \widetilde{S}_{i\ell}, \epsilon \right\} \right]}\nonumber\\
&=& \E \left[\Pr\left\{{\cal W}(T_{{\bf X}_{i \ell}}, \frac{\epsilon}{\sigma_{i \ell}}) <0  \;\middle|\; \widetilde{S}_{i\ell}, \epsilon  \right\} \right]\nonumber\\
& \leq&\E \left[ \sum_{\ell = 1}^c (-1)^{\ell +1} \left( 1- \frac{1}{2} {\mathbb I}(\ell = c) \right) \exp\left\{-2 a \gamma (2c - \ell) \ell \right\} \right]. \label{eq:lemma31}
\end{eqnarray*}
Then, if $n''_0\geq2$,
\begin{eqnarray*}
-2a\gamma (2c-\ell) \ell &=&  -2 \times \frac{\eta_E (n'_0+n''_0-2) \widetilde{S}_{i\ell}^2}{\epsilon \sigma_{i\ell}} \times \frac{\epsilon}{2c\sigma_{i \ell}} \times(2c-\ell) \ell\\
 &=&  - \frac{\eta_E(2c-\ell) \ell}{c} \times \frac{\left\{(n'_0-1) S_{i\ell}^{'2}+(n''_0-1) S_{i\ell}^{''2}\right\}}{\sigma_{i \ell}^2}
\end{eqnarray*}
and if $n''_0=0$,
\begin{eqnarray*}
-2a\gamma (2c-\ell) \ell &=&  -2 \times \frac{\eta_E (n'_0-1) \widetilde{S}_{i\ell}^2}{\epsilon \sigma_{i\ell}} \times \frac{\epsilon}{2c\sigma_{i \ell}}\times (2c-\ell) \ell\\
 &=&  - \frac{\eta_E(2c-\ell) \ell}{c} \times \frac{(n'_0-1) S_{i\ell}^{'2}}{\sigma_{i \ell}^2}.
\end{eqnarray*}
Since $\E[\exp\{\omega \chi_\upsilon^2\}] = (1-2\omega)^{-\upsilon/2}$ for $\omega< 1/2$ when $\chi_\upsilon^2$ is a chi-squared random variable with $\upsilon$ degrees of freedom, 
\begin{eqnarray*}
\lefteqn{\E \left[ \sum_{\ell = 1}^c (-1)^{\ell +1} \left( 1- \frac{1}{2} {\mathbb I}(\ell = c) \right) \exp\left\{-2 a \gamma (2c - \ell) \ell \right\}  \right]} \nonumber \\
& = & \left\{
\begin{array}{ll}
\sum_{\ell = 1}^c (-1)^{\ell +1} \left( 1- \frac{1}{2} {\mathbb I}(\ell = c) \right) \times \left( 1 + \frac{2\eta_E (2c-\ell) \ell}{c}  \right)^{-(n'_0+n''_0-2)/2}  ,& \mbox{if } n''_0\geq2,\\ 
\vspace{-6pt}\\
\sum_{\ell = 1}^c (-1)^{\ell +1} \left( 1- \frac{1}{2} {\mathbb I}(\ell = c) \right) \times \left( 1 + \frac{2\eta_E (2c-\ell) \ell}{c}  \right)^{-(n'_0-1)/2}   ,& \mbox{if } n''_0=0,
\end{array}\right.  \\
&=&  \beta_E,
\end{eqnarray*}
where the last equality holds since $\eta_E$ is the solution of $g_{n'_0+n''_0-2}(\eta_E) =  \beta_E$ if $n''_0\geq2$ and  $g_{n'_0-1}(\eta_E) =  \beta_E$ if $n''_0=0$, see Algorithm~\ref{alg:IZR+++1}. Thus, the result in (\ref{eq:LowerProb2}) holds. 

Similarly, when $y_{i\ell} \leq q_\ell-\epsilon_{\ell}$, the result in (\ref{eq:UpperProb2}) can be derived if $X_{i\ell j}$ is replaced by $X'_{i\ell j} = \left(q_\ell - \epsilon_\ell+\epsilon - Y_{i\ell j}\right) - \left( q_\ell - \epsilon_\ell - y_{i \ell}\right)$, so that $\E[X'_{i\ell j}|\epsilon]=\epsilon$. Thus,
\begin{equation*} 
\Pr\left\{{\cal D}^{(\epsilon)}_{i\ell}(-1)\right\}\leq \E \left[\Pr\left\{\sum_{j=1}^{T_{{\bf X'}_{i \ell}}} \frac{X'_{i \ell j}}{\sigma_{i\ell}} \leq \min \left\{0, -\frac{h_E^2 \widetilde{S}_{i\ell}^2}{2c\epsilon \sigma_{i\ell}} + \frac{\epsilon}{2c \sigma_{i\ell}} T_{{\bf X'}_{i \ell}} \right\}  \;\middle|\; \widetilde{S}_{i\ell}, \epsilon \right\} \right] \leq \beta_E.
\end{equation*}
\end{proof}
The next lemma follows from Lemma~\ref{lemma:IZE1} in the same manner as Lemma~\ref{lemma:IZR2} with $T_\ell=2$ for all $\ell$ follows from Lemma~\ref{lemma:IZR1}.

\begin{lemma} \label{lemma:IZE2}
Under Assumption~\ref{assump:normal}, when either $y_{i\ell} \geq q_\ell+\epsilon_{\ell}$ or $y_{i\ell} \leq q_\ell-\epsilon_{\ell}$, ${\cal IZE}$ guarantees
\begin{eqnarray*}
\Pr \left\{{\rm ICD}_{i\ell} \right\}&\leq&  2\beta_E.
\end{eqnarray*}
\end{lemma}

Finally, we have the following theorem.

\begin{theorem}\label{thm:IZE}
Under Assumption~\ref{assump:normal}, ${\cal IZE}$ guarantees
\begin{equation*}
\Pr \left\{ \bigcap_{i=1}^k \bigcap_{\ell =1}^s {\rm CD}_{i \ell} \right\} \geq \ 1-\alpha.
\end{equation*}
\end{theorem}
\begin{proof}
From Lemma \ref{lemma:IZE2}, the result directly holds by replacing $\beta$ by $\beta_E$ in the proof of Theorem~\ref{thm:IZR}.
\end{proof}

\begin{remark}
    We used a specific formula for $\epsilon^{(1)}_{i\ell}$ in equations (\ref{eq:delta})--(\ref{eq:xihat}) and showed the statistical validity of ${\cal IZE}$ with this formula. Nevertheless, one may define $\epsilon^{(1)}_{i\ell}$ differently as a function of  $\bar{Y}'_{i\ell}$ and/or $\widetilde{S}_{i\ell}^2$ and obtain the same guarantee of statistical validity. Thus Theorem~\ref{thm:IZE} proves the statistical validity of a broad class of procedures employing a relaxed indifference-zone and estimation to achieve increased efficiency in practical settings.
\end{remark}

\section{Numerical Results}
\label{sec5}

In this section, we present numerical results to demonstrate the performance of our proposed procedures ${\cal IZR}$ and ${\cal IZE}$ relative to the existing procedure ${\cal F}_B$ due to \citet{bk:constraint}.
All the numerical results are based on 10,000 macro replications with $\alpha=0.05$ and we report the average number of observations (OBS) and estimated probability of correct decision (PCD). 
For Procedure ${\cal IZR}$, we particularly focus on setting $T_\ell=2$ and $\xi\in \{2, 3\}$ (i.e., $\{\epsilon_\ell^{(1)}=2\epsilon_\ell, \epsilon_\ell^{(2)}=\epsilon_\ell\}$ for $\xi=2$ and $\{\epsilon_\ell^{(1)}=3\epsilon_\ell, \epsilon_\ell^{(2)}=\epsilon_\ell\}$ for $\xi=3$). The choice of $T = 2$ follows the general recommendation provided in Section 3.4. The selection of $\xi =2$ is motivated by the discussion in Section \ref{sec3_1}, with $\xi=3$ being a larger choice. 

We begin by describing the experimental configurations in Section~\ref{sec:ExpConfig}. Section~\ref{sec:ImplementationParams} presents a detailed discussion and results on various implementation parameters for the three competing procedures. The statistical validity and efficiency of our proposed procedures are demonstrated in Sections~\ref{sec:StatisticalValidity} and~\ref{sec:Efficiency}, respectively. Section~\ref{sec:Inventory} illustrates the performance of our proposed procedures through an inventory control example. Finally, we provide a summary of the overall findings from our experimental results, along with general recommendations for selecting the implementation parameters for our proposed procedures, in Section \ref{sec:Discussion}.

\subsection{Experimental Configurations}
\label{sec:ExpConfig}
In this section, we describe the mean and variance configurations used in the cases we test. It should be noted that we set the threshold constant $q_\ell=0$ and tolerance level $\epsilon_\ell=0.02$, where $\ell=1,\ldots,s$, for all procedures considered.
To describe the mean configurations, we place systems into three groups based on their feasibility status using parameters $0 \leq \underline{b} \leq \overline{b} \leq k$. Specifically, we have
\begin{itemize}
	\item $\underline{b}$ systems with all constraints feasible,
	\item $(\overline{b} - \underline{b})$ systems with $m$ feasible and $(s - m)$ infeasible constraints, where $0\leq m\leq s$, and
	\item $(k - \overline{b})$ systems with all constraints infeasible.  	
\end{itemize}
Note that the choices of $\underline{b}, \overline{b}$, and $m$ impact how fast feasibility checks can be completed for the $k$ systems.

We now present the two mean configurations considered in our experiments: Concentrated Means (CM) and Scattered Means (SM). 

In the CM configuration, each system's mean is set to a fixed value that is $d_\ell\geq 0$ away from the corresponding threshold $q_\ell$ for constraint $\ell$. Specifically, the means are assigned as follows:
\begin{align*}
	y_{i\ell} = \begin{cases}
		-d_\ell, & i=1, 2, \ldots, \underline{b} \text{ and } \ell=1,\ldots,s,  \\
		-d_\ell, & i= \underline{b}+1, \underline{b}+2, \ldots, \overline{b} \text{ and } \ell=1,\ldots,m,  \\
		d_\ell, & i= \underline{b}+1, \underline{b}+2, \ldots, \overline{b} \text{ and } \ell=m+1,\ldots,s,  \\
		d_\ell, & i= \overline{b}+1, \overline{b}+2, \ldots, k \text{ and } \ell=1,\ldots,s.
	\end{cases}
\end{align*}
This configuration is considered as the ``most difficult'' case on constraint $\ell$ when $d_\ell=\epsilon_\ell$. When $0\leq d_\ell<\epsilon_\ell$, all systems are acceptable systems; when $d_\ell\geq \epsilon_\ell$, there are no acceptable systems. 

In contrast, the SM configuration sets system means to a wider range of values that vary with the system index. This setting is considered to better resemble practical settings. The means are given by:
\begin{align*}
	y_{i\ell} = \begin{cases}
		-(\underline{b}-i+1)d_\ell, & i=1, 2, \ldots, \underline{b} \text{ and } \ell=1,\ldots,s,  \\
		-(i-\underline{b})d_\ell, & i= \underline{b}+1, \underline{b}+2, \ldots, \overline{b} \text{ and } \ell=1,\ldots,m,  \\
		(i-\underline{b})d_\ell, & i= \underline{b}+1, \underline{b}+2, \ldots, \overline{b} \text{ and } \ell=m+1,\ldots,s,  \\
		(i-\overline{b})d_\ell, & i= \overline{b}+1, \overline{b}+2, \ldots, k \text{ and } \ell=1,\ldots,s.
	\end{cases}
\end{align*}
Note that when $0< d_\ell <\epsilon_\ell$, the number of acceptable systems with respect to constraint $\ell$ depends on the choices of $\underline{b}, \overline{b}, d_\ell$, and $\epsilon_\ell$ as explained below. 
\begin{itemize}
	\item For the first $\underline{b}$ systems, which are feasible for all constraints, system $i$ is acceptable on constraint $\ell$ only if $-(\underline{b} - i + 1)d_\ell > -\epsilon_\ell$ holds. This is equivalent to $\max \{0, \underline{b}+1-  \epsilon_\ell/d_\ell  \}< i\leq \underline{b}$ (since we also have $1\leq i\leq \underline{b}$ and $i$ is an integer). Therefore, $\min \{ \lceil \epsilon_\ell/d_\ell \rceil -1, \underline{b}\}$ systems are acceptable among the first $\underline{b}$ systems when $0< d_\ell <\epsilon_\ell$.
	\item Among the next $(\overline{b} - \underline{b})$ systems, which have $m$ feasible and $(s - m)$ infeasible constraints, system $i$ is acceptable on constraint $\ell$ only if $(i - \underline{b})d_\ell < \epsilon_\ell$, which is equivalent to $\underline{b}+1\leq i < \min \{\underline{b}+\epsilon_\ell/d_\ell, \overline{b}+1\}$ (as we also have $\underline{b}+1\leq i\leq \overline{b}$ and $i$ is an integer). Thus, $\min \{ \lceil \epsilon_\ell/d_\ell \rceil - 1, \overline{b}-\underline{b}\}$ systems are acceptable among the $(\overline{b} - \underline{b})$ systems when $0< d_\ell <\epsilon_\ell$.
	\item Similarly, among the last $(k - \overline{b})$ systems, which are infeasible on all constraints, system $i$ is acceptable only if $(i - \overline{b})d_\ell < \epsilon_\ell$, or equivalently, $\overline{b}+1\leq i< \min\{\overline{b}+\epsilon_\ell/d_\ell, k+1\}$ (as we also need to ensure $\overline{b} + 1\leq i\leq k$ and $i$ is an integer). Therefore, $\min \{ \lceil \epsilon_\ell/d_\ell \rceil - 1, k-\overline{b}\}$ systems are acceptable among the last $(k-\overline{b})$ systems when $0< d_\ell <\epsilon_\ell$.
\end{itemize}
In addition, when $d_\ell=0$, all systems are acceptable systems; when $d_\ell\geq \epsilon_\ell$, there are no acceptable systems.

We also consider three variance configurations: Constant Variance (CV), Increasing Variance (IV), and Decreasing Variance (DV). In all cases, we set $\sigma^2 =1$ as the average variance value.
\begin{itemize}
	\item {\bf IV:} We consider two variants:
	\begin{itemize}
		\item \textbf{IV-C (Constraint-based):} Variance increases with the constraint index. We set $\sigma_{i\ell}^2 = 2[\ell/(s+1)]\sigma^2$ for all $i = 1, \ldots, k$ and $\ell = 1, \ldots, s$.
		\item \textbf{IV-S (System-based):} Variance increases with the system index. We set $\sigma_{i\ell}^2 = 2[i/(k+1)]\sigma^2$ for all $i = 1, \ldots, k$ and $\ell = 1, \ldots, s$.
	\end{itemize}
	\item {\bf DV:} We consider two similar variants:
	\begin{itemize}
		\item \textbf{DV-C (Constraint-based):} Variance decreases with the constraint index. We set $\sigma_{i\ell}^2 = 2[(s - \ell+1)/(s+1)]\sigma^2$ for all $i = 1, \ldots, k$ and $\ell = 1, \ldots, s$.
		\item \textbf{DV-S (System-based):} Variance decreases with the system index. We set $\sigma_{i\ell}^2 = 2 [(k - i+1)/(k+1)]\sigma^2$ for all $i = 1, \ldots, k$ and $\ell = 1, \ldots, s$.
	\end{itemize}
	\item {\bf CV:} All systems have identical variance across all constraints as $\sigma^2$, where $i=1,\ldots,k$ and $\ell=1,\ldots,s$. 
\end{itemize}	

Finally, we simulate the systems and constraints independently (except for Section \ref{sec:Inventory} where the constraints are correlated). This assumption is justified for the systems, as the use of common random numbers is not expected to significantly affect the feasibility check as each system is evaluated against fixed threshold values rather than being compared to other systems (see \citet{Zhou2022}). On the other hand, \citet{bk:constraint} demonstrate that the correlation across constraints has minimal impact on the feasibility check of their proposed procedure ${\cal F}_B$. We expect a similar phenomenon to hold for our proposed procedures, ${\cal IZR}$ and ${\cal IZE}$, and therefore focus exclusively on the case where constraints are independent.
Finally, when systems are correlated, a more conservative setting of the implementation parameter $\eta$ (or $\eta_E$) may lead to slightly worse performance than in the independent case. To provide results that reflect more practical scenarios, we adopt the setting of $\eta$ (or $\eta_E$) corresponding to the correlated case throughout our experiments.

\subsection{Implementation Parameters}
\label{sec:ImplementationParams}  

In this section, we describe the choices for the implementation parameters $n_0'$, $n_0''$, and $\nu$ for ${\cal IZE}$ and $n_0$ for ${\cal F}_B$ and ${\cal IZR}$. 

\paragraph{Parameters $n_0'$ and $n_0''$.}

${\cal IZE}$ uses $n_0'>0$ initial observations to estimate $|y_{i\ell}-q_\ell|$ and $\sigma_{i\ell}^2$ and determine the tolerance level $\epsilon_{i\ell}^{(1)}$. These $n_0'$ observations are discarded and not used later in the feasibility check. Additionally, ${\cal IZE}$ uses $n_0''\geq 0$ observations to estimate $\sigma_{i\ell}^2$, which are subsequently included in the feasibility check.

Recall that Procedures ${\cal F}_B$ and ${\cal IZR}$ also require $n_0$ initial samples for estimating $\sigma_{i\ell}^2$, and those samples are included in the subsequent feasibility check. To ensure a fair comparison, we set the total number of initial samples in ${\cal IZE}$ as $n_0' + n_0'' = n_0$, matching the initial sample size used by ${\cal F}_B$ and ${\cal IZR}$. 
However, it is not clear how $n_0'$ and $n_0''$ should be compared given that both are used to estimate $\sigma_{i\ell}^2$ but only $n_0'$ is used to estimate $|y_{i\ell}-q_\ell|$ while only $n_0''$ is used for feasibility check. Hence, taking into account $n_0'> 0$, $n_0''\geq 0$ ($n_0''\ne 1$) and setting $n_0'+n_0''=n_0=20$, we evaluate four combinations: $(n_0', n_0'') \in \{(5, 15), (10, 10), (15, 5), (20,0)\}$.

\paragraph{Parameter $\nu$.} Among the three competing procedures, the parameter $\nu$ is only required by ${\cal IZE}$. This parameter is used to determine $\epsilon_{i\ell}^{(1)}$. As we have $0.5<\nu\leq 1$, we particularly focus on $\nu \in \{0.6, 0.8, 1\}$. 

We set the number of systems to $k = 99$, with $\underline{b} = 33$ and $\overline{b} = 66$. We also set $s=4$ constraints and $m=2$, i.e., systems $\underline{b}+1, \ldots, \overline{b}$ have $m=2$ feasible constraints and $(s-m)=2$ infeasible constraints. 
Since the SM configuration better reflects practical scenarios, we focus on this setting (except for the results regarding statistical validity; see Section \ref{sec:StatisticalValidity}) and vary the separation parameters as $d_\ell \in \{0.02, 0.5\}$ for $\ell = 1, \ldots, s$. Recall that we set $\epsilon_\ell = 0.02$ for all $\ell$, so $d_\ell = 0.02$ leads to some systems in the slippage configuration, which is typically considered the most challenging scenario for statistical validity. In contrast, $d_\ell = 0.5$ represents a less difficult and more practically representative setting.
Table \ref{tab:DiffCombInitSample_IZRE} presents the experimental results for $(n_0', n_0'')\in \{(5, 15), (10, 10), (15, 5), (20, 0)\}$ and $\nu\in \{0.6, 0.8, 1\}$. In addition, experimental results that include acceptable systems, where we set $d_\ell = 0.01$, are presented in  Appendix~\ref{sec:IZE_DiffComb_d_n0_Additional}.

\begin{table}[h!]
	\centering
	\resizebox{\columnwidth}{!}{
		\begin{tabular}{ c || c | c cccc | c cccc}
			\toprule
			& & \multicolumn{5}{c|}{$d_\ell=0.02$} & \multicolumn{5}{c}{$d_\ell=0.5$}   \\  
			$\nu$ & $(n_0', n_0'')$ & CV & IV-C & DV-C & IV-S & DV-S & CV & IV-C & DV-C & IV-S & DV-S  \\  \midrule
			\multirow{ 8}{*}{0.6} 
			& $(5, 15)$ & 240772 & 276892 & 223198 & 179481 & 306835 & 3266 & 3455 & 3224 & 2837 & 3720 \\ 
			& & (1.000) & (1.000) & (1.000) & (1.000) & (0.999) & (1.000) & (1.000) & (1.000) & (1.000) & (1.000) \\   \cline{2 - 12}
			& $(10, 10)$ & 229394 & 263549 & 213111 & 172549 & 290828 & 3169 & 3364 & 3117 & 2799 & 3582 \\  
			& & (1.000) & (1.000) & (1.000) & (1.000) & (1.000) & (1.000) & (1.000) & (1.000) & (1.000) & (1.000) \\    \cline{2 - 12}
			& $(15, 5)$ & 224429 & 257264 & 208317 & 169776 & 283473 & 3161 & 3342 & 3096 & 2810 & 3537 \\ 
			& & (1.000) & (1.000) & (0.999) & (0.999) & (1.000) & (1.000) & (1.000) & (1.000) & (1.000) & (1.000) \\   \cline{2 - 12}
			& $(20,0)$ & 214633 & 245522 & 198905 & 164013 & 269946 & 3267 & 3449 & 3197 & 2957 & 3623 \\ 
			& & (1.000) & (0.999) & (1.000) & (1.000) & (0.999) & (1.000) & (1.000) & (1.000) & (1.000) & (1.000)
			\\  \midrule
			\multirow{ 8}{*}{0.8} 
			& $(5, 15)$ & 247724 & 286519 & 230509 & 183466 & 320074 & 3169 & 3392 & 3174 & 2723 & 3697 \\ 
			& & (1.000) & (0.999) & (1.000) & (1.000) & (1.000) & (1.000) & (1.000) & (1.000) & (1.000) & (1.000) \\   \cline{2 - 12}
			& $(10, 10)$ & 230623 & 266345 & 214217 & 172300 & 297261 & 2973 & 3153 & 2966 & 2633 & 3346  \\  
			& & (1.000) & (0.999) & (0.999) & (1.000) & (0.999) & (1.000) & (1.000) & (1.000) & (1.000) & (1.000) \\    \cline{2 - 12}
			& $(15, 5)$ & 221805 & 256277 & 206108 & 166740 & 284332 & 2917 & 3097 & 2891 & 2630 & 3254 \\ 
			& & (1.000) & (1.000) & (1.000) & (0.999) & (1.000) & (1.000) & (1.000) & (1.000) & (1.000) & (1.000) \\   \cline{2 - 12}
			& $(20,0)$ & 209563 & 241632 & 194322 & 158257 & 267407 & 3026 & 3190 & 2984 & 2776 & 3327 \\ 
			& & (0.999) & (1.000) & (1.000) & (1.000) & (1.000) & (1.000) & (1.000) & (1.000) & (1.000) & (1.000)
			\\  \midrule
			\multirow{ 8}{*}{1} 
			& $(5, 15)$ & 259923 & 300723 & 241671 & 191022 & 336611 & 3410 & 3707 & 3446 & 2811 & 4107 \\ 
			& & (1.000) & (1.000) & (0.999) & (1.000) & (0.999) & (1.000) & (1.000) & (1.000) & (1.000) & (1.000) \\   \cline{2 - 12}
			& $(10, 10)$ & 241358 & 279959 & 224711 & 178056 & 313283 & 3075 & 3295 & 3097 & 2633 & 3561 \\  
			& & (0.999) & (0.999) & (1.000) & (1.000) & (1.000) & (1.000) & (1.000) & (1.000) & (1.000) & (1.000) \\    \cline{2 - 12}
			& $(15, 5)$ & 230734 & 266827 & 215059 & 171046 & 299728 & 2934 & 3142 & 2965 & 2596 & 3369 \\ 
			& & (0.999) & (1.000) & (1.000) & (0.999) & (1.000) & (1.000) & (1.000) & (1.000) & (1.000) & (1.000) \\   \cline{2 - 12}
			& $(20,0)$ & 217018 & 251120 & 201783 & 161885 & 280439 & 3002 & 3170 & 3001 & 2712 & 3377  \\ 
			& & (1.000) & (0.999) & (1.000) & (1.000) & (0.999) & (1.000) & (1.000) & (1.000) & (1.000) & (1.000)
			\\ 
			\bottomrule    
	\end{tabular}}
	\caption{Estimated OBS and PCD (in parentheses) of ${\cal IZE}$ with respect to different combinations of $(n_0', n_0'')$ and $\nu$ under the SM configuration with $k=99$ systems, $s=4$ constraints, and $d_\ell \in \{0.02, 0.5\}$}
	\label{tab:DiffCombInitSample_IZRE}
\end{table}

Based on the results in Table \ref{tab:DiffCombInitSample_IZRE}, with a small $n_0'$, we observe that a smaller value of $\nu$ (e.g., $\nu=0.6$) tends to perform better when $d_\ell$ is small, whereas a moderate value (e.g., $\nu = 0.8$) generally yields the best performance when $d_\ell$ is large. This trend is primarily driven by the density of system means: when system means are closely spaced, as in the case of $d_\ell = 0.02$, a smaller $\nu$ is more effective in conducting feasibility checks. In contrast, when the means are more widely separated (e.g., $d_\ell = 0.5$), a larger $\nu$ becomes more suitable.

We also find that the combination $(n_0', n_0'') = (20, 0)$ performs best when $d_\ell = 0.02$, while $(n_0', n_0'') = (15, 5)$ performs best when $d_\ell = 0.5$. In addition, among the values of $(n_0', n_0'')$ considered, the fact that $n_0'>10$ yields better performance suggests that higher importance is put on estimating $|y_{i\ell}-q_\ell|$ to choose reasonable $\epsilon_{i\ell}^{(1)}$ than on having more observations for performing feasibility checks. This tendency is more pronounced when systems are dense.

By comparing results across different variance configurations, we identify the following three key findings (see detailed discussion in Appendix \ref{sec:IZE_DiffVarConfig_Additional}). 
First, IV-C requires more observations than DV-C and CV falls in between. 
Second, the DV-S configuration requires more observations than IV-S, with CV falling in between. 
Finally, we note that the distinction between IV-C and DV-C only affects the systems with indices $\underline{b} + 1 \leq i \leq \overline{b}$, whereas the difference between IV-S and DV-S impacts all systems. As a result, the relative impact of variance configurations indexed by systems versus those indexed by constraints depends on the specific values of $\underline{b}$ and $\overline{b}$. In the particular setting considered here, we observe that DV-S (IV-S) results in the highest (lowest) overall observation requirement among the five variance configurations.

To decide on values of $\nu, n_0'$, and $n_0''$ to use in the remaining sections based on our experiments, we primarily focus on $d_\ell=0.5$ as it is more reflective of reality. We find that choosing a medium value $\nu=0.8$ works best in general. Moreover, $(n_0', n_0'') = (15, 5)$ provides the best performance across all variance configurations when $d_\ell=0.5$ (however, $(n_0', n_0'') = (20, 0)$ performs better than $(n_0', n_0'') = (15, 5)$ when $d_\ell=0.02$). 
We adopt the configuration of $\nu=0.8$ and $n_0'=n_0-5$ in the remainder of this section and also focus on the CV variance configuration.

\paragraph{Parameter $n_0$.} Since ${\cal IZE}$ allocates a portion of the initial $n_0$ observations for estimating $|y_{i\ell}-q_\ell|$ and not exclusively for estimating $\sigma_{i\ell}^2$, while ${\cal F}_B$ and ${\cal IZR}$ use all $n_0$ initial observations for estimating $\sigma_{i\ell}^2$ only, increasing $n_0$ may improve the relative performance of ${\cal IZE}$. To examine the impact of $n_0$, we evaluate all three procedures under varying values of $n_0 \in \{10, 20, \ldots, 50\}$. For ${\cal IZE}$, we fix the allocation as $(n_0', n_0'') = (n_0 - 5, 5)$ and $\nu=0.8$ for consistency and comparability.

The results are presented in Table~\ref{tab:DiffInitSample}, where we again consider $k = 99$ systems, $s = 4$ constraints, $\underline{b}=33, \overline{b}=66, m=2$, and focus on the SM and CV configurations. Additional results for $n_0 > 50$ and other choices of $n_0'$ and $\nu$ are provided in Appendix \ref{sec:FB_IZR_Diff_n0_Additional}.

\begin{table}[h!]
	\centering
	\resizebox{\columnwidth}{!}{
		\begin{tabular}{ c || c cc c | c cc c}
			\toprule
			& \multicolumn{4}{c|}{$d_\ell=0.02$} & \multicolumn{4}{c}{$d_\ell=0.5$}   \\  
			& ${\cal F}_B$ & ${\cal IZR}$ & ${\cal IZR}$ & ${\cal IZE}$ & ${\cal F}_B$ & ${\cal IZR}$ & ${\cal IZR}$ & ${\cal IZE}$  \\  
			$n_0$ & & $T=2, \xi=2$ & $T=2, \xi=3$ & & & $T=2, \xi=2$ & $T=2, \xi=3$ \\  \midrule
			$10$ & 587192 & 424793 & 381433 & 545038 & 27374 & 16539 & 11149 & 3879 \\
			& (1.000) & (0.999) & (1.000) & (0.999) & (1.000) & (1.000) & (1.000) & (1.000) \\   \hline
			$20$ & 328495 & 223273 & 199134 & 221806 & 15378 & 8824 & 6092 & 2917 \\ 
			& (0.999) & (1.000) & (0.999) & (1.000) & (1.000) & (1.000) & (1.000) & (1.000) \\   \hline
			$30$ & 279600 & 186483 & 165654 & 173285 & 13141 & 7722 & 5707 & 3637  \\  
			& (0.999) & (1.000) & (1.000) & (1.000) & (1.000) & (1.000) & (1.000) & (1.000) \\    \hline
			$40$ & 259272 & 171437 & 152048 & 154347 & 12395 & 7717 & 6077 & 4528 \\ 
			& (1.000) & (1.000) & (0.999) & (1.000) & (1.000) & (1.000) & (1.000) & (1.000) \\    \hline
			$50$ & 248162 & 163213 & 144637 & 144653 & 12268 & 8118 & 6691 & 5468\\ 
			& (1.000) & (0.999) & (0.999) & (1.000) & (1.000) & (1.000) & (1.000) & (1.000)  \\   		
			\bottomrule    
		\end{tabular}
	}
	\caption{Estimated OBS and PCD (in parentheses) of ${\cal F}_B, {\cal IZR}$, and ${\cal IZE}$ with $n_0\in \{10, 20, \ldots, 50\}$ under the SM configuration with $k=99$ systems and $s=4$ constraints. We consider the CV variance configuration.}
	\label{tab:DiffInitSample}
\end{table}

As shown in Table~\ref{tab:DiffInitSample}, all three procedures benefit from larger values of $n_0$ when the problem is more difficult (i.e., $d_\ell = 0.02$), as indicated by reductions in the required OBS. When the problem becomes easier (i.e., $d_\ell = 0.5$), ${\cal F}_B$ achieves the smallest OBS with the largest $n_0$, while ${\cal IZR}$ performs best with $n_0 = 40$ for $\xi = 2$ and $n_0 = 30$ for $\xi = 3$. In contrast, ${\cal IZE}$ reaches optimal performance at a smaller $n_0 = 20$.

These findings can be explained by two key factors. First, when the variance is large relative to the difference between system means and the threshold (i.e., $d_\ell = 0.02$), accurate variance estimation is crucial, and thus a larger $n_0$ is beneficial. Conversely, when system means are well separated from the threshold (e.g., $d_\ell = 0.5$), so that the relative magnitude between variance and the mean-threshold difference smaller,  increasing $n_0$ further may be unnecessary.

Second, when comparing ${\cal IZR}$ and ${\cal IZE}$ under $d_\ell = 0.5$, we observe that ${\cal IZR}$ performs best with a larger $n_0$ (i.e., $n_0=40$ for $\xi = 2$ and $n_0=30$ for $\xi=3$), whereas ${\cal IZE}$ performs best with a smaller $n_0$ (i.e., $n_0=20$). On one hand, ${\cal IZE}$ benefits from more proper selection of $\epsilon_{i\ell}^{(1)}$ (which improves with a larger $n_0'$), which suggests that increasing $n_0$ can enhance its performance. On the other hand, an unnecessarily large $n_0'$ can be detrimental. This is because once $\epsilon_{i\ell}^{(1)}$ is reasonably selected, a larger $n_0'$ leads to more observations used for estimating $|y_{i\ell}-q_\ell|$ that need to be excluded for feasibility check when many systems, especially those with means further from the constraint threshold, may be able to conclude feasibility decisions with fewer observations than $n_0$. In such cases, the procedure is forced to take more samples than necessary.
A similar trade-off applies to ${\cal IZR}$. When a larger $\xi$ is used, it may enable early termination for clearly feasible or infeasible systems. If $n_0$ is set too high, it can again result in oversampling, thereby reducing efficiency. Finally, we see that ${\cal F}_B$ benefits from the largest $n_0$. As ${\cal F}_B$ only has one tolerance level, a larger $n_0$ tends to lead to a better estimation of the continuation region (due to a more accurate variance estimation), which further improves its performance.  

In the remainder of this paper, we focus on setting $n_0=20$ as it is a common choice in the literature (e.g., \cite{bk:constraint}, \citet{Zhou2022}, \citet{Zhou2024}) and all competing procedures perform reasonably well with $n_0=20$ under $d_\ell=0.5$ (which is a more practical setting than $d_\ell=0.02$). 

\subsection{Statistical Validity}
\label{sec:StatisticalValidity}

In this section, we demonstrate the statistical validity of our proposed procedures. 
We focus on the CM configuration, which is known to represent challenging cases. Since the CM configuration does not reflect typical practical scenarios, we provide a comprehensive discussion of the efficiency of our proposed procedures in more realistic settings in Sections~\ref{sec:Efficiency} and~\ref{sec:Inventory}.
Since the presence of acceptable systems does not impact statistical validity (as such systems are always considered correctly classified regardless of their true feasibility), we exclude scenarios with acceptable systems from this analysis. We present the experimental results when considering a single system with a single constraints and multiple systems with multiple constraints in Sections \ref{sec:Validity_SingleSysSingleConstr} and \ref{sec:Validity_MultiSysMultiConstr}, respectively. 

\subsubsection{Single System with Single Constraint}
\label{sec:Validity_SingleSysSingleConstr}

We first consider a scenario with a single system (i.e., $k = \underline{b} = \overline{b} = 1$) and a single constraint, (i.e., $s=m = 1$). We evaluate the performance of ${\cal F}_B$, ${\cal IZR}$, and ${\cal IZE}$ for $d_1 \in \{0.02, 0.5\}$ and assume CV with $\sigma^2 = 1$. Recall that the tolerance level is $\epsilon_1 = 0.02$. Thus, the case where $d_1 = 0.02$ corresponds to the slippage configuration.
Table \ref{tab:Validity_k1s1} shows the experimental results. We also include additional results when $d_1\in \{0.05, 0.1, 1\}$ in Appendix \ref{sec:ValidityAdditionalResults_SingleSys}. 
\begin{table}[!h]\small
	\centering
		\begin{tabular}{  c | c | cc|c}
			\toprule
			$d_1$ & ${\cal F}_B$ & \multicolumn{2}{c|}{${\cal IZR}$} & ${\cal IZE}$  \\
			& & $T=2, \xi=2$ & $T=2, \xi=3$ &  \\ \midrule 
			0.02 & 4130 & 4545 & 4911 & 5140 \\
			& (0.949) & (0.961) & (0.957) & (0.959) \\ \hline
			0.5 & 259 & 178 & 121 & 82 \\
			& (1.000) & (1.000) & (1.000) & (1.000) \\ 
			\bottomrule
		\end{tabular}
	\caption{Estimated OBS and PCD (in parentheses) of ${\cal F}_B$, ${\cal IZR}$, and ${\cal IZE}$ under the CM configuration with $k=s=1$ and $d_1\in \{0.02, 0.5\}$. }
	\label{tab:Validity_k1s1}
\end{table}

We observe that all three procedures maintain statistical validity\footnotemark\footnotetext{Although our empirical results under this setting do not show a PCD above 0.95 for ${\cal F}_B$, we believe this is due to simulation error. When increasing the number of macro-replications to 100,000 and 1,000,000, ${\cal F}_B$ achieves PCD values of 0.9493 and 0.9503, respectively.} across the tested scenarios. When $d_1 = 0.5$, all procedures achieve a PCD of 1, whereas the PCD is much lower when $d_1=0.02$. This is because the slippage configuration represents a much more difficult scenario when concluding feasibility decisions. Among the three procedures, ${\cal F}_B$ achieves the lowest PCD when $d_1 = 0.02$, while ${\cal IZR}$ and ${\cal IZE}$ achieve comparably higher PCD values. 

Moreover, both ${\cal IZR}$ and ${\cal IZE}$ require more observations to conclude feasibility decisions compared to ${\cal F}_B$ when $d_1=0.02$. This is expected, as the implementation parameter $\eta_B$ for ${\cal F}_B$ is derived specifically for the slippage configuration, while ${\cal IZR}$ and ${\cal IZE}$ are designed to employ multiple tolerance levels to address the typical situation where not all performance measures are one tolerance level away from a constraint threshold. When $k=s=1$, the values of $\eta$ and $\eta_E$ in ${\cal IZR}$ and ${\cal IZE}$ are larger than the value of $\eta_B$ in ${\cal F}_B$, leading to larger OBS (in fact, this is the only case in all our numerical results where ${\cal F}_B$ outperforms either ${\cal IZR}$ or ${\cal IZE}$). Additionally, ${\cal IZE}$ tends to underperform ${\cal IZR}$ when $d_1 = 0.02$. This is primarily because ${\cal IZE}$ often uses a larger $\epsilon_{11}^{(1)}$ than ${\cal IZR}$, making it less likely to conclude feasibility decisions using the subroutines ${\cal F_U}$ and ${\cal F_D}$ with tolerance level $\epsilon_{11}^{(1)}$. 
However, when $d_1 = 0.5$, both ${\cal IZR}$ and ${\cal IZE}$ significantly outperform ${\cal F}_B$ in terms of efficiency. Specifically, ${\cal IZR}$ with $\xi=2$ or 3 requires only 68.7\% or 46.7\% of the observations used by ${\cal F}_B$, while ${\cal IZE}$ requires just 31.7\%.

\subsubsection{Multiple Systems with Multiple Constraints}
\label{sec:Validity_MultiSysMultiConstr}

We next consider the case with multiple systems ($k = 99$) and multiple constraints ($s = 4$). We evaluate performance for $d_\ell \in \{0.02, 0.5\}$, where $\ell=1,\ldots,s$, and consider all five variance configurations CV, IV-C, DV-C, IV-S, and DV-S. We also evaluate four different settings for $(\underline{b}, \overline{b}) \in \{(55, 77), (33, 66), (22, 44), (22, 77)\}$, capturing varying distributions of system feasibility:
\begin{itemize}
	\item $(\underline{b}, \overline{b}) = (55, 77)$: The majority of systems have all constraints feasible, with the remaining systems evenly split between fully infeasible and partially feasible (i.e., half the constraints feasible, half infeasible).
	\item $(\underline{b}, \overline{b}) = (33, 66)$: The proportions of systems that are fully feasible, fully infeasible, and partially feasible are balanced.
	\item $(\underline{b}, \overline{b}) = (22, 44)$: Most systems have all constraints infeasible, with the remaining systems split evenly between fully feasible and partially feasible.
	\item $(\underline{b}, \overline{b}) = (22, 77)$: The majority of systems are partially feasible, with the remaining systems split evenly between fully feasible and fully infeasible.
\end{itemize}

Table \ref{tab:Validity_k99s4} presents the experimental results, where we focus on setting $(\underline{b}, \overline{b})=(33, 66)$. Additional results for $d_\ell\in \{0.01, 0.1, 1\}$ and $(\underline{b}, \overline{b})\in \{(55, 77), (22, 44), (22, 77)\}$ are provided in Appendix \ref{sec:ValidityAdditionalResults_MultiSys}. 

\begin{table}[!h]\small
	\centering
		\begin{tabular}{  c||c | c | cc|c}
			\toprule
			& $d_\ell$ & ${\cal F}_B$ & \multicolumn{2}{c|}{${\cal IZR}$} & ${\cal IZE}$  \\
			& & & $T=2, \xi=2$ & $T=2, \xi=3$ &  \\ \midrule 
			\multirow{4}{*}{CV} & 0.02 & 2063068 & 1939161 & 2251000 & 2240086 \\
			& & (0.984) & (0.984) & (0.983) & (0.984) \\ \cline{2-6}
			& 0.5 & 122310 & 70524 & 47981 & 27408 \\
			& & (1.000) & (1.000) & (1.000) & (1.000) \\   \hline 
			\multirow{4}{*}{IV-C} & 0.02 & 2323070 & 2201370 & 2558362  & 2549204   \\
			& & (0.985) & (0.984) & (0.985) & (0.986) \\ \cline{2-6}
			& 0.5 & 136933 & 78920 & 53774 & 30302 \\
			& & (1.000) & (1.000) & (1.000) & (1.000) \\   \hline 
			\multirow{4}{*}{DV-C} & 0.02 & 1835631 & 1783829 & 2027587 & 2014044  \\
			& & (0.983) & (0.984) & (0.984) & (0.983) \\ \cline{2-6}
			& 0.5 & 106737 & 61605 & 41998 & 25676 \\
			& & (1.000) & (1.000) & (1.000) & (1.000) \\   \hline 
			\multirow{4}{*}{IV-S} & 0.02 & 1680646 & 1460462 & 1792751  & 1790352 \\
			& & (0.984) & (0.985) & (0.985) & (0.985) \\ \cline{2-6}
			& 0.5 & 103444 & 59514 & 40428 & 16504 \\
			& & (1.000) & (1.000) & (1.000) & (1.000)  \\   \hline 
			\multirow{4}{*}{DV-S} & 0.02 & 2446196 & 2419192 & 2708883 & 2708519 \\
			& & (0.985) & (0.985) & (0.982) & (0.985) \\ \cline{2-6}
			& 0.5 & 141151 & 81522 & 55585 & 40731 \\
			& & (1.000) & (1.000) & (1.000) & (1.000) \\   
			\bottomrule
		\end{tabular}
	\caption{Estimated OBS and PCD (in parentheses) of ${\cal F}_B, {\cal IZR}$, and ${\cal IZE}$ under the CM configuration with $k=99$ and $s=4$. }
	\label{tab:Validity_k99s4}
\end{table}

We observe that all three procedures provide statistical validity across all variance configurations. Consistent with Section~\ref{sec:Validity_SingleSysSingleConstr}, all procedures achieve a PCD of 1 when $d_1 = 0.5$, regardless of the variance setting, but exhibit a lower PCD when $d_1 = 0.02$. Although $d_1 = 0.02$, which corresponds to the slippage configuration, is considered the most difficult case, the PCD values in this setting are higher than those observed in the scenario with a single system and a single constraint (Section~\ref{sec:Validity_SingleSysSingleConstr}). This is due to the conservative setting of the parameter $\eta$ (or $\eta_E$) when multiple dependent systems or constraints or tolerance levels are considered. 
Moreover, the PCDs of the three procedures are comparable for each value of $d_\ell$ across all variance configurations.

We observe a similar pattern in terms of OBS as in Section~\ref{sec:Validity_SingleSysSingleConstr}. Specifically, ${\cal IZE}$ tends to require more observations than ${\cal F}_B$ when $d_1 = 0.02$. However, unlike in the single system case, ${\cal IZR}$ now requires fewer observations than ${\cal F}_B$ when $\xi = 2$, although it still requires more observations when $\xi = 3$. This is likely because some systems benefit from concluding feasibility decisions using a moderately larger tolerance level (i.e., $\epsilon_\ell^{(1)} = 2\epsilon_\ell$) rather than the fixed $\epsilon_\ell$ used in ${\cal F}_B$. In contrast, when $\xi = 3$, the tolerance level $\epsilon_\ell^{(1)}$ increases to $\epsilon_\ell^{(1)} = 3\epsilon_\ell$, but the feasibility decisions are less likely to terminate at this stage.
Consistent with the findings in Section~\ref{sec:Validity_SingleSysSingleConstr}, both ${\cal IZR}$ and ${\cal IZE}$ demonstrate significant improvements in efficiency over ${\cal F}_B$ in terms of OBS when $d_\ell=0.5$, where ${\cal IZE}$ constantly performs the best among all cases. 

\subsection{Efficiency}
\label{sec:Efficiency}

In this section, we demonstrate the efficiency of our proposed procedures. 
We adopt the same experimental configuration as in Section~\ref{sec:Validity_MultiSysMultiConstr}, with two key modifications. First, we focus on the SM configuration to better represent a practical setting.
Second, we consider a broader range of values for the mean parameter $d_\ell \in \{0.005, 0.01, 0.02, 0.05, 0.1, 0.5, 1, 2\}$.

Note that when $d_\ell = 0.005$ and $d_\ell = 0.01$, there are 9 and 3 acceptable systems, respectively, under all combinations of $(\underline{b}, \overline{b})$ considered. We focus on the CV variance configuration, as the results under the other four variance configurations exhibit similar patterns. The experimental results are summarized in Table~\ref{tab:Efficiency_k99s4_CV}. The experimental results based on the IV and DV variance configurations are included in Appendix \ref{sec:EfficiencyAdditional}. 

\begin{table}[h!]
	\centering
			\resizebox{\textwidth}{!}{ 
			\begin{tabular}{ c|c | c | cc|c || c|c | c | cc|c}
				\toprule
				$(\underline{b}, \overline{b})$ & $d_1$ & ${\cal F}_B$ & \multicolumn{2}{c|}{${\cal IZR}$} & ${\cal IZE}$ & $(\underline{b}, \overline{b})$ & $d_1$ & ${\cal F}_B$ & \multicolumn{2}{c|}{${\cal IZR}$} & ${\cal IZE}$  \\
				& & & $T=2, \xi=2$ & $T=2, \xi=3$ & & & & & $T=2, \xi=2$ & $T=2, \xi=3$ \\ \midrule  
				\multirow{16}{*}{$(55,77)$}
				& 0.005 & 995054 & 800026 & 795506 & 944899 & \multirow{16}{*}{$(33,66)$}
				& 0.005 & 981205 & 792406 & 790514 & 923142 \\
				& & (1.000) & (1.000) & (0.999) & (0.999)
				& & & (1.000) & (0.999) & (1.000) & (0.999) \\ \cline{2-6} \cline{8-12}
				& 0.01 & 591875 & 439363 & 413486 & 484775 & & 0.01 & 585965 & 435605 & 411886 & 479669 \\
				& & (1.000) & (1.000) & (1.000) & (1.000) & & & (0.999) & (1.000) & (1.000) & (1.000) \\ \cline{2-6}  \cline{8-12}
				& 0.02 & 331576 & 224920 & 200192 & 222075 & & 0.02 & 328495 & 223272 & 199134 & 221805 \\
				& & (0.999) & (1.000) & (0.999) & (1.000) & & & (0.999) & (1.000) & (0.999) & (1.000) \\ \cline{2-6}  \cline{8-12}
				& 0.05 & 144873 & 87814 & 66766 & 67669
                & & 0.05 & 143732 & 87176 & 66451 & 67684 \\
				& & (1.000) & (1.000) & (1.000) & (1.000) & & & (1.000) & (1.000) & (1.000) & (1.000) \\ \cline{2-6}  \cline{8-12}
				& 0.1 & 74913 & 43939 & 30570 & 24502
                & & 0.1 & 74394 & 43665 & 30340 & 24450 \\
				& & (1.000) & (1.000) & (1.000) & (1.000) & & & (1.000) & (1.000) & (1.000) & (1.000) \\ \cline{2-6} \cline{8-12}
				& 0.5 & 15495 & 8855 & 6092 & 2922 
                & & 0.5 & 15378 & 8823 & 6092 & 2917 \\
				& & (1.000) & (1.000) & (1.000) & (1.000) & & & (1.000) & (1.000) & (1.000) & (1.000) \\ \cline{2-6}  \cline{8-12}
				& 1 & 7846 & 4737 & 3548 & 2133 & & 1 & 7819 & 4733 & 3548 & 2133 \\
				& & (1.000) & (1.000) & (1.000) & (1.000)
                & & & (1.000) & (1.000) & (1.000) & (1.000) \\ \cline{2-6} \cline{8-12}
				& 2 & 4298 & 3000 & 2529 & 2001 & & 2 & 4297 & 3002 & 2528 & 2001 \\
				& & (1.000) & (1.000) & (1.000) & (1.000) & & & (1.000) & (1.000) & (1.000) & (1.000) \\  \hline
				\multirow{16}{*}{$(22,44)$}
				& 0.005 & 917369 & 754091 & 761425 & 862550 & \multirow{16}{*}{$(22,77)$}
				& 0.005 & 932379 & 762529 & 767152 & 875013 \\
				& & (0.999) & (1.000) & (0.999) & (0.999) & & & (1.000) & (0.999) & (0.999) & (0.999) \\  \cline{2-6} \cline{8-12}
				& 0.01 & 552557 & 416330 & 397735 & 458979 & & 0.01 & 560552 & 420375 & 401032 & 462774 \\
				& & (1.000) & (1.000) & (0.999) & (0.999) & & & (1.000) & (0.999) & (0.999) & (0.999) \\ \cline{2-6}  \cline{8-12}
				& 0.02 & 311766 & 213530 & 192602 & 217501 & & 0.02 & 315792 & 215805 & 194141 & 217444 \\
				& & (0.999) & (1.000) & (0.999) & (0.999) 
                & & & (1.000) & (0.999) & (1.000) & (0.999) \\ \cline{2-6}  \cline{8-12}
				& 0.05 & 137119 & 83227 & 63811 & 67075 &  & 0.05 & 138460 & 84140 & 64483 & 67285 \\
				& & (1.000) & (1.000) & (1.000) & (1.000) & & & (1.000) & (1.000) & (1.000) & (1.000) \\ \cline{2-6}  \cline{8-12}
				& 0.1 & 70983 & 41678 & 29070 & 24405 & & 0.1 & 71837 & 42167 & 29360 & 24397 \\
				& & (1.000) & (1.000) & (1.000) & (1.000) & & & (1.000) & (1.000) & (1.000) & (1.000) \\ \cline{2-6}  \cline{8-12}
				& 0.5 & 14742 & 8600 & 6023 & 2923 & & 0.5 & 14879 & 8629 & 6029 & 2916 \\
				& & (1.000) & (1.000) & (1.000) & (1.000) & & & (1.000) & (1.000) & (1.000) & (1.000) \\ \cline{2-6} \cline{8-12}
				& 1 & 7657 & 4717 & 3546 & 2134 & & 1 & 7678 & 4720 & 3548 & 2133 \\
				& & (1.000) & (1.000) & (1.000) & (1.000) & & & (1.000) & (1.000) & (1.000) & (1.000) \\ \cline{2-6}  \cline{8-12}
				& 2 & 4293 & 3002 & 2530 & 2001 & & 2 & 4295 & 3001 & 2529 & 2001  \\
				& & (1.000) & (1.000) & (1.000) & (1.000) & & & (1.000) & (1.000) & (1.000) & (1.000) \\ 
				\bottomrule
			\end{tabular}
		}
		\caption{Estimated OBS and PCD (in parentheses) of ${\cal F}_B, {\cal IZR}$, and ${\cal IZE}$ under the SM configuration with $k=99$ and $s=4$. The variance configuration is CV.}
		\label{tab:Efficiency_k99s4_CV}
	\end{table}

	We observe that both ${\cal IZR}$ and ${\cal IZE}$ consistently outperform ${\cal F}_B$ across all values of $d_\ell$ and under all combinations of $(\underline{b}, \overline{b})$. The patterns of observation savings for ${\cal IZR}$ and ${\cal IZE}$ are similar across the four $(\underline{b}, \overline{b})$ combinations. Figure~\ref{fig:OBSPercent} illustrates the percentage of observations required by each procedure, relative to ${\cal F}_B$, for all values of $d_\ell$ and each $(\underline{b}, \overline{b})$ combination considered.

    \begin{figure}[h!]
        \centering
        \begin{minipage}[b]{0.48\textwidth}
            \centering
            \subfigure[$(\underline{b}, \overline{b})=(55, 77)$]{
            \scalebox{0.3}{\includegraphics[width=10.5in]{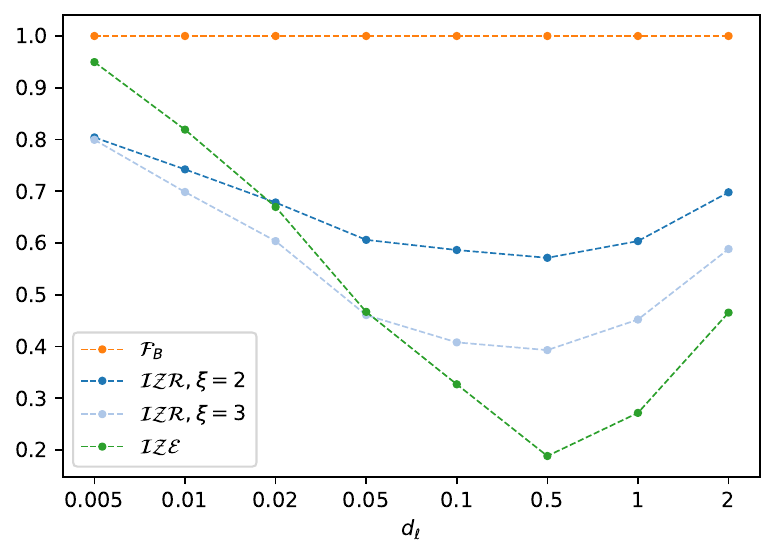}}}
        \end{minipage}
        \begin{minipage}[b]{0.48\textwidth}
            \centering
            \subfigure[$(\underline{b}, \overline{b})=(33, 66)$]{
            \scalebox{0.3}{\includegraphics[width=10.5in]{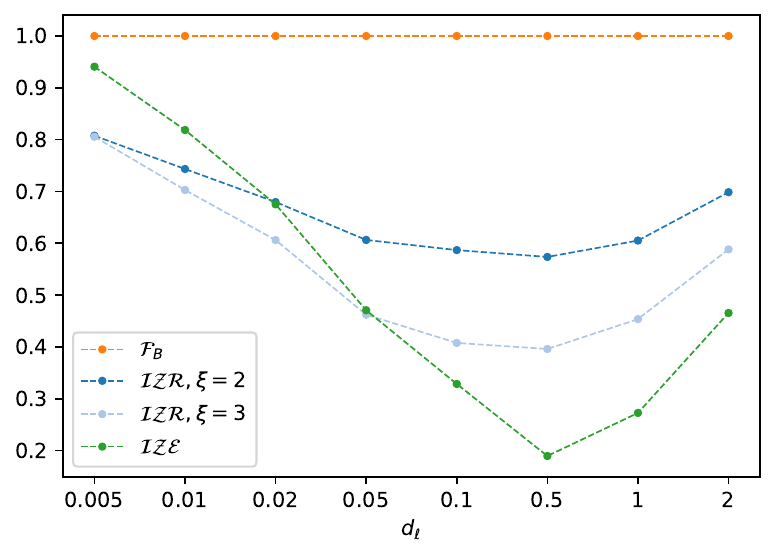}}}
        \end{minipage}
        \begin{minipage}[b]{0.48\textwidth}
            \centering
            \subfigure[$(\underline{b}, \overline{b})=(22, 44)$]{
            \scalebox{0.3}{\includegraphics[width=10.5in]{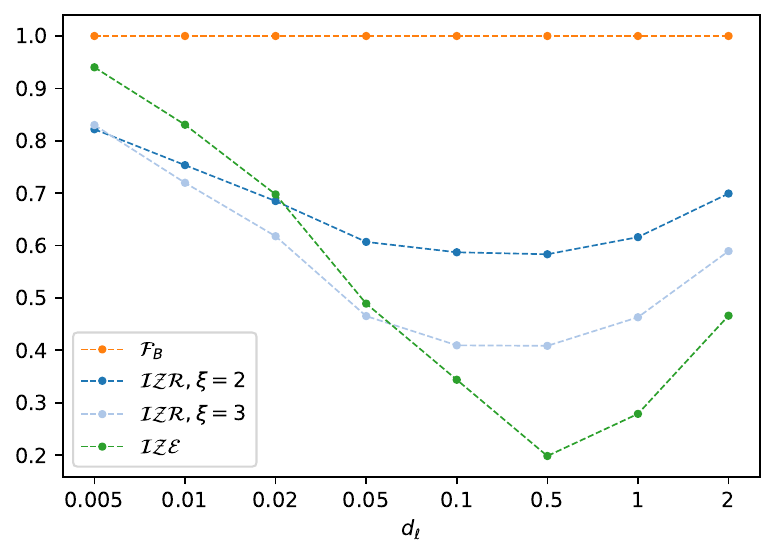}}}
        \end{minipage}
        \begin{minipage}[b]{0.48\textwidth}
            \centering
            \subfigure[$(\underline{b}, \overline{b})=(22, 77)$]{
            \scalebox{0.3}{\includegraphics[width=10.5in]{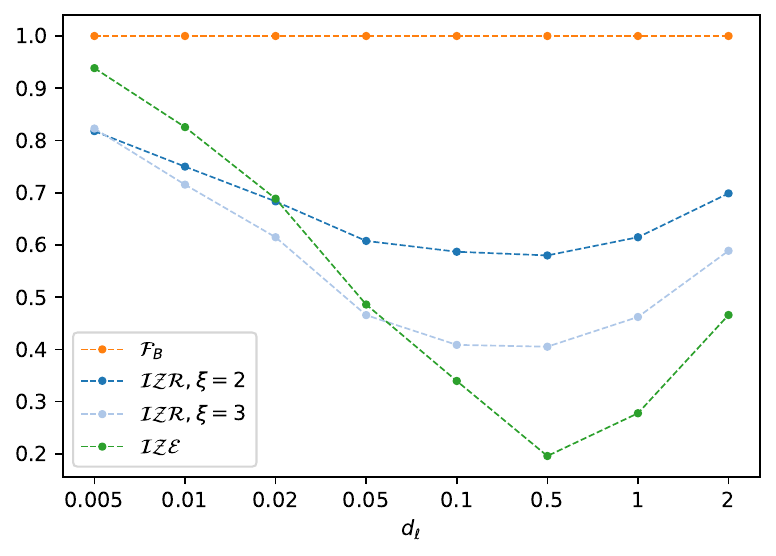}}}
        \end{minipage}
        \vspace{-6pt}
        \caption{Percentage of required OBS compared with ${\cal F}_B$ as a function of the value of $d_\ell$.} \label{fig:OBSPercent}
\end{figure}

	When the problem is difficult, i.e., when system means are close to the thresholds (such as $d_\ell = 0.005$), ${\cal IZE}$ requires fewer observations than ${\cal F}_B$ (approximately 95\%), while ${\cal IZR}$ achieves a more noticeable reduction, requiring approximately 80\% of the observations. 
	This limited savings from ${\cal IZE}$ is likely because the selected tolerance levels $\epsilon_\ell^{(1)}$ for ${\cal IZR}$ and the estimated $\epsilon_{i\ell}^{(1)}$ for ${\cal IZE}$ may not be sufficiently accurate to improve efficiency in such a difficult setting (in fact, $\epsilon_{i\ell}^{(1)}$ used in ${\cal IZE}$ may be even larger than for ${\cal IZR}$). It is also worth noting that although the savings from ${\cal IZR}$ and ${\cal IZE}$ are less pronounced in this setting compared to easier cases, this is expected. Neither procedure is specifically designed for scenarios where system means are densely concentrated around the threshold (as in the case of $d_\ell = 0.005$).
	
	As the problem becomes easier with increasing $d_\ell$, both ${\cal IZR}$ and ${\cal IZE}$ demonstrate substantial savings in efficiency. ${\cal IZR}$ requires only 40\%–60\% of the observations used by ${\cal F}_B$ when $d_\ell=0.5$, while ${\cal IZE}$ requires as little as 20\%. However, when the problem becomes very easy (e.g., $d_\ell \in \{1, 2\}$), the differences in required observations among ${\cal F}_B$, ${\cal IZR}$, and ${\cal IZE}$ diminish. This is expected, as in such settings, the majority of observations are used for estimating $\sigma_{i\ell}^2$ (and estimating $|y_{i\ell}-q_\ell|$ in the case of ${\cal IZE}$), rather than for performing feasibility decisions.

\subsection{Inventory Example}
\label{sec:Inventory}

In this section, we demonstrate the efficiency of our proposed procedures using an $({\sf s},{\sf S})$ inventory control example, where the decision-maker places an order to immediately raise the inventory level up to ${\sf S}$ whenever it falls below ${\sf s}$ at the beginning of a review period. If the inventory level is at least ${\sf s}$, no order is placed. This setting follows the example in \citet{LawKeltonbook}, and is also used in \citet{Zhou2022} and \citet{Zhou2024}.

The demand in each review period is assumed to follow a Poisson distribution with a mean of 25 units, and demands are independent across periods. The ordering cost consists of a fixed cost of 32 per order and a variable cost of 3 per unit ordered. Holding costs are 1 per unit per review period, and a penalty cost of 5 is incurred for each unit of unsatisfied demand.
We set the run length as 30 review periods\footnotemark\footnotetext{To ensure observations reflect steady-state behavior, we initialize the inventory level at ${\sf S}$ and discard the first 100 review periods as warm-up. The percentage of periods with a stockout and the average cost over the subsequent 30 review periods are used as single observations for the first and second performance measures, respectively. } and evaluate system performance using two measures:
\begin{enumerate}
	\item The probability that a stockout occurs in a review period ($\ell = 1$),
	\item The expected total cost per review period, including ordering, holding, and penalty costs ($\ell = 2$).
\end{enumerate}
The goal is to identify feasible $({\sf s},{\sf S})$ combinations such that the stockout probability does not exceed a threshold $q_1$ and the expected cost does not exceed $q_2$.

We consider 2901 independent systems as $\Theta=\{({\sf s},{\sf S})\mid 20\leq {\sf s}\leq 80, 40\leq {\sf S}\leq 100, {\sf s}\in \mathbb{Z}^+, {\sf S}\in \mathbb{Z}^+, \text{ and } {\sf s}\leq {\sf S}\}$. The analytical results for both performance measures are obtained using a steady-state analysis of a Markov chain model.
Figure \ref{fig:SysPerformance} shows the estimated performance of each system with respect to the two constraints. We also estimate the correlation between the two performance measures across all systems using simulation with 1,000,000 replications, finding the correlation ranges from -0.235 to 0.553. Note that this test case corresponds to multiple independent systems with correlated constraints. 
\begin{figure}[!t]
	\centering
	\includegraphics[width=.6\linewidth]{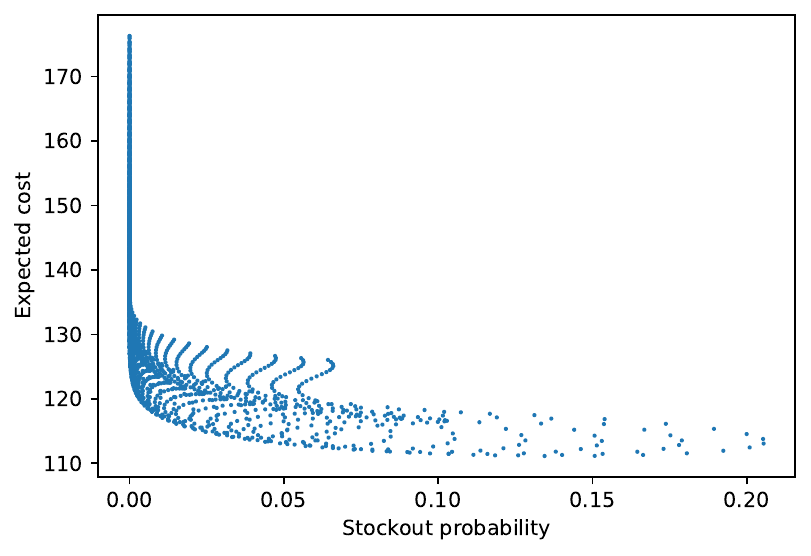}
	\caption{System performance with respect to the two constraints.}
	\label{fig:SysPerformance}
\end{figure}

We evaluate the performance of the procedures using three different sets of thresholds: $(q_1, q_2) \in \{(0.01, 120),\ (0.05, 125),\ (0.1, 130)\}$. Recall from Section~\ref{sec:ExpConfig} that systems with $i \leq \underline{b}$ satisfy both constraints, systems with $\underline{b} + 1 \leq i \leq \overline{b}$ satisfy only one constraint, and systems with $i \geq \overline{b} + 1$ violate both constraints. The threshold settings correspond to $(\underline{b}, \overline{b}) = (31, 2661),\ (526, 2892),\ \text{and } (1081, 2901)$, respectively.
Due to the nature of the two constraints, we set the tolerance levels as $\epsilon_1 = 0.001$ for the shortage probability constraint and $\epsilon_2 = 0.1$ for the expected cost constraint. The numbers of acceptable systems with respect to the three sets of threshold combinations are 2, 22, and 25, respectively.  Table \ref{tab:InventoryResults} shows the results where we set $n_0=20$. We also include additional results for $n_0\in \{30, 40, \ldots, 100\}$ in Appendix \ref{sec:InventoryAdditional}.  

\begin{table}[!tb]
	\centering
		{\footnotesize \begin{tabular}{ c || c c c c}
			\toprule
			& ${\cal F}_B$ & \multicolumn{2} {c} {${\cal IZR}$} & ${\cal IZE}$  \\   
			$(q_1, q_2)$ & & $T=2, \xi=2$ & $T=2, \xi=3$  \\ \midrule
			\multirow{ 2}{*}{(0.01, 120)} & 585540 & 397165 & 337587 & 288939  \\  
			& (0.991) & (0.990) & (0.990) & (0.990) \\  
			\hline
			\multirow{ 2}{*}{(0.05, 125)} & 1832392 & 1467543 & 1388945 & 1432279  \\  
			& (1.000) & (0.999) & (1.000) & (1.000) \\     
			\hline
			\multirow{ 2}{*}{(0.1, 130)} &  1157330 & 839989 & 766356 & 789751 \\  
			& (1.000) & (1.000) & (0.999) & (0.999) \\     
			\bottomrule
		\end{tabular}}
	\caption{Estimated OBS and PCD (in parentheses) of ${\cal F}_B, {\cal IZR}$, and ${\cal IZE}$ with respect to three sets of thresholds and $n_0=20$. For ${\cal IZE}$, we set $n_0'=n_0- 5$ and $\nu=0.8$. }
	\label{tab:InventoryResults}
\end{table}

First, we observe that all three procedures maintain statistical validity across all settings even though Assumption \ref{assump:normal} does not hold. In terms of efficiency, our proposed procedures, ${\cal IZR}$ and ${\cal IZE}$, consistently outperform the benchmark procedure ${\cal F}_B$ with respect to the required OBS, although the extent of improvement varies across the three threshold combinations (the percentage savings ranges within 20--51\%).

When $(q_1, q_2) = (0.01, 120)$, ${\cal IZE}$ achieves the best performance, requiring only 49\% of the observations used by ${\cal F}_B$. In contrast, ${\cal IZR}$ with $\xi = 2$ and $\xi = 3$ requires 68\% and 58\% of the observations, respectively. For $(q_1, q_2) = (0.05, 125)$, ${\cal IZR}$ with $\xi = 3$ performs best, requiring just 76\% of the observations required by ${\cal F}_B$, while ${\cal IZE}$ and ${\cal IZR}$ with $\xi = 2$ require 78\% and 80\% of the observations, respectively. Lastly, for $(q_1, q_2) = (0.1, 130)$, ${\cal IZE}$ uses 68\% of the observations used by ${\cal F}_B$, while ${\cal IZR}$ requires between 66\% and 73\% of the observations, depending on the choice of $\xi$.

The relative performance of ${\cal IZE}$ and ${\cal IZR}$ depends on the difference between the system means and the thresholds. Across all three threshold settings, using a larger $\xi$ (i.e., $\xi = 3$) in ${\cal IZR}$ leads to better performance, as it allows the procedure to eliminate clearly infeasible systems more quickly. This is consistent with the system performance distribution shown in Figure~\ref{fig:SysPerformance}, where many system means are far from the thresholds.

Comparing ${\cal IZE}$ and ${\cal IZR}$ with $\xi = 3$, we find that ${\cal IZE}$ outperforms ${\cal IZR}$ when $(q_1, q_2) = (0.01, 120)$, but requires slightly more observations (approximately 3\%) when $(q_1, q_2) = (0.05, 125)$ and $(0.1, 130)$. This performance difference can be attributed to two main factors: (i) the difficulty of ${\cal IZE}$ faces in accurately estimating $|y_{i\ell}-q_\ell|$ and selecting $\epsilon_{i\ell}^{(1)}$ with a limited number of initial samples, and (ii) the presence of many systems near the thresholds, which means that larger tolerance levels are not likely to be helpful when making feasibility decisions for these systems.
As a result, while ${\cal IZE}$ provides competitive performance, it does not consistently outperform ${\cal IZR}$. Nonetheless, ${\cal IZR}$ and ${\cal IZE}$ consistently and substantially outperform ${\cal F}_B$, and the performance difference between ${\cal IZE}$ and ${\cal IZR}$ remains relatively small, ranging from 14\% better to 3\% worse.

\subsection{Discussion}
\label{sec:Discussion}

Across all experiments conducted, we observe that ${\cal IZR}$ and ${\cal IZE}$ consistently outperform ${\cal F}_B$, except in the case where a single system with a single constraint is considered and its mean lies precisely at the acceptable boundary of the threshold. Overall, both ${\cal IZR}$ and ${\cal IZE}$ demonstrate substantial improvements in efficiency compared to ${\cal F}_B$.

When implementing ${\cal IZR}$, the choice of the parameter $\xi$ plays a critical role in the overall performance. If the decision-maker expects system means to be well spread out, implying a reasonable number of clearly feasible and clearly infeasible systems, a larger value of $\xi$ is recommended. This allows ``easy'' systems to conclude feasibility decisions with fewer observations.

Of course, a decision-maker may not know ahead of time how spread out the system means are. Overall, while ${\cal IZE}$ does not consistently outperform ${\cal IZR}$, it generally achieves at least comparable performance and it yields substantial savings in practical settings where many system means differ significantly from constraint thresholds. This indicates that ${\cal IZE}$ is more robust in terms of performance across varying problem instances.
Therefore, we recommend using ${\cal IZE}$ over ${\cal IZR}$, especially when the decision-maker lacks clear guidance for selecting the tuning parameter $\xi$ required by ${\cal IZR}$.

Regarding the implementation parameters of ${\cal IZE}$, the total number of initial observations $n_0$ and the proportion allocated to estimating $|y_{i\ell}-q_\ell|$ (i.e., $n_0'/n_0$) are both critical. If the decision-maker expects the problem to
be difficult (e.g., many systems are clustered near the thresholds), then allocating more samples to estimating $|y_{i\ell}-q_\ell|$ can improve performance. Conversely, if the problem is expected to be easier, then it is advisable to allocate fewer samples to estimating  $|y_{i\ell}-q_\ell|$, as doing so can avoid inefficiency caused by oversampling. In the absence of prior information about the system means, we recommend setting $n_0 = 20$ with $(n_0', n_0'') = (15, 5)$, as discussed in Section \ref{sec:ImplementationParams}.

\section{Conclusions}
\label{sec6}

In this paper, we address the problem of identifying feasible systems when performance is evaluated through stochastic simulation. We first propose ${\cal IZR}$, a fully sequential IZ procedure that incorporates multiple relaxed tolerance levels and uses two subroutines for each level. The ${\cal IZR}$ procedure enables us to reduce the number of required simulation observations, but the magnitude of the reduction varies depending on the selection of the number and values of the relaxed tolerance levels. We therefore suggest another procedure, the ${\cal IZE}$ procedure, which uses only two tolerance levels: one matches the original tolerance level and the other is estimated for each system and constraint using preliminary data on the system's performance measure. We prove the statistical validity of both procedures and demonstrate through experiments that they significantly reduce the number of observations required compared to an existing statistically valid procedure. 

In future research, we plan to explore potential improvements to the ${\cal IZE}$ procedure. Currently, ${\cal IZE}$ estimates a relaxed tolerance level for each system and constraint using preliminary data. To achieve computational efficiency, obtaining extensive preliminary data is often undesirable, resulting in inherently noisy estimates. To address this, instead of relying on tolerance levels that vary for each system and constraint, we aim to aggregate preliminary data across different systems and/or constraints. This approach would allow us to estimate fewer relaxed tolerance levels based on more data, thereby reducing the noise in our estimates.


\section*{\bf Acknowledgments}

Sigr\'{u}n Andrad\'{o}ttir was supported by NSF under grants CMMI-2127778 and CMMI-2348409. Seong-Hee Kim was supported by NSF under grant CMMI-2348409. Chuljin Park was supported by the National Research Foundation of Korea (NRF), which is funded by the Korea government (MSIT) under grant RS-2025-16067815.

\bibliographystyle{nonumber}

\clearpage
\newpage

\appendix
\section*{Appendices}
\section{Lemma for Section~3.3}

\begin{lemma} \label{lemma:IZR_Additional}
Suppose $h^2$, $S_{i \ell}^2$, $\epsilon_\ell$, and $\epsilon^{(\tau)}_\ell$ for $\tau = 1, 2,\ldots,T_\ell$ are positive real values and $\epsilon^{(\tau-1)}_\ell > \epsilon^{(\tau)}_\ell$ for any $\tau = 2,3,\ldots,T_\ell$. From (\ref{eq:cts1}), the shifted continuation region of ${\cal F_U}(\epsilon^{(\tau-1)}_\ell)$ is always included in that of ${\cal F_U}(\epsilon^{(\tau)}_\ell)$ for any $\tau = 2,3,\ldots,T_\ell$. Similarly, from  (\ref{eq:cts2}), the shifted continuation region of ${\cal F_D}(\epsilon^{(\tau-1)}_\ell)$ is always included in that of ${\cal F_D}(\epsilon^{(\tau)}_\ell)$ for any $\tau = 2,3,\ldots,T_\ell$.
\end{lemma}
\begin{proof}
By solving $-\frac{h^2 S_{i \ell}^2}{2\epsilon^{(\tau-1)}_\ell} + \frac{\epsilon^{(\tau-1)}_\ell}{2}r_i+ \epsilon_\ell r_i- \epsilon^{(\tau-1)}_\ell r_i= \frac{h^2 S_{i \ell}^2}{2\epsilon^{(\tau-1)}_\ell} - \frac{\epsilon^{(\tau-1)}_\ell}{2}r_i+ \epsilon_\ell r_i- \epsilon^{(\tau-1)}_\ell r_i$ for $r_i$, the upper and lower bounds of the shifted continuation region of ${\cal F_U}(\epsilon^{(\tau-1)}_\ell)$ cross at $r_i = \frac{h^2 S^2_{i\ell}}{(\epsilon^{(\tau-1)}_\ell)^2}$. The upper bounds of both ${\cal F_U}(\epsilon^{(\tau-1)}_\ell)$ and ${\cal F_U}(\epsilon^{(\tau)}_\ell)$ are lines whose $y$-intercepts equal $\frac{h^2 S_{i \ell}^2}{2\epsilon^{(\tau-1)}_\ell}$ and $\frac{h^2 S_{i \ell}^2}{2\epsilon^{(\tau)}_\ell}$, respectively. Then, $\frac{h^2 S_{i \ell}^2}{2\epsilon^{(\tau)}_\ell}>\frac{h^2 S_{i \ell}^2}{2\epsilon^{(\tau-1)}_\ell}$ (at $r_i = 0$) and the difference between the two upper bounds at $r_i = \frac{h^2 S^2_{i\ell}}{(\epsilon^{(\tau-1)}_\ell)^2}$ is  
\begin{eqnarray*}
    \lefteqn{\left\{\frac{h^2 S_{i \ell}^2}{2\epsilon^{(\tau)}_\ell} + \frac{h^2 S^2_{i\ell}}{(\epsilon^{(\tau-1)}_\ell)^2}\left(-\frac{\epsilon^{(\tau)}_\ell}{2}+ \epsilon_\ell - \epsilon^{(\tau)}_\ell \right) \right\} - \left\{\frac{h^2 S_{i \ell}^2}{2\epsilon^{(\tau-1)}_\ell} + \frac{h^2 S^2_{i\ell}}{(\epsilon^{(\tau-1)}_\ell)^2}\left(-\frac{\epsilon^{(\tau-1)}_\ell}{2}+ \epsilon_\ell - \epsilon^{(\tau-1)}_\ell \right) \right\}}\\
    &=&h^2 S_{i \ell}^2 \left\{\frac{1}{2\epsilon^{(\tau)}_\ell} + \frac{1}{(\epsilon^{(\tau-1)}_\ell)^2}\left(-\frac{3\epsilon^{(\tau)}_\ell}{2} \right) -\frac{1}{2\epsilon^{(\tau-1)}_\ell} - \frac{1}{(\epsilon^{(\tau-1)}_\ell)^2}\left(-\frac{3\epsilon^{(\tau-1)}_\ell}{2}\right) \right\}\\
    &=&h^2 S_{i \ell}^2 \left(\epsilon^{(\tau-1)}_\ell-\epsilon^{(\tau)}_\ell\right)\left\{\frac{1}{2\epsilon^{(\tau-1)}_\ell\epsilon^{(\tau)}_\ell} + \frac{3}{2(\epsilon^{(\tau-1)}_\ell)^2}\right\}>0.
\end{eqnarray*}
Thus, the upper bound of ${\cal F_U}(\epsilon^{(\tau)}_\ell)$ is always greater than that of ${\cal F_U}(\epsilon^{(\tau-1)}_\ell)$ at any $r_i$ such that $n_0 \leq r_i \leq \frac{h^2 S^2_{i\ell}}{(\epsilon^{(\tau-1)}_\ell)^2}$. Also, both lower bounds of ${\cal F_U}(\epsilon^{(\tau-1)}_\ell)$ and ${\cal F_U}(\epsilon^{(\tau)}_\ell)$ are lines with their $y$-intercepts, $-\frac{h^2 S_{i \ell}^2}{2\epsilon^{(\tau-1)}_\ell}$ and $-\frac{h^2 S_{i \ell}^2}{2\epsilon^{(\tau)}_\ell}$, respectively. Then, $-\frac{h^2 S_{i \ell}^2}{2\epsilon^{(\tau)}_\ell}<-\frac{h^2 S_{i \ell}^2}{2\epsilon^{(\tau-1)}_\ell}$ (at $r_i = 0$) and the difference between two lower bounds at $r_i = \frac{h^2 S^2_{i\ell}}{(\epsilon^{(\tau-1)}_\ell)^2}$ is  
\begin{eqnarray*}
    \lefteqn{\left\{-\frac{h^2 S_{i \ell}^2}{2\epsilon^{(\tau)}_\ell} + \frac{h^2 S^2_{i\ell}}{(\epsilon^{(\tau-1)}_\ell)^2}\left(\frac{\epsilon^{(\tau)}_\ell}{2}+ \epsilon_\ell - \epsilon^{(\tau)}_\ell \right) \right\} - \left\{-\frac{h^2 S_{i \ell}^2}{2\epsilon^{(\tau-1)}_\ell} + \frac{h^2 S^2_{i\ell}}{(\epsilon^{(\tau-1)}_\ell)^2}\left(\frac{\epsilon^{(\tau-1)}_\ell}{2}+ \epsilon_\ell - \epsilon^{(\tau-1)}_\ell \right) \right\}}\\
    &=&h^2 S_{i \ell}^2 \left\{-\frac{1}{2\epsilon^{(\tau)}_\ell} + \frac{1}{(\epsilon^{(\tau-1)}_\ell)^2}\left(-\frac{\epsilon^{(\tau)}_\ell}{2} \right) +\frac{1}{2\epsilon^{(\tau-1)}_\ell} - \frac{1}{(\epsilon^{(\tau-1)}_\ell)^2}\left(-\frac{\epsilon^{(\tau-1)}_\ell}{2}\right) \right\}\\
    &=&h^2 S_{i \ell}^2 \frac{-(\epsilon^{(\tau-1)}_\ell)^2-(\epsilon^{(\tau)}_\ell)^2 + 2 \epsilon^{(\tau)}_\ell \epsilon^{(\tau-1)}_\ell}{2\epsilon^{(\tau)}_\ell (\epsilon^{(\tau-1)}_\ell)^2}\\
    &=& - h^2 S_{i \ell}^2 \frac{\left(\epsilon^{(\tau-1)}_\ell-\epsilon^{(\tau)}_\ell\right)^2 }{2\epsilon^{(\tau)}_\ell (\epsilon^{(\tau-1)}_\ell)^2}<0.
\end{eqnarray*}
Thus, the lower bound of ${\cal F_U}(\epsilon^{(\tau)}_\ell)$ is always less than that of ${\cal F_U}(\epsilon^{(\tau-1)}_\ell)$ at any $r_i$ such that $n_0 \leq r_i \leq \frac{h^2 S^2_{i\ell}}{(\epsilon^{(\tau-1)}_\ell)^2}$. As a result, the shifted continuation region of ${\cal F_U}(\epsilon^{(\tau-1)}_\ell)$ is always included in that of ${\cal F_U}(\epsilon^{(\tau)}_\ell)$ for any $\tau = 2,3,\ldots,T_\ell$. Since (\ref{eq:cts1}) and (\ref{eq:cts2}) are symmetric, one can directly show that the shifted continuation region of ${\cal F_D}(\epsilon^{(\tau-1)}_\ell)$ is always included in that of ${\cal F_D}(\epsilon^{(\tau)}_\ell)$ for any $\tau = 2,3,\ldots,T_\ell$. \end{proof}

\section{Additional Results for Section \ref{sec:ImplementationParams}}
\label{sec:ImplementationParamAdditional}

In this section, we present additional experimental results related to the implementation parameters discussed in Section~\ref{sec:ImplementationParams}. 
Specifically, Section \ref{sec:IZE_DiffComb_d_n0_Additional} presents additional results of ${\cal IZE}$ for $d_\ell=0.01$ under different values of $\nu, n_0'$, and $n_0''$ and different variance configurations. Section \ref{sec:IZE_DiffVarConfig_Additional} includes additional discussion regarding the performance of ${\cal IZE}$ under different variance configurations. Finally, Section \ref{sec:FB_IZR_Diff_n0_Additional} shows additional results of Procedures ${\cal F}_B, {\cal IZR}$, and ${\cal IZE}$ with a wider range of $n_0$ values. 

\subsection{Additional Results of ${\cal IZE}$ for $d_\ell=0.01$}
\label{sec:IZE_DiffComb_d_n0_Additional}

Table~\ref{tab:DiffCombInitSample_IZRE_d0.01} summarizes the results for ${\cal IZE}$ under the setting $d_\ell = 0.01$, with $(n_0', n_0'') \in \{(5, 15), (10, 10), (15, 5), (20, 0)\}$ and $\nu \in \{0.6, 0.8, 1\}$, across all five variance configurations CV, IV-C, DV-C, IV-S, and DV-S.

\begin{table}[h!]
	\centering
		\begin{tabular}{ c || c | c cccc}
			\toprule
			& & \multicolumn{5}{c}{$d_\ell=0.01$}  \\  
			$\nu$ & $(n_0', n_0'')$ & CV & IV-C & DV-C & IV-S & DV-S  \\  \midrule
			\multirow{ 8}{*}{0.6} 
			& $(5, 15)$ & 507381 & 582873 & 467675 & 384759 & 643188 \\ 
			& & (0.999) & (0.999) & (0.999) & (0.999) & (0.999) \\   \cline{2 - 7}
			& $(10, 10)$ & 487178 & 559160 & 449413 & 369911 & 614107 \\  
			& & (0.999) & (1.000) & (0.999) & (1.000) & (0.999) \\   \cline{2 - 7}
			& $(15, 5)$ & 478735 & 547702 & 440987 & 363909 & 600270 \\ 
			& & (1.000) & (0.999) & (0.999) & (0.999) & (0.999) \\   \cline{2 - 7}
			& $(20,0)$ & 459273 & 524638 & 423052 & 351245 & 573748 \\ 
			& & (0.999) & (1.000) & (1.000) & (0.999) & (0.999)
			\\  \midrule
			\multirow{ 8}{*}{0.8} 
			& $(5, 15)$ & 522016 & 601879 & 481081 & 395005 & 662881 \\ 
			& & (0.999) & (0.999) & (0.999) & (0.999) & (0.999) \\  \cline{2 - 7}
			& $(10, 10)$ & 494197 & 569412 & 456553 & 372855 & 629393 \\  
			& & (0.999) & (1.000) & (1.000) & (1.000) & (1.000) \\   \cline{2 - 7}
			& $(15, 5)$ & 479669 & 551741 & 442222 & 362422 & 609485 \\ 
			& & (1.000) & (0.999) & (0.999) & (0.999) & (0.999) \\   \cline{2 - 7}
			& $(20,0)$ & 455594 & 524039 & 419502 & 345723 & 576925 \\ 
			& & (0.999) & (1.000) & (0.999) & (1.000) & (1.000)
			\\  \midrule
			\multirow{ 8}{*}{1} 
			& $(5, 15)$ & 540194 & 622610 & 496282 & 409477 & 683579 \\ 
			& & (0.999) & (0.999) & (0.999) & (0.999) & (1.000) \\  \cline{2 - 7}
			& $(10, 10)$ & 511785 & 590947 & 471571 & 385310 & 652446 \\  
			& & (1.000) & (0.999) & (1.000) & (0.999) & (0.999) \\   \cline{2 - 7}
			& $(15, 5)$ & 496284 & 570993 & 456918 & 372694 & 632438 \\ 
			& & (1.000) & (0.999) & (0.998) & (0.999) & (0.999)  \\   \cline{2 - 7}
			& $(20,0)$ & 469609 & 540550 & 433140 & 353993 & 599204 \\ 
			& & (0.999) & (0.999) & (1.000) & (0.999) & (0.999) \\  
			\bottomrule    
		\end{tabular}
	\caption{Estimated OBS and PCD (in parentheses) of ${\cal IZE}$ with respect to different combinations of $(n_0', n_0'')$ and $\nu$ under the SM configuration with $k=99$ systems, $s=4$ constraints, and $d_\ell=0.01$.}
	\label{tab:DiffCombInitSample_IZRE_d0.01}
\end{table}

\subsection{Additional Discussion on ${\cal IZE}$ Under Different Variance Configurations}
\label{sec:IZE_DiffVarConfig_Additional}

In this section, we provide additional discussion on the performance of ${\cal IZE}$ under different variance configurations. Specifically, we consider the setting where $\nu=0.8$ and $(n_0', n_0'')=(15, 5)$ under the SM configuration with $k=99$ systems and $s=4$ constraints (as discussed in Section \ref{sec:ImplementationParams}, see Table \ref{tab:DiffCombInitSample_IZRE}). We compare the required OBS from each system under different variance configurations. 

\paragraph{IV-C vs. DV-C.} As shown in Figure~\ref{fig:OBSBySys_C}, which presents the estimated OBS by system under IV-C and DV-C with $\nu=0.8$ and $(n_0', n_0'')=(15, 5)$, the difference between IV-C and DV-C arises due to how difficulty varies across constraints for systems $\underline{b}+1,\ldots, \overline{b}$. Specifically, for the systems with indices between $\underline{b} + 1$ and $\overline{b}$, IV-C assigns lower variance to the first $m$ constraints, making them relatively easier, while DV-C makes the remaining $(s - m)$ constraints easier. Because concluding that a system is infeasible requires identifying just one violated constraint, DV-C, which favors easier infeasible constraints, tends to terminate sooner, resulting in fewer observations. This difference is more pronounced when $d_\ell = 0.02$ (i.e., in the slippage configuration), where feasibility checks are inherently more difficult. When $d_\ell = 0.5$, the feasibility checks become easier overall, reducing the impact of the variance configuration and narrowing the performance gap between DV-C and IV-C.
\begin{figure}[h!]
        \centering
        \begin{minipage}[b]{0.48\textwidth}
            \centering
            \subfigure[$d_\ell=0.02$]{
            \scalebox{0.3}{\includegraphics[width=10.5in]{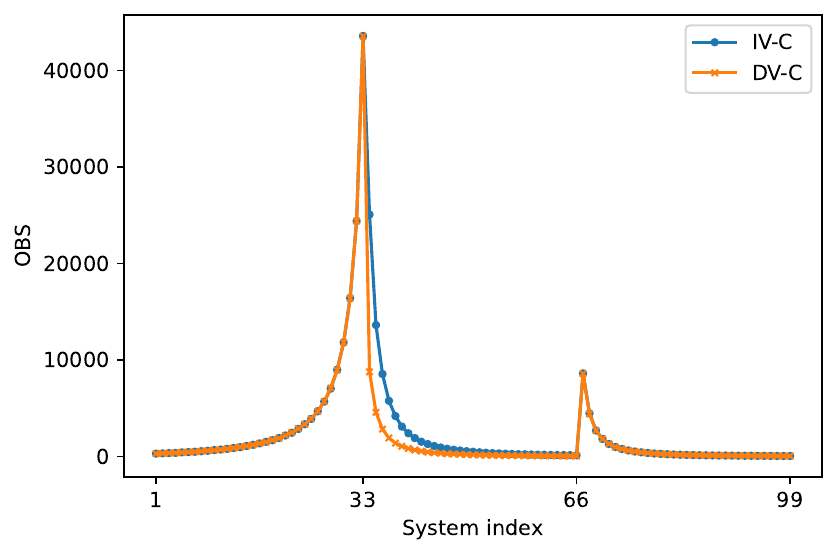}}}
        \end{minipage}
        \begin{minipage}[b]{0.48\textwidth}
            \centering
            \subfigure[$d_\ell=0.5$]{
            \scalebox{0.3}{\includegraphics[width=10.5in]{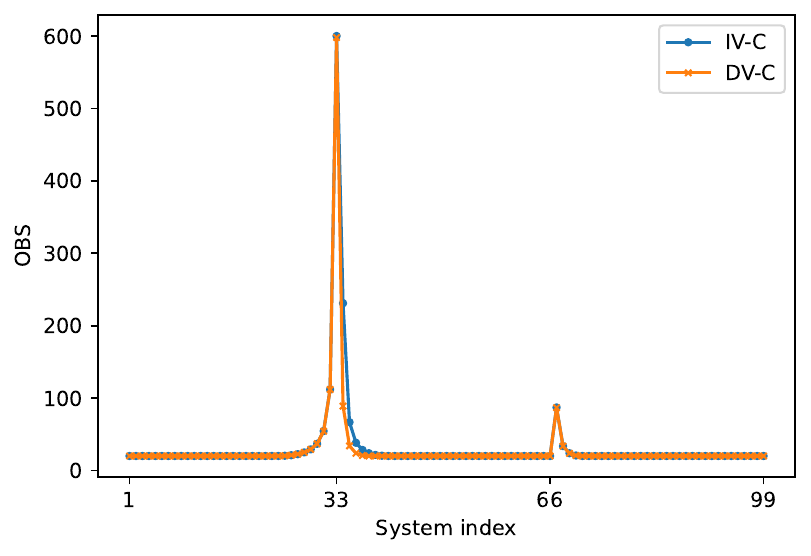}}}
        \end{minipage}
        \vspace{-6pt}
        \caption{Estimated OBS by systems for ${\cal IZE}$ under the IV-C and DV-C variance configurations.} \label{fig:OBSBySys_C}
\end{figure}

\paragraph{IV-S vs.\ DV-S.} Figure~\ref{fig:OBSBySys_S} illustrates the estimated OBS per system for both configurations, where we set $\nu=0.8$ and $(n_0', n_0'')=(15, 5)$. To interpret the results, we consider three segments of systems: (i) $i \leq \underline{b}$, (ii) $\underline{b} + 1 \leq i \leq \overline{b}$, and (iii) $i \geq \overline{b} + 1$.
\begin{itemize}
	\item Among the first $\underline{b}=33$ systems, feasibility checks are the most difficult (as all constraints are feasible and concluding a correct decision requires a correct decision for each constraint), with difficulty increasing as the system index increases. Because DV-S assigns higher variance to these systems compared to IV-S, it results in higher OBS.
	\item In the middle group, the feasibility difficulty is moderate (as only half of the constraints are feasible). DV-S requires more observations due to its relatively high variance in the earlier portion of this group where the feasibility decision is more difficult. The difference between DV-S and IV-S diminishes toward the latter half of the group, where IV-S leads to higher variance but the feasibility decision is easier.
	\item For the final $(k - \overline{b})$ systems, IV-S requires slightly more observations due to higher variance. However, since these systems are generally easier to classify, the increase in OBS is not substantial.
\end{itemize}
\begin{figure}[h!]
        \centering
        \begin{minipage}[b]{0.48\textwidth}
            \centering
            \subfigure[$d_\ell=0.02$]{
            \scalebox{0.3}{\includegraphics[width=10.5in]{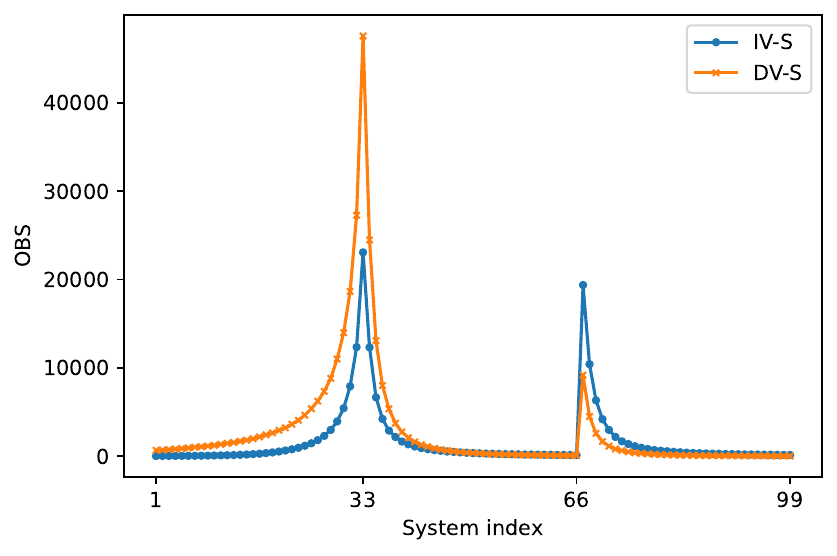}}}
        \end{minipage}
        \begin{minipage}[b]{0.48\textwidth}
            \centering
            \subfigure[$d_\ell=0.5$]{
            \scalebox{0.3}{\includegraphics[width=10.5in]{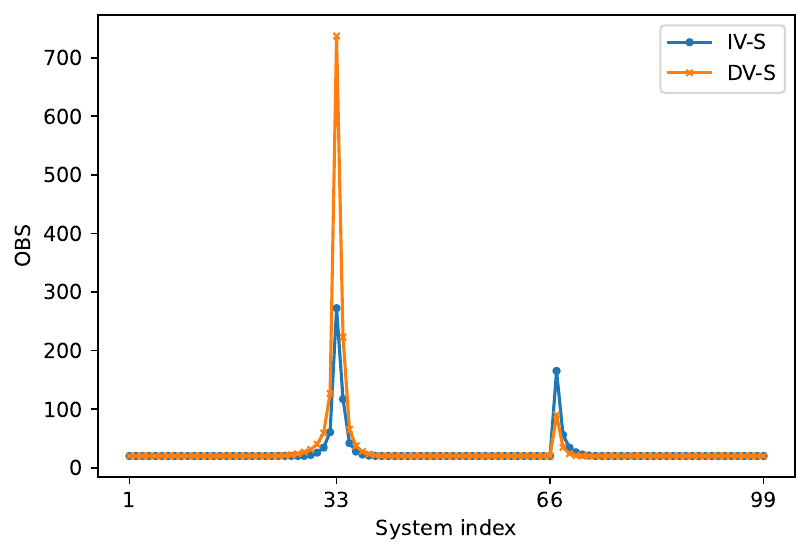}}}
        \end{minipage}
        \vspace{-6pt}
        \caption{Estimated OBS by systems for ${\cal IZE}$ under the IV-S and DV-S variance configurations.} \label{fig:OBSBySys_S}
\end{figure}

\subsection{Additional Results for ${\cal F}_B$, ${\cal IZR}$, and ${\cal IZE}$ Regarding $n_0$}
\label{sec:FB_IZR_Diff_n0_Additional}

Tables~\ref{tab:DiffCombInitSample_IZRE_nu0.6}--\ref{tab:DiffCombInitSample_IZRE_nu1} report the results for ${\cal IZE}$ under a wide range of $n_0$ values, with $\nu$ set to 0.6, 0.8, and 1, respectively.
Similarly, Table~\ref{tab:DiffCombInitSample_FB_IZR} presents the results for ${\cal F}_B$ and ${\cal IZR}$ under the same wide range of $n_0$ values as in Tables \ref{tab:DiffCombInitSample_IZRE_nu0.6}--\ref{tab:DiffCombInitSample_IZRE_nu1} for comparison.

\begin{table}[h!]
	\centering
	\resizebox{\columnwidth}{!}{
		\begin{tabular}{ c || c cccc | c cccc}
			\toprule
			& \multicolumn{5}{c|}{$d_\ell=0.02$} & \multicolumn{5}{c}{$d_\ell=0.5$}   \\  
			\diagbox{$n_0$}{$n_0'$} & $n_0-20$ & $n_0-15$ & $n_0-10$ & $n_0-5$ & $n_0$ & $n_0-20$ & $n_0-15$ & $n_0-10$ & $n_0-5$ & $n_0$  \\  \midrule
			$10$ & -- & -- & -- & 529315 & 427226 & -- & -- & -- & 4180 & 3672 \\ 
			& -- & -- & -- & (1.000) & (0.999) & -- & -- & -- & (1.000) & (1.000) \\   \hline
			$20$ & -- & 240772 & 229395 & 224430 & 214634 & -- & 3267 & 3169 & 3161 & 3268 \\ 
			& -- & (1.000) & (1.000) & (1.000) & (1.000) & -- & (1.000) & (1.000) & (1.000) & (1.000) \\  \hline
			$30$ & 187349 & 183065 & 180855 & 179740 & 176430 & 3838 & 3800 & 3792 & 3824 & 3988  \\  
			& (0.999) & (1.000) & (1.000) & (0.999) & (0.999) & (1.000) & (1.000) & (1.000) & (1.000) & (1.000) \\   \hline
			$40$ & 165028 & 163690 & 163071 & 162522 & 161302 & 4617 & 4615 & 4640 & 4690 & 4868 \\ 
			& (1.000) & (1.000) & (1.000) & (1.000) & (0.999) & (1.000) & (1.000) & (1.000) & (1.000) & (1.000) \\  \hline
			$50$ & 154729 & 154315 & 154136 & 154224 & 153741 & 5505 & 5529 & 5560 & 5616 & 5804 \\ 
			& (1.000) & (1.000) & (1.000) & (1.000) & (0.999) & (1.000) & (1.000) & (1.000) & (1.000) & (1.000)   \\   \hline
			$60$ & 149082 & 149076 & 149249 & 149382 & 149194 & 6446 & 6471 & 6511 & 6569 & 6760 \\
			& (1.000) & (0.999) & (1.000) & (1.000) & (0.999) & (1.000) & (1.000) & (1.000) & (1.000) & (1.000)  \\  \hline
			$70$ & 145846 & 145912 & 146145 & 146614 & 146703 & 7401 & 7434 & 7476 & 7535 & 7727 \\
			& (1.000) & (1.000) & (1.000) & (1.000) & (1.000) & (1.000) & (1.000) & (1.000) & (1.000) & (1.000)   \\  \hline
			$80$ & 143582 & 143977 & 144225 & 144827 & 145061 & 8373 & 8405 & 8448 & 8510 & 8702 \\
			& (1.000) & (1.000) & (0.999) & (1.000) & (1.000) & (1.000) & (1.000) & (1.000) & (1.000) & (1.000)  \\  \hline
			$90$ & 142690 & 142910 & 143192 & 143652 & 144005 & 9350 & 9384 & 9427 & 9487 & 9681  \\   
			& (1.000) & (1.000) & (0.999) & (0.999) & (1.000) & (1.000) & (1.000) & (1.000) & (1.000) & (1.000) \\    \hline
			$100$ & 141693 & 142309 & 142872 & 143114 & 143422 & 10329 & 10363 & 10408 & 10469 & 10662  \\
			& (1.000) & (1.000) & (0.999) & (1.000) & (1.000) & (1.000) & (1.000) & (1.000) & (1.000) & (1.000)  \\    \hline
			$150$ & 142265 & 142886 & 143317 & 143916 & 144220 & 15254 & 15289 & 15332 & 15394 & 15586 \\
			& (1.000) & (0.999) & (1.000) & (1.000) & (1.000) & (1.000) & (1.000) & (1.000) & (1.000) & (1.000)  \\    \hline
			$200$ & 145355 & 145932 & 146522 & 146937 & 147418 & 20193 & 20228 & 20272 & 20333 & 20524  \\
			& (1.000) & (1.000) & (1.000) & (1.000) & (1.000) & (1.000) & (1.000) & (1.000) & (1.000) & (1.000)  \\    \hline
			$300$ & 153954 & 154403 & 154758 & 155478 & 155974 & 30084 & 30119 & 30161 & 30222 & 30413  \\
			& (0.999) & (0.999) & (1.000) & (0.999) & (1.000) & (1.000) & (1.000) & (1.000) & (1.000) & (1.000) \\   \hline
			$400$ & 163298 & 163865 & 164236 & 164901 & 165333 & 39978 & 40013 & 40057 & 40116 & 40307  \\
			& (1.000) & (1.000) & (1.000) & (1.000) & (0.999) & (1.000) & (1.000) & (1.000) & (1.000) & (1.000) \\  \hline
			$500$ & 173030 & 173598 & 174050 & 174549 & 175099 & 49875 & 49910 & 49953 & 50013 & 50205  \\
			& (1.000) & (0.999) & (1.000) & (1.000) & (1.000) & (1.000) & (1.000) & (1.000) & (1.000) & (1.000)   \\    \hline
			$600$ & 182857 & 183447 & 183942 & 184405 & 184916 & 59775 & 59808 & 59852 & 59911 & 60102  \\
			& (0.999) & (0.999) & (1.000) & (1.000) & (1.000) & (1.000) & (1.000) & (1.000) & (1.000) & (1.000)  \\    \hline
			$700$ & 192716 & 193191 & 193767 & 194401 & 194812 & 69673 & 69707 & 69751 & 69810 & 70001 \\
			& (0.999) & (1.000) & (1.000) & (0.999) & (0.999) & (1.000) & (1.000) & (1.000) & (1.000) & (1.000) \\    \hline
			$800$ & 202643 & 203294 & 203913 & 204203 & 204541 & 79572 & 79607 & 79649 & 79710 & 79900 \\
			& (1.000) & (0.999) & (1.000) & (1.000) & (1.000) & (1.000) & (1.000) & (1.000) & (1.000) & (1.000) \\    \hline
			$900$ & 212639 & 213112 & 213642 & 214252 & 214743 & 89471 & 89506 & 89549 & 89608 & 89800 \\
			& (0.999) & (1.000) & (1.000) & (1.000) & (1.000) & (1.000) & (1.000) & (1.000) & (1.000) & (1.000) \\    \hline
			$1000$ & 222528 & 223117 & 223546 & 224123 & 224649 & 99371 & 99405 & 99448 & 99508 & 99700  \\
			& (1.000) & (1.000) & (1.000) & (1.000) & (0.999) & (1.000) & (1.000) & (1.000) & (1.000) & (1.000)   \\  
			\bottomrule    
	\end{tabular}}
	\caption{Estimated OBS and PCD (in parentheses) of ${\cal IZE}$ with $n_0\in \{10, 20, 30, 40, 50, 60, 70, 80, 90, 100, 150, 200, 300, 400, 500, 600, 700,800,900, 1000\}, n_0'\in \{n_0-20, n_0-15, n_0-10, n_0-5, n_0\}$, and $\nu=0.6$ under the SM and CV configurations with $k=99$ systems and $s=4$ constraints.}
	\label{tab:DiffCombInitSample_IZRE_nu0.6}
\end{table}

\begin{table}[h!]
	\centering
	\resizebox{\columnwidth}{!}{
		\begin{tabular}{ c || c cccc | c cccc}
			\toprule
			& \multicolumn{5}{c|}{$d_\ell=0.02$} & \multicolumn{5}{c}{$d_\ell=0.5$}   \\  
			\diagbox{$n_0$}{$n_0'$} & $n_0-20$ & $n_0-15$ & $n_0-10$ & $n_0-5$ & $n_0$ & $n_0-20$ & $n_0-15$ & $n_0-10$ & $n_0-5$ & $n_0$  \\  \midrule
			$10$ & -- & -- & -- & 545038 & 429440 & -- & -- & -- & 3879 & 3217 \\ 
			& -- & -- & -- & (0.999) & (1.000) & -- & -- & -- & (1.000) & (1.000) \\   \hline
			$20$ & -- & 247725 & 230624 & 221806 & 209564 & -- & 3170 & 2973 & 2917 & 3026 \\ 
			& -- & (1.000) & (1.000) & (1.000) & (0.999) & -- & (1.000) & (1.000) & (1.000) & (1.000) \\   \hline
			$30$ & 188131 & 181229 & 176363 & 173285 & 168685 & 3683 & 3629 & 3613 & 3637 & 3798 \\  
			& (1.000) & (0.999) & (1.000) & (1.000) & (1.000) & (1.000) & (1.000) & (1.000) & (1.000) & (1.000) \\   \hline
			$40$ & 160894 & 158045 & 156049 & 154347 & 152224 & 4469 & 4462 & 4485 & 4528 & 4701 \\ 
			& (0.999) & (1.000) & (0.999) & (1.000) & (1.000) & (1.000) & (1.000) & (1.000) & (1.000) & (1.000) \\   \hline
			$50$ & 148242 & 146705 & 145541 & 144653 & 143537 & 5368 & 5386 & 5419 & 5468 & 5645 \\ 
			& (1.000) & (1.000) & (1.000) & (1.000) & (1.000) & (1.000) & (1.000) & (1.000) & (1.000) & (1.000)    \\ \hline
			$60$ & 140654 & 139964 & 139065 & 138795 & 137936 & 6316 & 6340 & 6376 & 6427 & 6609 \\
			& (1.000) & (0.999) & (0.999) & (1.000) & (0.999) & (1.000) & (1.000) & (1.000) & (1.000) & (1.000)  \\  \hline
			$70$ & 136102 & 135568 & 135483 & 135091 & 134676 & 7280 & 7309 & 7347 & 7398 & 7581 \\
			&  (1.000) & (0.999) & (1.000) & (1.000) & (1.000) & (1.000) & (1.000) & (1.000) & (1.000) & (1.000)   \\  \hline
			$80$ & 133079 & 132818 & 132644 & 132317 & 132450 & 8255 & 8285 & 8322 & 8376 & 8558 \\
			& (0.999) & (0.999) & (0.999) & (1.000) & (0.999) & (1.000) & (1.000) & (1.000) & (1.000) & (1.000)  \\  \hline
			$90$ & 131185 & 130993 & 131104 & 131091 & 130890 & 9236 & 9266 & 9303 & 9357 & 9539  \\  
			& (1.000) & (1.000) & (1.000) & (1.000) & (1.000) & (1.000) & (1.000) & (1.000) & (1.000) & (1.000) \\    \hline
			$100$ & 129439 & 129620 & 129788 & 130007 & 129892 & 10218 & 10248 & 10286 & 10340 & 10523  \\
			& (1.000) & (1.000) & (0.999) & (1.000) & (1.000) & (1.000) & (1.000) & (1.000) & (1.000) & (1.000)  \\   \hline
			$150$ & 127724 & 128070 & 128368 & 128578 & 128915 & 15148 & 15179 & 15217 & 15271 & 15452  \\
			& (1.000) & (0.999) & (0.999) & (1.000) & (0.999) & (1.000) & (1.000) & (1.000) & (1.000) & (1.000)  \\    \hline
			$200$ & 129409 & 129825 & 130133 & 130441 & 130938 & 20089 & 20121 & 20158 & 20212 & 20393  \\
			& (0.999) & (0.999) & (0.999) & (1.000) & (0.999) & (1.000) & (1.000) & (1.000) & (1.000) & (1.000)   \\   \hline
			$300$ & 136188 & 136667 & 137188 & 137423 & 137926 & 29981 & 30013 & 30051 & 30103 & 30285  \\
			& (0.999) & (0.999) & (1.000) & (1.000) & (1.000) & (1.000) & (1.000) & (1.000) & (1.000) & (1.000) \\    \hline
			$400$ & 144701 & 145105 & 145639 & 146044 & 146554 & 39878 & 39908 & 39946 & 39999 & 40179  \\
			& (1.000) & (1.000) & (1.000) & (1.000) & (0.999) & (1.000) & (1.000) & (1.000) & (1.000) & (1.000) \\   \hline
			$500$ & 153734 & 154350 & 154608 & 155030 & 155836 & 49776 & 49807 & 49843 & 49897 & 50077  \\
			& (0.999) & (0.999) & (1.000) & (0.999) & (0.999) & (1.000) & (1.000) & (1.000) & (1.000) & (1.000)   \\   \hline
			$600$ & 163128 & 163739 & 164265 & 164555 & 165140 & 59675 & 59705 & 59743 & 59795 & 59976  \\
			& (0.999) & (1.000) & (1.000) & (1.000) & (1.000) & (1.000) & (1.000) & (1.000) & (1.000) & (1.000)  \\   \hline
			$700$ & 172770 & 173370 & 173701 & 174187 & 174710 & 69573 & 69604 & 69641 & 69694 & 69875  \\
			& (0.999) & (1.000) & (0.999) & (1.000) & (0.999) & (1.000) & (1.000) & (1.000) & (1.000) & (1.000)  \\   \hline
			$800$ & 182386 & 182951 & 183485 & 183895 & 184504 & 79473 & 79504 & 79542 & 79593 & 79773  \\
			& (1.000) & (1.000) & (1.000) & (1.000) & (1.000) & (1.000) & (1.000) & (1.000) & (1.000) & (1.000)  \\   \hline
			$900$ & 192286 & 192667 & 193135 & 193673 & 194176 & 89372 & 89403 & 89440 & 89493 & 89673  \\
			& (1.000) & (0.999) & (1.000) & (1.000) & (0.999) & (1.000) & (1.000) & (1.000) & (1.000) & (1.000)  \\   \hline
			$1000$ & 201888 & 202527 & 203035 & 203512 & 203927 & 99272 & 99302 & 99339 & 99392 & 99572  \\
			& (0.999) & (0.999) & (1.000) & (1.000) & (1.000) & (1.000) & (1.000) & (1.000) & (1.000) & (1.000)   \\  
			\bottomrule    
	\end{tabular}}
	\caption{Estimated OBS and PCD (in parentheses) of ${\cal IZE}$ with $n_0\in \{10, 20, 30, 40, 50, 60, 70, 80, 90, 100, 150, 200, 300, 400, 500, 600, 700, 800, 900, 1000\}, n_0'\in \{n_0-20, n_0-15, n_0-10, n_0-5, n_0\}$, and $\nu=0.8$ under the SM and CV configurations with $k=99$ systems and $s=4$ constraints.}
	\label{tab:DiffCombInitSample_IZRE_nu0.8}
\end{table}

\begin{table}[h!]
	\centering
	\resizebox{\columnwidth}{!}{
		\begin{tabular}{ c || c cccc | c cccc}
			\toprule
			& \multicolumn{5}{c|}{$d_\ell=0.02$} & \multicolumn{5}{c}{$d_\ell=0.5$}   \\  
			\diagbox{$n_0$}{$n_0'$} & $n_0-20$ & $n_0-15$ & $n_0-10$ & $n_0-5$ & $n_0$ & $n_0-20$ & $n_0-15$ & $n_0-10$ & $n_0-5$ & $n_0$  \\  \midrule
			$10$ & -- & -- & -- & 573559 & 449910 & -- & -- & -- & 4340 & 3323 \\ 
			& -- & -- & -- & (0.999) & (0.999) & -- & -- & -- & (1.000) & (1.000) \\   \hline
			$20$ & -- & 259923 & 241358 & 230735 & 217018 & -- & 3410 & 3075 & 2935 & 3002 \\ 
			& -- & (1.000) & (0.999) & (0.999) & (1.000) & -- & (1.000) & (1.000) & (1.000) & (1.000) \\  \hline
			$30$ & 196666 & 188252 & 182737 & 178536 & 173032 & 3782 & 3669 & 3615 & 3614 & 3750 \\  
			& (0.999) & (0.999) & (1.000) & (1.000) & (0.999) & (1.000) & (1.000) & (1.000) & (1.000) & (1.000) \\    \hline
			$40$ & 166407 & 162797 & 159864 & 157705 & 154958 & 4480 & 4462 & 4458 & 4481 & 4642 \\ 
			& (1.000) & (0.999) & (1.000) & (0.999) & (0.999) & (1.000) & (1.000) & (1.000) & (1.000) & (1.000) \\   \hline
			$50$ & 151800 & 149717 & 148129 & 146721 & 144978 & 5353 & 5363 & 5377 & 5415 & 5588 \\ 
			& (1.000) & (0.999) & (1.000) & (1.000) & (1.000) & (1.000) & (1.000) & (1.000) & (1.000) & (1.000)    \\  \hline
			$60$ & 143253 & 141782 & 140897 & 139878 & 138880 & 6289 & 6305 & 6331 & 6372 & 6546 \\
			& (0.999) & (1.000) & (0.999) & (1.000) & (1.000) & (1.000) & (1.000) & (1.000) & (1.000) & (1.000)  \\  \hline
			$70$ & 137727 & 136740 & 136153 & 135654 & 134652 & 7244 & 7266 & 7300 & 7342 & 7516 \\
			& (1.000) & (1.000) & (1.000) & (1.000) & (1.000) & (1.000) & (1.000) & (1.000) & (1.000) & (1.000)   \\  \hline
			$80$ & 133968 & 133208 & 132729 & 132532 & 132118 & 8215 & 8240 & 8271 & 8317 & 8491 \\
			& (1.000) & (0.999) & (1.000) & (0.999) & (1.000) & (1.000) & (1.000) & (1.000) & (1.000) & (1.000)  \\  \hline
			$90$ & 131146 & 130578 & 130667 & 130214 & 129938 & 9192 & 9216 & 9247 & 9299 & 9471  \\  
			& (0.999) & (0.999) & (1.000) & (1.000) & (1.000) & (1.000) & (1.000) & (1.000) & (1.000) & (1.000) \\    \hline
			$100$ & 129252 & 129003 & 128889 & 128728 & 128515 & 10173 & 10197 & 10232 & 10280 & 10455  \\
			& (1.000) & (1.000) & (0.999) & (1.000) & (0.999) & (1.000) & (1.000) & (1.000) & (1.000) & (1.000)  \\    \hline
			$150$ & 125458 & 125549 & 125814 & 125981 & 126119 & 15097 & 15125 & 15159 & 15207 & 15382  \\
			& (1.000) & (1.000) & (1.000) & (0.999) & (0.999) & (1.000) & (1.000) & (1.000) & (1.000) & (1.000)  \\   \hline
			$200$ & 126032 & 126187 & 126509 & 126839 & 127153 & 20037 & 20065 & 20099 & 20147 & 20321  \\
			& (1.000) & (0.999) & (0.999) & (1.000) & (0.999) & (1.000) & (1.000) & (1.000) & (1.000) & (1.000)   \\    \hline
			$300$ & 131324 & 131649 & 132133 & 132599 & 133061 & 29929 & 29957 & 29991 & 30038 & 30212  \\
			& (1.000) & (1.000) & (1.000) & (1.000) & (1.000) & (1.000) & (1.000) & (1.000) & (1.000) & (1.000) \\    \hline
			$400$ & 138912 & 139278 & 139897 & 140179 & 140648 & 39823 & 39853 & 39887 & 39934 & 40108  \\
			& (1.000) & (1.000) & (0.999) & (1.000) & (1.000) & (1.000) & (1.000) & (1.000) & (1.000) & (1.000) \\    \hline
			$500$ & 147376 & 147814 & 148375 & 148784 & 149319 & 49722 & 49749 & 49784 & 49832 & 50005  \\
			& (0.999) & (1.000) & (1.000) & (1.000) & (0.999) & (1.000) & (1.000) & (1.000) & (1.000) & (1.000)   \\    \hline
			$600$ & 156413 & 156755 & 157421 & 157881 & 158371 & 59621 & 59649 & 59684 & 59731 & 59904  \\
			& (0.999) & (0.999) & (1.000) & (1.000) & (1.000) & (1.000) & (1.000) & (1.000) & (1.000) & (1.000)    \\    \hline
			$700$ & 165650 & 166142 & 166680 & 167054 & 167536 & 69520 & 69548 & 69581 & 69629 & 69803 \\
			& (1.000) & (1.000) & (1.000) & (1.000) & (1.000) & (1.000) & (1.000) & (1.000) & (1.000) & (1.000) \\    \hline
			$800$ & 174916 & 175515 & 176136 & 176467 & 176912 & 79420 & 79446 & 79482 & 79529 & 79702 \\
			& (0.999) & (1.000) & (1.000) & (1.000) & (1.000) & (1.000) & (1.000) & (1.000) & (1.000) & (1.000) \\    \hline
			$900$ & 184532 & 185119 & 185516 & 186053 & 186557 & 89318 & 89346 & 89380 & 89428 & 89601 \\
			& (1.000) & (0.999) & (0.999) & (1.000) & (1.000) & (1.000) & (1.000) & (1.000) & (1.000) & (1.000) \\    \hline
			$1000$ & 194117 & 194651 & 195339 & 195676 & 196310 & 99218 & 99246 & 99280 & 99328 & 99501  \\
			& (1.000) & (1.000) & (1.000) & (0.999) & (1.000) & (1.000) & (1.000) & (1.000) & (1.000) & (1.000)   \\  
			\bottomrule    
	\end{tabular}}
	\caption{Estimated OBS and PCD (in parentheses) of ${\cal IZE}$ with $n_0\in \{10, 20, 30, 40, 50, 60, 70, 80, 90, 100, 150, 200, 300, 400, 500, 600, 700, 800, 900, 1000\}, n_0'\in \{n_0-20, n_0-15, n_0-10, n_0-5, n_0\}$, and $\nu=1$ under the SM and CV configurations with $k=99$ systems and $s=4$ constraints.}
	\label{tab:DiffCombInitSample_IZRE_nu1}
\end{table}

\begin{table}[H]
	\centering
		\begin{tabular}{ c | c cc | c cc}
			\toprule
			& \multicolumn{3}{c|}{$d_\ell=0.02$} & \multicolumn{3}{c}{$d_\ell=0.5$}   \\  
			& ${\cal F}_B$ & ${\cal IZR}$ & ${\cal IZR}$ & ${\cal F}_B$ & ${\cal IZR}$ & ${\cal IZR}$  \\  
			$n_0$ & & $T=2, \xi=2$ & $T=2, \xi=3$ & & $T=2, \xi=2$ & $T=2, \xi=3$ \\  \midrule
			$10$ & 587192 & 424793 & 381433 & 27374 & 16539 & 11149  \\
			& (1.000) & (0.999) & (1.000) & (1.000) & (1.000) & (1.000) \\   \cline{1 - 7}
			$20$ & 328495 & 223273 & 199134 & 15378 & 8824 & 6092  \\ 
			& (0.999) & (1.000) & (0.999) & (1.000) & (1.000) & (1.000) \\   \cline{1 - 7}
			$30$ & 279600 & 186483 & 165654 & 13141 & 7722 & 5707  \\  
			& (0.999) & (1.000) & (1.000) & (1.000) & (1.000) & (1.000) \\    \cline{1 - 7}
			$40$ & 259272 & 171437 & 152048 & 12395 & 7717 & 6077 \\ 
			& (1.000) & (1.000) & (0.999) & (1.000) & (1.000) & (1.000) \\    \cline{1 - 7}
			$50$ & 248162 & 163213 & 144637 & 12268 & 8118 & 6691 \\ 
			& (1.000) & (0.999) & (0.999) & (1.000) & (1.000) & (1.000)  \\   \cline{1 - 7}
			$60$ & 241174 & 158214 & 139881 & 12469 & 8695 & 7423 \\
			& (1.000) & (0.999) & (0.999) & (1.000) & (1.000) & (1.000)  \\  
			\cline{1 - 7}
			$70$ & 236591 & 154702 & 136598 & 12867 & 9375 & 8219 \\
			& (0.999) & (0.999) & (0.999) & (1.000) & (1.000) & (1.000)  \\  \cline{1 - 7}
			$80$ & 233079 & 152212 & 134618 & 13385 & 10117 & 9057 \\
			& (0.999) & (1.000) & (1.000) & (1.000) & (1.000) & (1.000)  \\
			\cline{1 - 7}
			$90$ & 230534 & 150141 & 132727 & 13983 & 10906 & 9920 \\
			& (1.000) & (1.000) & (0.999) & (1.000) & (1.000) & (1.000)  \\ 
			\cline{1 - 7}
			$100$ & 228425 & 148734 & 131377 & 14639 & 11725 & 10804 \\
			& (1.000) & (0.999) & (0.999) & (1.000) & (1.000) & (1.000)  \\  \cline{1 - 7}
			$150$ & 222488 & 144360 & 127401 & 18434 & 16088 & 15394 \\  
			& (1.000) & (1.000) & (1.000) & (1.000) & (1.000) & (1.000) \\   \cline{1 - 7}
			$200$ & 219845 & 142375 & 125634 & 22649 & 20680 & 20139 \\  
			& (1.000) & (1.000) & (0.999) & (1.000) & (1.000) & (1.000) \\   \cline{1 - 7}
			$300$ & 217067 & 140324 & 124901 & 31623 & 30164 & 29775 \\
			& (1.000) & (1.000) & (1.000) & (1.000) & (1.000) & (1.000) \\ \cline{1 - 7}
			$400$ & 215710 & 140214 & 127456 & 40919 & 39770 & 39602 \\
			& (1.000) & (0.999) & (0.999) & (1.000) & (1.000) & (1.000) \\ \cline{1 - 7}
			$500$ & 214945 & 142119 & 131778 & 50443 & 49524 & 49500 \\  
			& (1.000) & (1.000) & (0.999) & (1.000) & (1.000) & (1.000) \\   \cline{1 - 7}
			$600$ & 214525 & 145761 & 137438 & 60035 & 59400 & 59400  \\
			& (0.999) & (1.000) & (1.000) & (1.000) & (1.000) & (1.000)  \\ \cline{1 - 7}
			$700$ & 214868 & 150473 & 143866 & 69639 & 69300 & 69300  \\
			& (1.000) & (0.999) & (0.999) & (1.000) & (1.000) & (1.000) \\ \cline{1 - 7}
			$800$ & 216656 & 155835 & 150513 & 79312 & 79200 & 79200  \\
			& (1.000) & (0.999) & (1.000) & (1.000) & (1.000) & (1.000) \\ \cline{1 - 7}
			$900$ & 219265 & 161963 & 157966 & 89117 & 89100 & 89100  \\
			& (1.000) & (1.000) & (0.999) & (1.000) & (1.000) & (1.000)  \\ \cline{1 - 7}
			$1000$ & 222928 & 168436 & 165423 & 99001 & 99000 & 99000 \\  
			& (0.999) & (1.000) & (1.000) & (1.000) & (1.000) & (1.000) \\
			\bottomrule    
		\end{tabular}
	\caption{Estimated OBS and PCD (in parentheses) of ${\cal F}_B$ and ${\cal IZR}$ with $n_0\in \{10, 20, 30, 40, 50, 60, 70, 80, 90, 100, 150, 200, 300, 400, 500, 600, 700,800,900, 1000\}$ under the SM and CV configurations with $k=99$ systems and $s=4$ constraints.}
	\label{tab:DiffCombInitSample_FB_IZR}
\end{table}

\section{Additional Results for Section \ref{sec:StatisticalValidity}}
\label{sec:ValidityAdditionalResults}

In this section, we present additional experimental results that further support the statistical validity of our proposed procedures. Specifically, Section~\ref{sec:ValidityAdditionalResults_SingleSys} provides results for the case including a single system with a single constraint, while Section~\ref{sec:ValidityAdditionalResults_MultiSys} presents results for the case involving multiple systems with multiple constraints.

\subsection{Single System with Single Constraint}
\label{sec:ValidityAdditionalResults_SingleSys}

In this section, we first present additional simulation results for a single system (i.e., $k=\underline{b}=\overline{b}=1$) with a single constraint (i.e., $s=m=1$) with additional values $d_1\in \{0.05, 0.1, 1\}$ in Table \ref{tab:Validity_k1s1_FullResults}. 

\begin{table}[h!]
	\centering
		\begin{tabular}{  c | c | cc|c}
			\toprule
			$d_1$ & ${\cal F}_B$ & \multicolumn{2}{c|}{${\cal IZR}$} & ${\cal IZE}$  \\
			& & $T=2, \xi=2$ & $T=2, \xi=3$ &  \\ \midrule 
			0.05 & 2174 & 1800 & 1850 & 2262 \\  
			& (0.999) & (0.998) & (0.996) & (0.995) \\  \hline
			0.1 & 1184 & 882 & 676 & 992 \\ 
			& (1.000) & (1.000) & (1.000) & (0.999) \\ \hline
			1 & 130 & 89 & 60 & 29 \\
			& (1.000) & (1.000) & (1.000) & (1.000)  \\
			\bottomrule
		\end{tabular}
	\caption{Estimated OBS and PCD (in parentheses) of ${\cal F}_B, {\cal IZR}$, and ${\cal IZE}$ under the CM configuration with $k=s=1$ and $d_\ell\in \{0.05, 0.1, 1\}$. }
	\label{tab:Validity_k1s1_FullResults}
\end{table}

A similar pattern emerges as discussed in Section~\ref{sec:Validity_SingleSysSingleConstr}. For a moderately difficult problem (i.e., $d_1 = 0.05$), ${\cal IZE}$ does not outperform ${\cal F}_B$, while ${\cal IZR}$ does offer some improvement, though the savings are not as substantial as in easier cases (i.e., $d_1 \in \{0.1, 1\}$). In contrast, both ${\cal IZE}$ and ${\cal IZR}$ consistently outperform ${\cal F}_B$ when the problem is easier (i.e., $d_1 \in \{0.1, 1\}$).

\subsection{Multiple Systems with Multiple Constraints}
\label{sec:ValidityAdditionalResults_MultiSys}

In this section, we provided additional results that demonstrate the statistical validity of our proposed procedures when multiple systems and multiple constraints are considered. Tables \ref{tab:Validity_k99s4_CV}--\ref{tab:Validity_k99s4_DV_S} show the results under variance configurations CV, IV-C, DV-C, IV-S, and DV-S, respectively. 

\begin{table}[!h]
	\centering
	{\footnotesize
			\begin{tabular}{ c| c|c | c | cc|c}
				\toprule
				& $(\underline{b}, \overline{b})$ & $d_1$ & ${\cal F}_B$ & \multicolumn{2}{c|}{${\cal IZR}$} & ${\cal IZE}$  \\
				& & & & $T=2, \xi=2$ & $T=2, \xi=3$ &  \\ \midrule  
				\multirow{40}{*}{CV} 
				& \multirow{10}{*}{$(55,77)$}
				& 0.02 & 2411999 & 2378234 & 2665526 & 2661448 \\
				& & & (0.975) & (0.977) & (0.978) & (0.974) \\ \cline{3-7}
				&  & 0.05 & 1195552 & 819634 & 848221 & 1168953 \\
				& & & (1.000) & (1.000) & (1.000) & (0.998) \\ \cline{3-7}
				&  & 0.1 & 649026 & 405758 & 302751 & 512758 \\
				& & & (1.000) & (1.000) & (1.000) & (1.000) \\ \cline{3-7}
				&  & 0.5 & 139456 & 80537 & 54910 & 33538 \\
				& & & (1.000) & (1.000) & (1.000) & (1.000) \\ \cline{3-7}
				&  & 1 & 70393 & 40240 & 27138 & 7149 \\
				& & & (1.000) & (1.000) & (1.000) & (1.000) \\ \cline{2-7}
				& \multirow{10}{*}{$(33,66)$}
				& 0.02 & 2063068 & 1939161 & 2251000 & 2240086 \\
				& & & (0.984) & (0.984) & (0.983) & (0.984) \\ \cline{3-7}
				&  & 0.05 & 1034686 & 699583 & 680371 & 964025 \\
				& & & (1.000) & (1.000) & (1.000) & (0.999) \\ \cline{3-7}
				&  & 0.1 & 565720 & 350898 & 259779 & 405280 \\
				& & & (1.000) & (1.000) & (1.000) & (1.000) \\ \cline{3-7}
				&  & 0.5 & 122310 & 70524 & 47981 & 26003 \\
				& & & (1.000) & (1.000) & (1.000) & (1.000) \\ \cline{3-7}
				&  & 1 & 61807 & 35307 & 23789 & 6245 \\
				& & & (1.000) & (1.000) & (1.000) & (1.000) \\ \cline{2-7}
				& \multirow{10}{*}{$(22,44)$}
				& 0.02 & 1832425 & 1653390 & 1973288 & 1961669 \\
				& & & (0.988) & (0.991) & (0.988) & (0.991) \\ \cline{3-7}
				&  & 0.05 & 928937 & 620968 & 575724 & 828121 \\
				& & & (1.000) & (1.000) & (1.000) & (1.000) \\ \cline{3-7}
				&  & 0.1 & 510125 & 314668 & 231234 & 336339 \\
				& & & (1.000) & (1.000) & (1.000) & (1.000) \\ \cline{3-7}
				&  & 0.5 & 110874 & 63861 & 43422 & 21560 \\
				& & & (1.000) & (1.000) & (1.000) & (1.000) \\ \cline{3-7}
				&  & 1 & 56052 & 32017 & 21562 & 5661 \\
				& & & (1.000) & (1.000) & (1.000) & (1.000) \\ \cline{2-7}
				& \multirow{10}{*}{$(22,77)$}
				& 0.02 & 1944820 & 1784879 & 2116229 & 2099287 \\
				& & & (0.989) & (0.990) & (0.988) & (0.991) \\ \cline{3-7}
				&  & 0.05 & 980878 & 659199 & 617384 & 894447 \\
				& & & (1.000) & (1.000) & (1.000) & (0.999) \\ \cline{3-7}
				&  & 0.1 & 537386 & 332548 & 245058 & 366285 \\
				& & & (1.000) & (1.000) & (1.000) & (1.000) \\ \cline{3-7}
				&  & 0.5 & 116521 & 67128 & 45660 & 22914 \\
				& & & (1.000) & (1.000) & (1.000) & (1.000) \\ \cline{3-7}
				&  & 1 & 58897 & 33626 & 22663 & 5915 \\
				& & & (1.000) & (1.000) & (1.000) & (1.000) \\
				\bottomrule
			\end{tabular}
		}
		\caption{Estimated OBS and PCD (in parentheses) of ${\cal F}_B, {\cal IZR}$, and ${\cal IZE}$ under the CM configuration with $k=99$ and $s=4$. The variance configuration is CV. }
		\label{tab:Validity_k99s4_CV}
	\end{table}
	
	\begin{table}[!h]
		\centering
		{\footnotesize
				\begin{tabular}{ c| c|c | c | cc|c}
					\toprule
					& $(\underline{b}, \overline{b})$ & $d_1$ & ${\cal F}_B$ & \multicolumn{2}{c|}{${\cal IZR}$} & ${\cal IZE}$  \\
					& & & & $T=2, \xi=2$ & $T=2, \xi=3$ &  \\ \midrule  
					\multirow{40}{*}{IV-C} 
					& \multirow{10}{*}{$(55,77)$}
					& 0.02 & 2828498 & 2759333 & 3139220 & 3134028 \\
					& & & (0.971) & (0.974) & (0.974) & (0.973) \\ \cline{3-7}
					&  & 0.05 & 1406653 & 960391 & 963085 & 1394123 \\
					& & & (1.000) & (1.000) & (1.000) & (0.999) \\ \cline{3-7}
					&  & 0.1 & 764273 & 476986 & 355300 & 614572 \\
					& & & (1.000) & (1.000) & (1.000) & (1.000) \\ \cline{3-7}
					&  & 0.5 & 164525 & 94870 & 64662 & 40073 \\
					& & & (1.000) & (1.000) & (1.000) & (1.000) \\ \cline{3-7}
					&  & 1 & 83015 & 47431 & 31986 & 8379 \\
					& & & (1.000) & (1.000) & (1.000) & (1.000) \\ \cline{2-7}
					& \multirow{10}{*}{$(33,66)$}
					& 0.02 & 2323070 & 2201370 & 2558362 & 2549204 \\
					& & & (0.985) & (0.984) & (0.985) & (0.986) \\ \cline{3-7}
					&  & 0.05 & 1162904 & 788549 & 759741 & 1116529 \\
					& & & (1.000) & (1.000) & (1.000) & (1.000) \\ \cline{3-7}
					&  & 0.1 & 634721 & 394494 & 292363 & 477229 \\
					& & & (1.000) & (1.000) & (1.000) & (1.000) \\ \cline{3-7}
					&  & 0.5 & 136933 & 78920 & 53774 & 30302 \\
					& & & (1.000) & (1.000) & (1.000) & (1.000) \\ \cline{3-7}
					&  & 1 & 69125 & 39509 & 26614 & 7027 \\
					& & & (1.000) & (1.000) & (1.000) & (1.000) \\ \cline{2-7}
					& \multirow{10}{*}{$(22,44)$}
					& 0.02 & 1825603 & 1707591 & 1994788 & 1982775 \\
					& & & (0.989) & (0.989) & (0.988) & (0.990) \\ \cline{3-7}
					&  & 0.05 & 916224 & 619675 & 587566 & 857221 \\
					& & & (1.000) & (1.000) & (1.000) & (0.999) \\ \cline{3-7}
					&  & 0.1 & 500036 & 310644 & 229749 & 361554 \\
					& & & (1.000) & (1.000) & (1.000) & (1.000) \\ \cline{3-7}
					&  & 0.5 & 108058 & 62285 & 42436 & 23078 \\
					& & & (1.000) & (1.000) & (1.000) & (1.000) \\ \cline{3-7}
					&  & 1 & 54561 & 31167 & 21017 & 5790 \\
					& & & (1.000) & (1.000) & (1.000) & (1.000) \\ \cline{2-7}
					& \multirow{10}{*}{$(22,77)$}
					& 0.02 & 2317023 & 2136943 & 2532391 & 2527129 \\
					& & & (0.989) & (0.991) & (0.991) & (0.991) \\ \cline{3-7}
					&  & 0.05 & 1167078 & 785672 & 729099 & 1095917 \\
					& & & (1.000) & (1.000) & (1.000) & (1.000) \\ \cline{3-7}
					&  & 0.1 & 638637 & 395548 & 291668 & 455465 \\
					& & & (1.000) & (1.000) & (1.000) & (1.000) \\ \cline{3-7}
					&  & 0.5 & 138188 & 79689 & 54211 & 27820 \\
					& & & (1.000) & (1.000) & (1.000) & (1.000) \\ \cline{3-7}
					&  & 1 & 69858 & 39894 & 26880 & 6880 \\
					& & & (1.000) & (1.000) & (1.000) & (1.000) \\
					\bottomrule
				\end{tabular}
			}
			\caption{Estimated OBS and PCD (in parentheses) of ${\cal F}_B, {\cal IZR}$, and ${\cal IZE}$ under the CM configuration with $k=99$ and $s=4$. The variance configuration is IV-C. }
			\label{tab:Validity_k99s4_IV_C}
		\end{table}
		
		\begin{table}[!h]
			\centering
			{\footnotesize
					\begin{tabular}{ c| c|c | c | cc|c}
						\toprule
						& $(\underline{b}, \overline{b})$ & $d_1$ & ${\cal F}_B$ & \multicolumn{2}{c|}{${\cal IZR}$} & ${\cal IZE}$  \\
						& & & & $T=2, \xi=2$ & $T=2, \xi=3$ &  \\ \midrule  
						\multirow{40}{*}{DV-C} 
						& \multirow{10}{*}{$(55,77)$}
						& 0.02 & 2504060 & 2480002 & 2787762 & 2777960 \\
						& & & (0.971) & (0.976) & (0.975) & (0.975) \\ \cline{3-7}
						&  & 0.05 & 1238584 & 851097 & 870670 & 1238052 \\
						& & & (1.000) & (1.000) & (1.000) & (0.999) \\ \cline{3-7}
						&  & 0.1 & 672712 & 420723 & 314014 & 554013 \\
						& & & (1.000) & (1.000) & (1.000) & (1.000) \\ \cline{3-7}
						&  & 0.5 & 144284 & 83357 & 56836 & 36989 \\
						& & & (1.000) & (1.000) & (1.000) & (1.000) \\ \cline{3-7}
						&  & 1 & 72868 & 41632 & 28086 & 7656 \\
						& & & (1.000) & (1.000) & (1.000) & (1.000) \\ \cline{2-7}
						& \multirow{10}{*}{$(33,66)$}
						& 0.02 & 1835631 & 1783829 & 2027587 & 2014044 \\
						& & & (0.983) & (0.984) & (0.984) & (0.983) \\ \cline{3-7}
						&  & 0.05 & 912466 & 623399 & 622200 & 880839 \\
						& & & (1.000) & (1.000) & (1.000) & (1.000) \\ \cline{3-7}
						&  & 0.1 & 496159 & 309771 & 230469 & 386504 \\
						& & & (1.000) & (1.000) & (1.000) & (1.000) \\ \cline{3-7}
						&  & 0.5 & 106737 & 61605 & 41998 & 25676 \\
						& & & (1.000) & (1.000) & (1.000) & (1.000) \\ \cline{3-7}
						&  & 1 & 53896 & 30816 & 20787 & 5929 \\
						& & & (1.000) & (1.000) & (1.000) & (1.000) \\ \cline{2-7}
						& \multirow{10}{*}{$(22,44)$}
						& 0.02 & 1501779 & 1429905 & 1641689 & 1626197 \\
						& & & (0.988) & (0.990) & (0.989) & (0.989) \\ \cline{3-7}
						&  & 0.05 & 748982 & 509608 & 495159 & 700677 \\
						& & & (1.000) & (1.000) & (1.000) & (1.000) \\ \cline{3-7}
						&  & 0.1 & 408155 & 254332 & 188703 & 300760 \\
						& & & (1.000) & (1.000) & (1.000) & (1.000) \\ \cline{3-7}
						&  & 0.5 & 87920 & 50761 & 34558 & 19975 \\
						& & & (1.000) & (1.000) & (1.000) & (1.000) \\ \cline{3-7}
						&  & 1 & 44428 & 25372 & 17099 & 5062 \\
						& & & (1.000) & (1.000) & (1.000) & (1.000) \\ \cline{2-7}
						& \multirow{10}{*}{$(22,77)$}
						& 0.02 & 1503061 & 1439626 & 1655072 & 1631740 \\
						& & & (0.990) & (0.990) & (0.991) & (0.989) \\ \cline{3-7}
						&  & 0.05 & 749756 & 510468 & 498795 & 703467 \\
						& & & (1.000) & (1.000) & (1.000) & (1.000) \\ \cline{3-7}
						&  & 0.1 & 408053 & 254416 & 188934 & 303802 \\
						& & & (1.000) & (1.000) & (1.000) & (1.000) \\ \cline{3-7}
						&  & 0.5 & 87974 & 50767 & 34587 & 20036 \\
						& & & (1.000) & (1.000) & (1.000) & (1.000) \\ \cline{3-7}
						&  & 1 & 44441 & 25395 & 17127 & 5066 \\
						& & & (1.000) & (1.000) & (1.000) & (1.000) \\
						\bottomrule
					\end{tabular}
				}
				\caption{Estimated OBS and PCD (in parentheses) of ${\cal F}_B, {\cal IZR}$, and ${\cal IZE}$ under the CM configuration with $k=99$ and $s=4$. The variance configuration is DV-C. }
				\label{tab:Validity_k99s4_DV_C}
			\end{table}
			
			\begin{table}[!h]
				\centering
				{\footnotesize
						\begin{tabular}{ c| c|c | c | cc|c}
							\toprule
							& $(\underline{b}, \overline{b})$ & $d_1$ & ${\cal F}_B$ & \multicolumn{2}{c|}{${\cal IZR}$} & ${\cal IZE}$  \\
							& & & & $T=2, \xi=2$ & $T=2, \xi=3$ &  \\ \midrule  
							\multirow{40}{*}{IV-S} 
							& \multirow{10}{*}{$(55,77)$}
							& 0.02 & 2011676 & 1874603 & 2188866 & 2189975 \\
							& & & (0.974) & (0.977) & (0.973) & (0.978) \\ \cline{3-7}
							&  & 0.05 & 1011197 & 682296 & 657433 & 942343 \\
							& & & (1.000) & (1.000) & (1.000) & (1.000) \\ \cline{3-7}
							&  & 0.1 & 553329 & 342885 & 253242 & 387773 \\
							& & & (1.000) & (1.000) & (1.000) & (1.000) \\ \cline{3-7}
							&  & 0.5 & 119730 & 69037 & 46959 & 23128 \\
							& & & (1.000) & (1.000) & (1.000) & (1.000) \\ \cline{3-7}
							&  & 1 & 60509 & 34576 & 23301 & 6025 \\
							& & & (1.000) & (1.000) & (1.000) & (1.000) \\ \cline{2-7}
							& \multirow{10}{*}{$(33,66)$}
							& 0.02 & 1680646 & 1460462 & 1792751 & 1790352 \\
							& & & (0.984) & (0.985) & (0.985) & (0.985) \\ \cline{3-7}
							&  & 0.05 & 859096 & 568772 & 500265 & 747298 \\
							& & & (1.000) & (1.000) & (1.000) & (1.000) \\ \cline{3-7}
							&  & 0.1 & 473965 & 290861 & 212256 & 286512 \\
							& & & (1.000) & (1.000) & (1.000) & (1.000) \\ \cline{3-7}
							&  & 0.5 & 103444 & 59514 & 40428 & 16504 \\
							& & & (1.000) & (1.000) & (1.000) & (1.000) \\ \cline{3-7}
							&  & 1 & 52327 & 29868 & 20116 & 5185 \\
							& & & (1.000) & (1.000) & (1.000) & (1.000) \\ \cline{2-7}
							& \multirow{10}{*}{$(22,44)$}
							& 0.02 & 1509947 & 1251412 & 1584134 & 1584761 \\
							& & & (0.989) & (0.991) & (0.989) & (0.991) \\ \cline{3-7}
							&  & 0.05 & 780116 & 510219 & 426558 & 648274 \\
							& & & (1.000) & (1.000) & (1.000) & (1.000) \\ \cline{3-7}
							&  & 0.1 & 432590 & 264255 & 191384 & 239225 \\
							& & & (1.000) & (1.000) & (1.000) & (1.000) \\ \cline{3-7}
							&  & 0.5 & 94902 & 54601 & 37020 & 14113 \\
							& & & (1.000) & (1.000) & (1.000) & (1.000) \\ \cline{3-7}
							&  & 1 & 48067 & 27426 & 18486 & 4793 \\
							& & & (1.000) & (1.000) & (1.000) & (1.000) \\ \cline{2-7}
							& \multirow{10}{*}{$(22,77)$}
							& 0.02 & 1646265 & 1412007 & 1760405 & 1751802 \\
							& & & (0.989) & (0.990) & (0.991) & (0.989) \\ \cline{3-7}
							&  & 0.05 & 843504 & 556878 & 477052 & 730473 \\
							& & & (1.000) & (1.000) & (1.000) & (1.000) \\ \cline{3-7}
							&  & 0.1 & 466135 & 285954 & 208149 & 277885 \\
							& & & (1.000) & (1.000) & (1.000) & (1.000) \\ \cline{3-7}
							&  & 0.5 & 101786 & 58582 & 39756 & 15881 \\
							& & & (1.000) & (1.000) & (1.000) & (1.000) \\ \cline{3-7}
							&  & 1 & 51547 & 29407 & 19811 & 5114 \\
							& & & (1.000) & (1.000) & (1.000) & (1.000) \\
							\bottomrule
						\end{tabular}
					}
					\caption{Estimated OBS and PCD (in parentheses) of ${\cal F}_B, {\cal IZR}$, and ${\cal IZE}$ under the CM configuration with $k=99$ and $s=4$. The variance configuration is IV-S. }
					\label{tab:Validity_k99s4_IV_S}
				\end{table}
				
				\begin{table}[!h]
					\centering
					{\footnotesize
							\begin{tabular}{ c| c|c | c | cc|c}
								\toprule
								& $(\underline{b}, \overline{b})$ & $d_1$ & ${\cal F}_B$ & \multicolumn{2}{c|}{${\cal IZR}$} & ${\cal IZE}$  \\
								& & & & $T=2, \xi=2$ & $T=2, \xi=3$ &  \\ \midrule  
								\multirow{40}{*}{DV-S} 
								& \multirow{10}{*}{$(55,77)$}
								& 0.02 & 2815029 & 2881582 & 3142071 & 3150378 \\
								& & & (0.973) & (0.976) & (0.973) & (0.976) \\ \cline{3-7}
								&  & 0.05 & 1379943 & 957361 & 1039168 & 1425129 \\
								& & & (1.000) & (1.000) & (1.000) & (0.999) \\ \cline{3-7}
								&  & 0.1 & 745140 & 468501 & 352075 & 665565 \\
								& & & (1.000) & (1.000) & (1.000) & (1.000) \\ \cline{3-7}
								&  & 0.5 & 159250 & 92036 & 62829 & 49385 \\
								& & & (1.000) & (1.000) & (1.000) & (1.000) \\ \cline{3-7}
								&  & 1 & 80337 & 45936 & 31021 & 9043 \\
								& & & (1.000) & (1.000) & (1.000) & (1.000) \\ \cline{2-7}
								& \multirow{10}{*}{$(33,66)$}
								& 0.02 & 2446196 & 2419192 & 2708883 & 2708519 \\
								& & & (0.985) & (0.985) & (0.982) & (0.985) \\ \cline{3-7}
								&  & 0.05 & 1210863 & 831292 & 860199 & 1210085 \\
								& & & (1.000) & (1.000) & (1.000) & (0.999) \\ \cline{3-7}
								&  & 0.1 & 657344 & 410913 & 306932 & 550136 \\
								& & & (1.000) & (1.000) & (1.000) & (1.000) \\ \cline{3-7}
								&  & 0.5 & 141151 & 81522 & 55585 & 40731 \\
								& & & (1.000) & (1.000) & (1.000) & (1.000) \\ \cline{3-7}
								&  & 1 & 71230 & 40755 & 27524 & 8034 \\
								& & & (1.000) & (1.000) & (1.000) & (1.000) \\ \cline{2-7}
								& \multirow{10}{*}{$(22,44)$}
								& 0.02 & 2154231 & 2054514 & 2359108 & 2358611 \\
								& & & (0.989) & (0.991) & (0.991) & (0.987) \\ \cline{3-7}
								&  & 0.05 & 1076961 & 731295 & 725138 & 1038839 \\
								& & & (1.000) & (1.000) & (1.000) & (1.000) \\ \cline{3-7}
								&  & 0.1 & 587571 & 365454 & 270846 & 459145 \\
								& & & (1.000) & (1.000) & (1.000) & (1.000) \\ \cline{3-7}
								&  & 0.5 & 126739 & 73126 & 49796 & 33520 \\
								& & & (1.000) & (1.000) & (1.000) & (1.000) \\ \cline{3-7}
								&  & 1 & 64070 & 36609 & 24719 & 7157 \\
								& & & (1.000) & (1.000) & (1.000) & (1.000) \\ \cline{2-7}
								& \multirow{10}{*}{$(22,77)$}
								& 0.02 & 2242467 & 2157917 & 2470525 & 2466008 \\
								& & & (0.988) & (0.990) & (0.989) & (0.991) \\ \cline{3-7}
								&  & 0.05 & 1117298 & 761329 & 756590 & 1089027 \\
								& & & (1.000) & (1.000) & (1.000) & (0.999) \\ \cline{3-7}
								&  & 0.1 & 608937 & 379138 & 281840 & 481287 \\
								& & & (1.000) & (1.000) & (1.000) & (1.000) \\ \cline{3-7}
								&  & 0.5 & 131212 & 75709 & 51573 & 34524 \\
								& & & (1.000) & (1.000) & (1.000) & (1.000) \\ \cline{3-7}
								&  & 1 & 66263 & 37869 & 25577 & 7350 \\
								& & & (1.000) & (1.000) & (1.000) & (1.000) \\
								\bottomrule
							\end{tabular}
						}
						\caption{Estimated OBS and PCD (in parentheses) of ${\cal F}_B, {\cal IZR}$, and ${\cal IZE}$ under the CM configuration with $k=99$ and $s=4$. The variance configuration is DV-S. }
						\label{tab:Validity_k99s4_DV_S}
					\end{table}

\clearpage
\newpage
\section{Additional Results for Section \ref{sec:Efficiency}}
\label{sec:EfficiencyAdditional}

In this section, we include additional results that demonstrate the efficiency of our proposed procedures. Tables \ref{tab:Efficiency_k99s4_IV_C}--\ref{tab:Efficiency_k99s4_DV_S} present the results under the IV-C, DV-C, IV-S, and DV-S variance configuration, respectively. 

    \begin{table}[!h]
	\centering
	\resizebox{\textwidth}{!}{
	\begin{tabular}{ c|c | c | cc|c || c|c | c | cc|c}
	\toprule
	$(\underline{b}, \overline{b})$ & $d_1$ & ${\cal F}_B$ & \multicolumn{2}{c|}{${\cal IZR}$} & ${\cal IZE}$ & $(\underline{b}, \overline{b})$ & $d_1$ & ${\cal F}_B$ & \multicolumn{2}{c|}{${\cal IZR}$} & ${\cal IZE}$  \\
	& & & $T=2, \xi=2$ & $T=2, \xi=3$ & & & & & $T=2, \xi=2$ & $T=2, \xi=3$ &  \\ \midrule  
	\multirow{16}{*}{$(55,77)$} & 0.005 & 1131604 & 909225 & 901672 & 1081831 & \multirow{16}{*}{$(33,66)$} & 0.005 & 1094338 & 887628 & 886453 & 1053645 \\ 
	& & (1.000) & (0.999) & (0.999) & (0.999) & & & (0.999) & (1.000) & (0.999) & (0.999) \\ \cline{2-6} \cline{8-12}
	& 0.01 & 675078 & 501205 & 472074 & 559146 & & 0.01 & 657569 & 490860 & 464320 & 551741 \\
	& & (1.000) & (1.000) & (1.000) & (0.999) & & & (1.000) & (1.000) & (0.999) & (0.999) \\ \cline{2-6} \cline{8-12}
	& 0.02 & 377946 & 256418 & 228273 & 256415 & & 0.02 & 368989 & 251068 & 225242 & 256277 \\
	& & (1.000) & (1.000) & (1.000) & (0.999) & & & (1.000) & (1.000) & (0.999) & (1.000) \\ \cline{2-6} \cline{8-12}
	& 0.05 & 164503 & 99845 & 75897 & 78663 & & 0.05 & 161016 & 97685 & 74462 & 78661 \\
	& & (1.000) & (1.000) & (1.000) & (1.000) & & & (1.000) & (1.000) & (1.000) & (1.000) \\ \cline{2-6} \cline{8-12}
	& 0.1 & 85129 & 49980 & 34713 & 28288 & & 0.1 & 83401 & 48917 & 34027 & 28393 \\
	& & (1.000) & (1.000) & (1.000) & (1.000) & & & (1.000) & (1.000) & (1.000) & (1.000) \\ \cline{2-6} \cline{8-12}
	& 0.5 & 17574 & 10084 & 6904 & 3092 & & 0.5 & 17244 & 9974 & 6878 & 3097 \\
	& & (1.000) & (1.000) & (1.000) & (1.000) & & & (1.000) & (1.000) & (1.000) & (1.000) \\ \cline{2-6} \cline{8-12}
	& 1 & 8912 & 5316 & 3908 & 2166 & & 1 & 8839 & 5310 & 3906 & 2166 \\
	& & (1.000) & (1.000) & (1.000) & (1.000) & & & (1.000) & (1.000) & (1.000) & (1.000) \\ \cline{2-6} \cline{8-12}
	& 2 & 4799 & 3251 & 2679 & 2007 & & 2 & 4798 & 3251 & 2678 & 2007 \\
	& & (1.000) & (1.000) & (1.000) & (1.000) & & & (1.000) & (1.000) & (1.000) & (1.000) \\ \hline
	\multirow{16}{*}{$(22,44)$} & 0.005 & 980672 & 818224 & 836951 & 971232 & \multirow{16}{*}{$(22,77)$} & 0.005 & 1061553 & 866142 & 868943 & 1005260  \\
	& & (1.000) & (0.999) & (0.999) & (0.999) & & & (0.999) & (0.999) & (0.999) & (0.999)  \\ \cline{2-6} \cline{8-12}
	& 0.01 & 598181 & 456521 & 440852 & 524597 & & 0.01 & 639271 & 480897 & 456946 & 533476 \\
	& & (0.999) & (1.000) & (1.000) & (0.999) & & & (0.999) & (1.000) & (0.999) & (0.999) \\ \cline{2-6} \cline{8-12}
	& 0.02 & 338934 & 233994 & 213510 & 249454 & & 0.02 & 359827 & 246056 & 221454 & 252284 \\
	& & (0.999) & (1.000) & (0.999) & (1.000) & & & (0.999) & (1.000) & (1.000) & (1.000) \\ \cline{2-6} \cline{8-12}
	& 0.05 & 149031 & 90909 & 69834 & 77818 & & 0.05 & 157285 & 95614 & 73011 & 78140 \\
	& & (1.000) & (1.000) & (1.000) & (1.000) & & & (1.000) & (1.000) & (1.000) & (1.000) \\ \cline{2-6} \cline{8-12}
	& 0.1 & 77250 & 45523 & 31794 & 28345 & & 0.1 & 81441 & 47887 & 33317 & 28325  \\
	& & (1.000) & (1.000) & (1.000) & (1.000) & & & (1.000) & (1.000) & (1.000) & (1.000) \\ \cline{2-6} \cline{8-12}
	& 0.5 & 16200 & 9562 & 6733 & 3094 & & 0.5 & 16847 & 9716 & 6758 & 3097 \\
	& & (1.000) & (1.000) & (1.000) & (1.000) & & & (1.000) & (1.000) & (1.000) & (1.000) \\ \cline{2-6} \cline{8-12}
	& 1 & 8523 & 5250 & 3900 & 2166 & & 1 & 8632 & 5261 & 3898 & 2165 \\
	& & (1.000) & (1.000) & (1.000) & (1.000) & & & (1.000) & (1.000) & (1.000) & (1.000) \\ \cline{2-6} \cline{8-12}
	& 2 & 4767 & 3249 & 2679 & 2007 & & 2 & 4770 & 3250 & 2678 & 2007 \\
	& & (1.000) & (1.000) & (1.000) & (1.000) & & & (1.000) & (1.000) & (1.000) & (1.000) \\ 
    \bottomrule
	\end{tabular}
    }
	\caption{Estimated OBS and PCD (in parentheses) of ${\cal F}_B, {\cal IZR}$, and ${\cal IZE}$ under the SM configuration with $k=99$ and $s=4$. The variance configuration is IV-C. }
    \label{tab:Efficiency_k99s4_IV_C}
	\end{table}
		
	\begin{table}[!h]
	\centering
	\resizebox{\textwidth}{!}{
	\begin{tabular}{ c|c | c | cc|c || c|c | c | cc|c}
	\toprule
	$(\underline{b}, \overline{b})$ & $d_1$ & ${\cal F}_B$ & \multicolumn{2}{c|}{${\cal IZR}$} & ${\cal IZE}$ & $(\underline{b}, \overline{b})$ & $d_1$ & ${\cal F}_B$ & \multicolumn{2}{c|}{${\cal IZR}$} & ${\cal IZE}$ \\
	& & & $T=2, \xi=2$ & $T=2, \xi=3$ & & & & & $T=2, \xi=2$ & $T=2, \xi=3$  \\ \midrule  
	\multirow{16}{*}{$(55,77)$} & 0.005 & 936082 & 751483 & 738377 & 885982 & \multirow{16}{*}{$(33,66)$} & 0.005 & 863828 & 708076 & 708974 & 839216 \\
	& & (0.999) & (1.000) & (0.999) & (0.999) & & & (1.000) & (1.000) & (0.999) & (1.000) \\ \cline{2-6} \cline{8-12}
	& 0.01 & 552540 & 410386 & 385084 & 454820 & & 0.01 & 516660 & 389957 & 371044 & 442222 \\
	& & (1.000) & (0.999) & (0.999) & (1.000) & & & (1.000) & (1.000) & (1.000) & (0.999) \\ \cline{2-6} \cline{8-12}
	& 0.02 & 307169 & 209457 & 186002 & 207989 & & 0.02 & 289089 & 199340 & 179171 & 206108 \\
	& & (0.999) & (0.999) & (1.000) & (0.999) & & & (0.999) & (0.999) & (1.000) & (1.000) \\ \cline{2-6} \cline{8-12}
	& 0.05 & 133081 & 80700 & 61977 & 64048 & & 0.05 & 125883 & 76678 & 59161 & 63727 \\
	& & (1.000) & (1.000) & (1.000) & (1.000) & & & (1.000) & (1.000) & (1.000) & (1.000) \\ \cline{2-6} \cline{8-12}
	& 0.1 & 68607 & 40271 & 28037 & 23560 & & 0.1 & 65029 & 38242 & 26654 & 23590 \\
	& & (1.000) & (1.000) & (1.000) & (1.000) & & & (1.000) & (1.000) & (1.000) & (1.000) \\ \cline{2-6} \cline{8-12}
	& 0.5 & 14161 & 8155 & 5689 & 2887 & & 0.5 & 13486 & 7946 & 5632 & 2891 \\
	& & (1.000) & (1.000) & (1.000) & (1.000) & & & (1.000) & (1.000) & (1.000) & (1.000) \\ \cline{2-6} \cline{8-12}
	& 1 & 7247 & 4467 & 3406 & 2123 & & 1 & 7087 & 4457 & 3407 & 2124 \\
	& & (1.000) & (1.000) & (1.000) & (1.000) & & & (1.000) & (1.000) & (1.000) & (1.000) \\ \cline{2-6} \cline{8-12}
	& 2 & 4075 & 2917 & 2489 & 1999 & & 2 & 4069 & 2913 & 2490 & 1999 \\
	& & (1.000) & (1.000) & (1.000) & (1.000) & & & (1.000) & (1.000) & (1.000) & (1.000) \\ \hline
	\multirow{16}{*}{$(22,44)$} & 0.005 & 785774 & 659957 & 674093 & 775020 & \multirow{16}{*}{$(22,77)$} & 0.005 & 785347 & 660285 & 674961 & 776108  \\
	& & (0.999) & (0.999) & (0.999) & (0.999) & & & (0.999) & (0.999) & (0.999) & (0.999) \\ \cline{2-6} \cline{8-12}
	& 0.01 & 476334 & 366561 & 353269 & 420129 & & 0.01 & 475661 & 365944 & 354018 & 420678 \\
	& & (1.000) & (1.000) & (1.000) & (1.000) & & & (1.000) & (0.999) & (1.000) & (0.999) \\ \cline{2-6} \cline{8-12}
	& 0.02 & 267961 & 186628 & 170710 & 200519 & & 0.02 & 268177 & 187061 & 170574 & 200801 \\
	& & (1.000) & (0.999) & (0.999) & (1.000) & & & (0.999) & (1.000) & (1.000) & (0.999) \\ \cline{2-6} \cline{8-12}
	& 0.05 & 117413 & 71934 & 55863 & 63290 & & 0.05 & 117489 & 71974 & 55903 & 63318 \\
	& & (1.000) & (1.000) & (1.000) & (1.000) & & & (1.000) & (1.000) & (1.000) & (1.000) \\ \cline{2-6} \cline{8-12}
	& 0.1 & 60887 & 35826 & 25100 & 23503 & & 0.1 & 60839 & 35836 & 25098 & 23462 \\
	& & (1.000) & (1.000) & (1.000) & (1.000) & & & (1.000) & (1.000) & (1.000) & (1.000) \\ \cline{2-6} \cline{8-12}
	& 0.5 & 12782 & 7652 & 5509 & 2893 & & 0.5 & 12790 & 7650 & 5507 & 2894 \\
	& & (1.000) & (1.000) & (1.000) & (1.000) & & & (1.000) & (1.000) & (1.000) & (1.000) \\ \cline{2-6} \cline{8-12}
	& 1 & 6854 & 4408 & 3402 & 2124 & & 1 & 6850 & 4407 & 3401 & 2123 \\
	& & (1.000) & (1.000) & (1.000) & (1.000) & & & (1.000) & (1.000) & (1.000) & (1.000) \\ \cline{2-6} \cline{8-12}
	& 2 & 4040 & 2915 & 2490 & 1999 & & 2 & 4040 & 2912 & 2491 & 1999 \\
	& & (1.000) & (1.000) & (1.000) & (1.000) & & & (1.000) & (1.000) & (1.000) & (1.000) \\ 
	\bottomrule
	\end{tabular}
	}
	\caption{Estimated OBS and PCD (in parentheses) of ${\cal F}_B, {\cal IZR}$, and ${\cal IZE}$ under the SM configuration with $k=99$ and $s=4$. The variance configuration is DV-C. }
	\label{tab:Efficiency_k99s4_DV_C}
	\end{table}
			
	\begin{table}[!h]
	\centering
	\resizebox{\textwidth}{!}{
	\begin{tabular}{ c|c | c | cc|c || c|c | c | cc|c}
	\toprule
	$(\underline{b}, \overline{b})$ & $d_1$ & ${\cal F}_B$ & \multicolumn{2}{c|}{${\cal IZR}$} & ${\cal IZE}$ & $(\underline{b}, \overline{b})$ & $d_1$ & ${\cal F}_B$ & \multicolumn{2}{c|}{${\cal IZR}$} & ${\cal IZE}$  \\
	& & & $T=2, \xi=2$ & $T=2, \xi=3$ & & & & & $T=2, \xi=2$ & $T=2, \xi=3$ \\ \midrule  
	\multirow{16}{*}{$(55,77)$} & 0.005 & 1079969 & 886216 & 896656 & 1044563 & \multirow{16}{*}{$(33,66)$} & 0.005 & 811177 & 642250 & 633041 & 712981 \\
	& & (0.999) & (1.000) & (0.999) & (0.999) & & & (0.999) & (1.000) & (0.999) & (0.999) \\ \cline{2-6} \cline{8-12}
	& 0.01 & 654047 & 491167 & 470503 & 548913 & & 0.01 & 485418 & 354538 & 331916 & 362422 \\
	& & (1.000) & (1.000) & (0.999) & (0.999) & & & (1.000) & (0.999) & (1.000) & (0.999) \\ \cline{2-6} \cline{8-12}
	& 0.02 & 370879 & 252784 & 228168 & 257048 & & 0.02 & 273434 & 183052 & 161855 & 166740 \\
	& & (1.000) & (1.000) & (1.000) & (0.999) & & & (0.999) & (1.000) & (1.000) & (0.999) \\ \cline{2-6} \cline{8-12}
	& 0.05 & 163282 & 99238 & 75711 & 79650 & & 0.05 & 119953 & 72468 & 54450 & 50798 \\
	& & (1.000) & (1.000) & (1.000) & (1.000) & & & (1.000) & (1.000) & (1.000) & (1.000) \\ \cline{2-6} \cline{8-12}
	& 0.1 & 84833 & 49856 & 34814 & 28826 & & 0.1 & 62190 & 36471 & 25317 & 18368 \\
	& & (1.000) & (1.000) & (1.000) & (1.000) & & & (1.000) & (1.000) & (1.000) & (1.000) \\ \cline{2-6} \cline{8-12}
	& 0.5 & 17799 & 10355 & 7176 & 3124 & & 0.5 & 13014 & 7554 & 5261 & 2630 \\
	& & (1.000) & (1.000) & (1.000) & (1.000) & & & (1.000) & (1.000) & (1.000) & (1.000) \\ \cline{2-6} \cline{8-12}
	& 1 & 9192 & 5532 & 4028 & 2173 & & 1 & 6716 & 4137 & 3176 & 2094 \\
	& & (1.000) & (1.000) & (1.000) & (1.000) & & & (1.000) & (1.000) & (1.000) & (1.000) \\ \cline{2-6} \cline{8-12}
	& 2 & 4986 & 3325 & 2713 & 2008 & & 2 & 3778 & 2744 & 2381 & 1994 \\
	& & (1.000) & (1.000) & (1.000) & (1.000) & & & (1.000) & (1.000) & (1.000) & (1.000) \\ \hline
	\multirow{16}{*}{$(22,44)$} & 0.005 & 589075 & 455704 & 439611 & 469517 & \multirow{16}{*}{$(22,77)$} & 0.005 & 699393 & 544792 & 533129 & 594719  \\
	& & (1.000) & (1.000) & (0.999) & (1.000) & & & (1.000) & (1.000) & (0.999) & (0.999) \\ \cline{2-6} \cline{8-12}
	& 0.01 & 349120 & 251264 & 231479 & 236464 & & 0.01 & 417916 & 301353 & 279876 & 299468 \\  
	& & (1.000) & (1.000) & (0.999) & (0.999) & & & (1.000) & (0.999) & (1.000) & (0.999) \\ \cline{2-6} \cline{8-12}
	& 0.02 & 195236 & 129441 & 112930 & 108215 & & 0.02 & 235841 & 155795 & 136709 & 136239 \\
	& & (1.000) & (0.999) & (1.000) & (1.000) & & & (1.000) & (0.999) & (0.999) & (1.000) \\ \cline{2-6} \cline{8-12}
	& 0.05 & 85225 & 51319 & 38305 & 33054 & & 0.05 & 103602 & 62357 & 46327 & 41162 \\
	& & (1.000) & (1.000) & (1.000) & (1.000) & & & (1.000) & (1.000) & (1.000) & (1.000) \\ \cline{2-6} \cline{8-12}
	& 0.1 & 44194 & 25843 & 17914 & 12331 & & 0.1 & 53839 & 31465 & 21782 & 14841 \\
	& & (1.000) & (1.000) & (1.000) & (1.000) & & & (1.000) & (1.000) & (1.000) & (1.000) \\ \cline{2-6} \cline{8-12}
	& 0.5 & 9235 & 5399 & 3887 & 2369 & & 0.5 & 11258 & 6579 & 4709 & 2479 \\
	& & (1.000) & (1.000) & (1.000) & (1.000) & & & (1.000) & (1.000) & (1.000) & (1.000) \\ \cline{2-6} \cline{8-12}
	& 1 & 4828 & 3197 & 2632 & 2047 & & 1 & 5887 & 3798 & 2990 & 2072 \\
	& & (1.000) & (1.000) & (1.000) & (1.000) & & & (1.000) & (1.000) & (1.000) & (1.000) \\ \cline{2-6} \cline{8-12}
	& 2 & 2984 & 2383 & 2178 & 1985 & & 2 & 3505 & 2622 & 2313 & 1991 \\
	& & (1.000) & (1.000) & (1.000) & (1.000) & & & (1.000) & (1.000) & (1.000) & (1.000)  \\ 
	\bottomrule
	\end{tabular}
	}
	\caption{Estimated OBS and PCD (in parentheses) of ${\cal F}_B, {\cal IZR}$, and ${\cal IZE}$ under the SM configuration with $k=99$ and $s=4$. The variance configuration is IV-S. }
	\label{tab:Efficiency_k99s4_IV_S}
	\end{table}
				
	\begin{table}[!h]
	\centering
	\resizebox{\textwidth}{!}{
	\begin{tabular}{ c|c | c | cc|c || c|c | c | cc|c}
	\toprule
	$(\underline{b}, \overline{b})$ & $d_1$ & ${\cal F}_B$ & \multicolumn{2}{c|}{${\cal IZR}$} & ${\cal IZE}$ & $(\underline{b}, \overline{b})$ & $d_1$ & ${\cal F}_B$ & \multicolumn{2}{c|}{${\cal IZR}$} & ${\cal IZE}$  \\
	& & & $T=2, \xi=2$ & $T=2, \xi=3$ & & & & & $T=2, \xi=2$ & $T=2, \xi=3$ &  \\ \midrule  
	\multirow{16}{*}{$(55,77)$} & 0.005 & 909360 & 715244 & 695380 & 868496 & \multirow{16}{*}{$(33,66)$} & 0.005 & 1150527 & 943387 & 946603 & 1153649 \\
	& & (1.000) & (0.999) & (0.999) & (1.000) & & & (1.000) & (1.000) & (0.999) & (0.999) \\ \cline{2-6} \cline{8-12}
	& 0.01 & 529522 & 386650 & 357618 & 431926 & & 0.01 & 685882 & 515676 & 490775 & 609485 \\
	& & (0.999) & (0.999) & (0.999) & (1.000) & & & (0.999) & (1.000) & (0.999) & (0.999) \\ \cline{2-6} \cline{8-12}
	& 0.02 & 292649 & 196811 & 172227 & 190463 & & 0.02 & 384204 & 263338 & 236787 & 284332 \\
	& & (1.000) & (0.999) & (0.999) & (0.999) & & & (1.000) & (0.999) & (0.999) & (1.000) \\ \cline{2-6} \cline{8-12}
	& 0.05 & 126219 & 76319 & 58082 & 56180 & & 0.05 & 167449 & 101849 & 78523 & 86733 \\
	& & (1.000) & (1.000) & (1.000) & (1.000) & & & (1.000) & (1.000) & (1.000) & (1.000) \\ \cline{2-6} \cline{8-12}
	& 0.1 & 65120 & 38177 & 26507 & 20418 & & 0.1 & 86652 & 51010 & 35545 & 31447 \\
	& & (1.000) & (1.000) & (1.000) & (1.000) & & & (1.000) & (1.000) & (1.000) & (1.000) \\ \cline{2-6} \cline{8-12}
	& 0.5 & 13546 & 7866 & 5453 & 2725 & & 0.5 & 18090 & 10537 & 7319 & 3254 \\
	& & (1.000) & (1.000) & (1.000) & (1.000) & & & (1.000) & (1.000) & (1.000) & (1.000) \\ \cline{2-6} \cline{8-12}
	& 1 & 6987 & 4241 & 3209 & 2097 & & 1 & 9352 & 5665 & 4120 & 2181 \\
	& & (1.000) & (1.000) & (1.000) & (1.000) & & & (1.000) & (1.000) & (1.000) & (1.000) \\ \cline{2-6} \cline{8-12}
	& 2 & 3853 & 2759 & 2387 & 1995 & & 2 & 5108 & 3383 & 2742 & 2009 \\
	& & (1.000) & (1.000) & (1.000) & (1.000) & & & (1.000) & (1.000) & (1.000) & (1.000) \\ \hline
	\multirow{16}{*}{$(22,44)$} & 0.005 & 1244819 & 1050310 & 1081661 & 1271267 & \multirow{16}{*}{$(22,77)$} & 0.005 & 1165555 & 979579 & 1001591 & 1179982 \\
	& & (1.000) & (0.999) & (0.999) & (0.999) & & & (0.999) & (0.999) & (1.000) & (0.999) \\ \cline{2-6} \cline{8-12}
	& 0.01 & 757133 & 580857 & 564553 & 695634 & & 0.01 & 702473 & 539076 & 521989 & 642505 \\
	& & (1.000) & (0.999) & (1.000) & (0.999) & & & (1.000) & (1.000) & (0.999) & (0.999) \\ \cline{2-6} \cline{8-12}
	& 0.02 & 428513 & 297485 & 271933 & 335713 & & 0.02 & 396139 & 275920 & 251229 & 310176 \\
	& & (1.000) & (1.000) & (0.999) & (1.000) & & & (1.000) & (1.000) & (1.000) & (1.000) \\ \cline{2-6} \cline{8-12}
	& 0.05 & 188692 & 115361 & 89516 & 105551 & & 0.05 & 173326 & 106048 & 82643 & 97976 \\
	& & (1.000) & (1.000) & (1.000) & (1.000) & & & (1.000) & (1.000) & (1.000) & (1.000) \\ \cline{2-6} \cline{8-12}
	& 0.1 & 98059 & 57817 & 40481 & 38288 & & 0.1 & 89878 & 52979 & 37037 & 35940 \\
	& & (1.000) & (1.000) & (1.000) & (1.000) & & & (1.000) & (1.000) & (1.000) & (1.000) \\ \cline{2-6} \cline{8-12}
	& 0.5 & 20597 & 12040 & 8367 & 3616 & & 0.5 & 18798 & 11014 & 7721 & 3510 \\
	& & (1.000) & (1.000) & (1.000) & (1.000) & & & (1.000) & (1.000) & (1.000) & (1.000) \\ \cline{2-6} \cline{8-12}
	& 1 & 10702 & 6459 & 4685 & 2244 & & 1 & 9800 & 6005 & 4416 & 2220 \\
	& & (1.000) & (1.000) & (1.000) & (1.000) & & & (1.000) & (1.000) & (1.000) & (1.000) \\ \cline{2-6} \cline{8-12}
	& 2 & 5820 & 3799 & 2993 & 2019 & & 2 & 5432 & 3624 & 2895 & 2015 \\
	& & (1.000) & (1.000) & (1.000) & (1.000) & & & (1.000) & (1.000) & (1.000) & (1.000) \\ \cline{2-6}
	\bottomrule
	\end{tabular}
	}
	\caption{Estimated OBS and PCD (in parentheses) of ${\cal F}_B, {\cal IZR}$, and ${\cal IZE}$ under the SM configuration with $k=99$ and $s=4$. The variance configuration is DV-S. }
	\label{tab:Efficiency_k99s4_DV_S}
	\end{table}

\clearpage
\newpage
\section{Additional Results for Section \ref{sec:Inventory}}
\label{sec:InventoryAdditional}
In Table~\ref{tab:InventoryResultsAdditional}, we present additional experimental results for ${\cal F}_B$, ${\cal IZR}$, and ${\cal IZE}$ using larger values of $n_0 \in \{30, 40, \ldots, 100\}$ across the three sets of thresholds discussed in Section~\ref{sec:Inventory}. 

\begin{table}[h!]
	\centering
		{\scriptsize \begin{tabular}{ c || c | c c c c}
				\toprule
				& & ${\cal F}_B$ & \multicolumn{2} {c} {${\cal IZR}$} & ${\cal IZE}$  \\   
				$(q_1, q_2)$ & $n_0$ & & $T=2, \xi=2$ & $T=2, \xi=3$  \\ \midrule
				\multirow{ 16}{*}{(0.01, 120)} & \multirow{ 2}{*}{30} & 483188 & 334065 & 290387 & 257448   \\  
				& & (0.999) & (0.998) & (0.996) & (0.999) \\     \cline{2-6}
				& \multirow{ 2}{*}{40} & 455332 & 323488 & 286612 & 264159 \\  
				& & (1.000) & (0.999) & (0.999) & (0.999) \\     \cline{2-6}
				& \multirow{ 2}{*}{50} & 450483 & 328941 & 295994 & 282096 \\  
				& & (1.000) & (1.000) & (1.000) & (0.999)  \\    \cline{2-6}
				& \multirow{ 2}{*}{60} & 456158 & 341728 & 311203 & 304272   \\  
				& & (1.000) & (1.000) & (0.999) & (1.000) \\    \cline{2-6}
				& \multirow{ 2}{*}{70} & 467527 & 358000 & 329400 & 328732  \\  
				& & (1.000) & (1.000) & (1.000) & (1.000)  \\  \cline{2-6}
				& \multirow{ 2}{*}{80} & 482058 & 376743 & 349590 & 354428  \\  
				& & (1.000) & (1.000) & (1.000) & (1.000) \\  \cline{2-6}
				& \multirow{ 2}{*}{90} & 498951 & 396859 & 370979 & 381102   \\  
				& & (1.000) & (1.000) & (1.000) & (1.000) \\  \cline{2-6}
				& \multirow{ 2}{*}{100} & 517328 & 418028 & 393212 & 408246  \\  
				& & (1.000) & (1.000) & (1.000) & (1.000) \\  \cline{2-6}
				\hline
				\multirow{ 16}{*}{(0.05, 125)} & \multirow{ 2}{*}{30} & 1466312 & 1159370 & 1101872 & 1103057  \\  
				& & (1.000) & (1.000) & (1.000) & (1.000) \\     \cline{2-6}
				& \multirow{ 2}{*}{40} & 1332004 & 1051375 & 1003934 & 999451   \\  
				& & (1.000) & (1.000) & (1.000) & (1.000) \\     \cline{2-6}
				& \multirow{ 2}{*}{50} & 1269112 & 1004034 & 962241 & 959318   \\  
				& & (1.000) & (1.000) & (1.000) & (1.000) \\     \cline{2-6}
				& \multirow{ 2}{*}{60} & 1238251 & 983940 & 944195 & 945353   \\  
				& & (1.000) & (1.000) & (1.000) & (1.000) \\    \cline{2-6}
				& \multirow{ 2}{*}{70} & 1224049 & 977437 & 940712 & 945543   \\  
				& & (1.000) & (1.000) & (1.000) & (1.000) \\  \cline{2-6}
				& \multirow{ 2}{*}{80} & 1219171 & 978328 & 943762 & 954026  \\  
				& & (1.000) & (1.000) & (1.000) & (1.000) \\  \cline{2-6}
				& \multirow{ 2}{*}{90} & 1220267 & 985688 & 952359 & 967166   \\  
				& & (1.000) & (1.000) & (1.000) & (1.000) \\  \cline{2-6}
				& \multirow{ 2}{*}{100} & 1226293 & 995894 & 964037 & 983192  \\  
				& & (1.000) & (1.000) & (1.000) & (1.000) \\  \cline{2-6}
				\hline
				\multirow{ 16}{*}{(0.1, 130)} & \multirow{ 2}{*}{30} & 928755 & 670839 & 617263 & 624690   \\  
				& & (1.000) & (1.000) & (1.000) & (1.000) \\     \cline{2-6}
				& \multirow{ 2}{*}{40} & 848349 & 617150 & 571602 & 580856   \\  
				& & (1.000) & (1.000) & (1.000) & (1.000) \\     \cline{2-6}
				& \multirow{ 2}{*}{50} & 814054 & 598337 & 558236 & 571734   \\  
				& & (1.000) & (1.000) & (1.000) & (1.000)  \\ \cline{2-6}
				& \multirow{ 2}{*}{60} & 800225 & 595238 & 559172 & 577904   \\  
				& & (1.000) & (1.000) & (1.000) & (1.000)   \\   \cline{2-6}
				& \multirow{ 2}{*}{70} & 796448 & 601135 & 568184 & 591159   \\  
				& & (1.000) & (1.000) & (1.000) & (1.000) \\  \cline{2-6}
				& \multirow{ 2}{*}{80} & 799736 & 611277 & 581443 & 609046  \\  
				& & (1.000) & (1.000) & (1.000) & (1.000) \\  \cline{2-6}
				& \multirow{ 2}{*}{90} & 808348 & 625842 & 597508 & 629132   \\  
				& & (1.000) & (1.000) & (1.000) & (1.000) \\  \cline{2-6}
				& \multirow{ 2}{*}{100} & 819216 & 642255 & 616067 & 651371  \\  
				& & (1.000) & (1.000) & (1.000) & (1.000) \\  
				\bottomrule
		\end{tabular}}
	\caption{Estimated OBS and PCD (in parentheses) of ${\cal F}_B, {\cal IZR}$, and ${\cal IZE}$ with respect to three sets of thresholds and $n_0\in \{30, 40, \ldots, 100\}$. For ${\cal IZE}$, we set $n_0'=n_0-5$ and $\nu=0.8$.}
	\label{tab:InventoryResultsAdditional}
\end{table}

\end{document}